\documentclass[12pt]{article}

\usepackage{graphicx}%
\usepackage{multirow}%
\usepackage{amsmath,amssymb,amsfonts}%
\usepackage{amsthm}%
\usepackage{mathrsfs}%
\usepackage[title]{appendix}%
\usepackage{xcolor}%
\usepackage{textcomp}%
\usepackage{manyfoot}%
\usepackage{booktabs}%
\usepackage{algorithm}%
\usepackage{algorithmicx}%
\usepackage{algpseudocode}%
\usepackage{listings}%
\usepackage{mathtools}
\usepackage{dsfont}
\usepackage{bm}
\usepackage[autostyle]{csquotes}
\usepackage{hyperref}
\usepackage{subcaption}
\usepackage{float}
\usepackage{tikz}
\usepackage{tikz-cd} 
\usepackage{pifont,array}
\usepackage{chngcntr}
\usepackage{geometry}
\usepackage{comment}
\usepackage{longtable, multirow}
\usepackage{adjustbox, enumitem}
\usepackage{authblk}
\usepackage{stackengine}
\usetikzlibrary{decorations.pathmorphing}
\usetikzlibrary{decorations.markings}
\tikzset{snake it/.style={decorate, decoration=snake}}
\usetikzlibrary{shapes.arrows}
\usetikzlibrary{patterns}

\newcommand{\ds}{\displaystyle}

\newcommand{\op}[1]{\operatorname{#1}}
\newcommand{\im}{\mathrm{i}}

\newcommand{\bb}[1]{\mathbb{#1}}
\newcommand{\bbC}{\mathbb{C}}
\newcommand{\bbN}{\mathbb{N}}
\newcommand{\bbZ}{\mathbb{Z}}
\newcommand{\bbR}{\mathbb{R}}
\renewcommand{\cal}[1]{\mathcal{#1}}
\newcommand{\diff}{\mathrm{d}}
\newcommand{\id}{\mathds{1}}
\newcommand{\eul}{\mathrm{e}}
\newcommand{\ul}[1]{\underline{#1}}

\definecolor{blueBand}{RGB}{153,204,255} 
\definecolor{blueBandEdge}{RGB}{47,129,188} 
\definecolor{EdgeState}{RGB}{255,128,16} 
\definecolor{FermiLine}{RGB}{224,93,93} 
\definecolor{nicerGray}{RGB}{164,164,164} 
\definecolor{niceGreen}{RGB}{64,198,77}

\newcommand\brabar{\scalebox{.3}{(}\raisebox{-1.7pt}{--}\scalebox{.3}{)}} 
\tikzset{
	pics/.cd,
	vector in/.style args={#1/#2/#3}{
		code={
			\draw[#1] (0,0)  circle (#3);
			\draw[#2] (45:#3) -- (225:#3) (135:#3) -- (315:#3);
		}
	}
}
\tikzset{
	pics/vector out/.style args={#1/#2/#3}{
		code={
			\draw[#1] (0,0)  circle (#3);
			\fill[#2] (0,0)  circle (#3/4);
		}
	}
}

\theoremstyle{definition}
\newtheorem{definition}{Definition}[section]

\newtheorem{theorem}[definition]{Theorem}
\newtheorem{corollary}[definition]{Corollary}
\newtheorem{lemma}[definition]{Lemma}
\newtheorem{proposition}[definition]{Proposition}
\newtheorem{claim}[definition]{Claim}

\theoremstyle{remark}
\newtheorem{remark}[definition]{Remark}

\numberwithin{equation}{section}

\makeatletter
\newcommand{\myitem}[1]{%
	\item[#1]\protected@edef\@currentlabel{#1}%
}
\makeatother

\makeatletter
\newcommand{\thickhline}{%
	\noalign {\ifnum 0=`}\fi \hrule height 1pt
	\futurelet \reserved@a \@xhline
}
\newcolumntype{"}{@{\hskip\tabcolsep\vrule width 1pt\hskip\tabcolsep}}

\usepackage{stackengine}
\newcommand\xrowht[2][0]{\addstackgap[.5\dimexpr#2\relax]{\vphantom{#1}}}

\usetikzlibrary{arrows}

\definecolor{xdxdff}{rgb}{0.49019607843137253,0.49019607843137253,1}
\definecolor{ccqqqq}{rgb}{0.8,0,0}
\definecolor{ttttff}{rgb}{0.2,0.2,1}
\definecolor{uuuuuu}{rgb}{0.26666666666666666,0.26666666666666666,0.26666666666666666}
\definecolor{qqzzff}{rgb}{0,0.6,1}

\title{Boundary conditions and violations of bulk-edge correspondence in a hydrodynamic model}

\date{}

\author[1]{Gian Michele Graf \thanks{gmgraf@phys.ethz.ch}}

\author[1]{Alessandro Tarantola \thanks{ataranto@phys.ethz.ch}}

\affil[1]{Institute for Theoretical Physics, ETH Z\"urich}

\begin{document}
	
\maketitle

\abstract{Bulk-edge correspondence is a wide-ranging principle that applies to topological matter, as well as a precise result established in a large and growing number of cases. According to the principle, the distinctive topological properties of matter, thought of as extending indefinitely in space, are equivalently reflected in the excitations running along its boundary, when one is present. Indices encode those properties, and their values, when differing, are witness to a violation of that correspondence. We address such violations, as they are encountered in a hydrodynamic context. The model concerns a shallow layer of fluid in a rotating frame and provides a local description of waves propagating either across the oceans or along a coastline; it becomes topological when suitably modified at short distances. The edge index is sensitive to boundary conditions, as exemplified in earlier work, hence exhibiting a violation. Here we present classification of all (local, self-adjoint) boundary conditions and a parameterization of their manifold. They come in four families, distinguished in part by the degree of their underlying differential operators. Essentially, that degree counts the degrees of freedom of the hydrodynamic field that are constrained at the boundary by way of their normal derivatives. Generally, both the correspondence and its violation are typical. Within families though, the maximally possible amount of violation is increasing with its degree. Several indices of interest are charted for all boundary conditions. A single spectral mechanism for the onset of violations is furthermore identified. The role of a symmetry is investigated.}

\section{Introduction}\label{sec:Intro}

Topological considerations in physics have risen to a prominent role in the quantum domain, notably with the Quantum Hall Effect \cite{QHESeminal}. Later developments took place in the classical domain, and among them in hydrodynamics \cite{TDV20, DMV17, Tauber19}. The present work lies in that same context.

A sweeping concept in the realm of topological matter is bulk-edge correspondence (BEC); in fact to the extent that it is often taken for granted. In essence, it states that such systems admit two topological indices, independently defined, but coinciding in value. The first one, known as the bulk index, is associated in physical terms to the infinitely extended and gapped system, and in more mathematical ones to some compact manifold, such as the Brillouin zone. The second index, i.e.\,the edge index, is by contrast associated to a system now having a spatial boundary. The index arises by way of a count of edge states, chiral or otherwise, found at energies previously lying in the gap. See the selected references \cite{ProdanSB, Avila12, BK17}, and \cite{Shapiro20} for a pedagogical review.

The correspondence was first established in the context of the Integer Quantum Hall Effect \cite{Hatsugai93PRB, Hatsugai93PRL, Halperin82, Elbau02, SchulzBaldes99}, and soon extended to a variety of other models, possibly having special symmetry properties \cite{GP13, GS18, Elgart05, BSS23}, yet dealing with independent particles. Different techniques were used and, among them, scattering theory related to Levinson's theorem \cite{RSIII}. That approach \cite{GJT21} will play some role presently.

Models enjoying a continuous translational symmetry lack a compact momentum space, thus putting the bulk index in need of attention. Some extensions in this direction appeared in the literature and were obtained with a wide range of mathematical tools, ranging from bifurcation theory for PDEs \cite{Fefferman16} to homotopy of Fredholm operators \cite{Bal19} and more \cite{MS12, Drouot19}. Notice that BEC was shown to hold in various models with \enquote{soft} boundaries, namely thick interface regions between different insulators \cite{BalInterface, RT24}.

In contrast to the aforementioned results, that aimed at extending the validity of BEC, the present work attempts to pinpoint its limits and violations. We shall do so in the context of a hydrodynamic model, known as the rotating shallow water model (SWM), thus extending the analysis of special cases \cite{TDV19,TDV20,GJT21} to the generality that is proper to the model.

The rotating shallow water model describes an incompressible fluid lying on a rotating, flat bottom. It is derived from Euler's equations by linearizing and assuming small fluid height relative to the typical wavelength. It is used to describe certain oceanic and atmospheric layers on Earth \cite{GeophysicalFluids}, where the Coriolis force plays the same role as the magnetic field in the quantum Hall effect \cite{Froehlich95}. Remarkably, it explains the stability of the eastward-propagating Kelvin equatorial wave \cite{DMV2017} by recognizing its topological nature.

The model is defined by a system of partial differential equations (PDEs) that are formally equivalent to a Schr\"odinger equation, whence a similar treatment as used for quantum mechanical systems becomes feasible. The corresponding Hamiltonian is formally that of a topological insulator of symmetry class D according to the periodic table of topological matter \cite{TenfoldWay, KitaevTable}. In order to restore the desired compactness, as hinted at before, we shall consider a variant of the model, that is characterized by odd viscosity \cite{AvronOddVis}, and which provides a small-scale regularization. Indeed, the latter in turn allows for a compactification of the (otherwise unbounded) momentum space and for a proper topological definition of the bulk index.

If BEC were to hold in the form introduced above, this bulk invariant should predict the number of chiral modes appearing when a boundary, i.e.\,a coastline in the hydrodynamic setting, is introduced. This is however not always the case. As shown in \cite{GJT21}, the number of edge modes depends on the boundary conditions imposed on the coastline, thus invalidating BEC, whose claim does not depend on details of the boundary.

Ref.\,\cite{GJT21} calls for systematic and extensive investigation on violations of bulk-edge correspondence. This work considers and answers the following questions:
\begin{enumerate}[itemindent=1em]
	\myitem{Q1} What are the most general local self-adjoint boundary conditions for this model? And what can be said about BEC for them? \label{Q1}
	
	\myitem{Q2} What is the most natural edge index, i.e.\,the best way to count edge states merging with the (bulk) continuum? \label{Q2}
	
	\myitem{Q3} Are there other integers of interest? \label{Q3}
	
	\myitem{Q4} What relations, like BEC or otherwise, exist between them? Do they hold unconditionally? \label{Q4}
\end{enumerate}
Before detailing the structure of this paper, we briefly review some recent literature. Like our work, Refs.\,\cite{TT23, Perez24, BalYu24} focus on bulk-edge correspondence in the shallow water model. Ref.\,\cite{TT23} defines an edge index via spectral flow and \enquote{averaging} over a loop of boundary conditions. This index coincides with its bulk counterpart. Ref.\,\cite{Perez24} studies topological properties of the SWM on a spherical geometry. Ref.\,\cite{BalYu24} forgoes odd-viscous regularization, and considers a Coriolis force that varies in space, producing an interface where it crosses zero. Violations of BEC are displayed for certain discontinuous profiles of the Coriolis force. Other recent works concern the 2D Dirac Hamiltonian, which is closely related to the SWM when viewing both as spin models. An ultraviolet-regularized Dirac is shown to host violations of BEC in \cite{JT24}. In fact, boundary conditions are comprehensively classified, based on whether they do, or do not, lead to bulk-edge correspondence. Finally, \cite{Cornean23} notices that BEC is restored for magnetic Dirac operators, at least with infinite-mass boundary conditions . 

Broadly speaking, this paper consists of three parts. The first part, comprising Sects.\,\ref{sec:Setup} and \ref{sec:Results}, introduces the setup and main results, while glossing over a number of technical details. The second part, comprising Sects.\,\ref{sec:PreliminaryDetails}, \ref{sec:DetailedResults}, expands on the details omitted in the first one, and culminates in the statement of broader, more precise results. The third part, comprising Sects.\,\ref{sec:Derivation} and \ref{sec:Proofs}, contains proofs, preceded by a preliminary exposition of the techniques used.

More precisely, Sect.\,\ref{subsec:Bulk} introduces the infinitely extended (bulk) shallow water model, its Hamiltonian and properties of the latter. Among them: Spectrum, particle-hole symmetry (PHS), associated Bloch bundle. A bulk index $ C_+ $ is moreover defined through a compactification of momentum space. In Sect.\,\ref{subsec:BCs}, the edge model is introduced by restricting physical space to the upper half-plane. Local boundary conditions, and associated von Neumann unitaries, are defined and grouped into four qualitatively different families DD, ND, DN, NN. A special role is played by those that preserve PHS. In Sect.\,\ref{subsec:Integers}, we define our edge index as the number $M$ of proper ($P$) and improper ($I$) edge eigenvalues merging with the bulk band of interest. Further relevant integers are $E$, the number of asymptotically flat edge eigenvalues, and the winding $B$ of the boundary condition. The indices $ (P,I,E,B) $ collectively form an index vector $ \cal{V} $.

In Sect.\,\ref{sec:Results}, results are reported in terms of the index vector $ \cal{V} $. Its values are charted in the three families DD, ND, DN, and in the particle-hole symmetric subfamily of the fourth one, NN. Violations of bulk-edge correspondence are proven \enquote{typical}, and a single violation-inducing mechanism is identified.

Sect.\,\ref{subsec:BulkDetails} provides further details on the bulk Hamiltonian, focusing on its eigenstates and the compactification of momentum space that enables the definition of the bulk index. Sect.\,\ref{subsec:ScattTheo} reviews the scattering theory of shallow-water waves. The scattering amplitude $S$ is introduced, and its poles identify edge eigenvalues. In Sect.\,\ref{subsec:TopScatt}, we connect $S$ to the bulk and edge indices at once, making it the central tool for detecting bulk-edge correspondence or lack thereof. Sect.\,\ref{subsec:DetailsBCs} contains further details on self-adjoint boundary conditions, and a parametrization of the aforementioned families and associated unitaries.

Sect.\,\ref{sec:DetailedResults} deepens the results of Sect.\,\ref{sec:Results} by directly relating parameters of the boundary conditions to values of $\cal{V}$, and extends those same results by treating the full family NN (rather than its particle-hole symmetric subfamily).

Sects.\,\ref{subsec:ComputingP}, \ref{subsec:ParabolicStates}, \ref{subsec:FlatStates}, \ref{subsec:BoundaryWinding} detail how the evaluation of the integer $P$, $I$, $E$, $B$ is performed, respectively, given certain boundary conditions. Finally, Sect.\,\ref{sec:Proofs} presents the full proofs of the results of Secs.\,\ref{sec:Results}, \ref{sec:DetailedResults}.

The appendices either contain additional results that are not essential to the main discussion, or proofs which were omitted in the interest of readability.

\section{Setup} \label{sec:Setup}

In this section, we will provide an overview of the model and introduce the notions needed in order to state the results in Sect.\,\ref{sec:Results}. We shall discuss the bulk picture, including its geometric objects, to be followed by the edge picture, including a first discussion of the possible boundary conditions, that will turn out to be of four main types.

\subsection{Bulk model} \label{subsec:Bulk}

The rotating and odd-viscous shallow water model is a 2D hydrodynamical model describing a thin layer of incompressible fluid that lies on a flat, rotating bottom extending over the whole plane. Its physical relevance, its derivation by linearization, as well as its applications will not be treated in this work. The interested reader is referred to \cite{DMV2017, TDV19, TDV20, GJT21}.

The model is governed by the following system of linear PDEs
\begin{equation} \label{eq:SWM}
	\begin{cases}
		\partial_t \eta = - \partial_x u - \partial_y v \\
		\partial_t u = - \partial_x \eta - (f + \nu \Delta) v \\
		\partial_t v = - \partial_y \eta + (f + \nu \Delta) u \,,
	\end{cases}
\end{equation}
where $ (u,v): \bbR^2 \to \bbR^2 $ is the velocity field on the $ (x,y) $-plane, $ \eta: \bbR^2 \to \bbR $ the height of the fluid w.r.t.\,to the undisturbed fluid level; while the parameters $ f,\nu > 0 $ determine the Coriolis force $ f (-v, u) $ and the (odd) viscosity force $ \nu \Delta (-v,u) $, respectively. We denote the Laplacian by $ \Delta $ and assume $ 4 \nu f < 1 $ for simplicity, as will become clear later. 

The system of equations \eqref{eq:SWM} is formally equivalent to a Schr\"odinger equation $ \im \partial_t \psi = H \psi $, with wave function $ \psi = \psi (\ul{x}) $, ($\ul{x} \coloneqq (x,y)$) and Hamiltonian given as
\begin{equation} \label{eq:HamGen}
	\psi \coloneqq 
	\begin{pmatrix}
		\eta \\
		u \\
		v
	\end{pmatrix},
	\qquad H = 
	\begin{pmatrix}
		0 & - \im \partial_x & -\im \partial_y \\
		-\im \partial_x & 0 & -\im (f + \nu (\partial_x^2 + \partial_y^2) ) \\
		-\im \partial_y & \im (f + \nu (\partial_x^2 + \partial_y^2) ) & 0
	\end{pmatrix} \,.
\end{equation}
Let us at first regard $H$ as a formal differential operator. Being translation invariant, its normal modes are of the form
\begin{equation} \label{eq:NormalModes}
	\psi = \hat{\psi} (\ul{k}, \omega) \eul^{\im (\ul{k} \cdot \ul{x} - \omega t)} \,, \qquad \underline{k} \coloneqq (k_x, k_y) \in \mathbb{R}^2 \,, 
\end{equation}
with $ \hat{\psi} $ \textit{bulk} solutions of the time-independent eigenvalue problem (see Sect.\,\ref{subsec:BulkDetails})
\begin{equation} \label{eq:BulkHam}
	H \hat{\psi} = \omega \hat{\psi} \,, \qquad \hat{\psi} = 
	\begin{pmatrix}
		\hat{\eta} \\
		\hat{u} \\
		\hat{v}
	\end{pmatrix} \,, \qquad H(\ul{k}) = 
	\begin{pmatrix}
		0 & k_x & k_y \\
		k_x & 0 & -\im (f - \nu k^2 ) \\
		k_y & \im (f - \nu k^2 ) & 0
	\end{pmatrix} \,,
\end{equation}
where $ k^2 \coloneqq | \underline{k} |^2 = k_x^2 + k_y^2 $. The differential operator $H$ is realized as a self-adjoint operator on the Hilbert space $ \mathcal{H} = L^2 (\mathbb{R}^2)^{\otimes 3} $, with domain given in terms of Sobolev spaces as $ \cal{D} (H) = H^1(\mathbb{R}^2) \oplus H^2(\mathbb{R}^2) \oplus H^2 (\mathbb{R}^2) $.

The spectrum of $H$ is purely essential \cite{RSI} and consists of three bands, labeled $ \sigma = \pm, 0 $. They arise from the discrete spectrum of $ H (\ul{k}) \,, (\ul{k} \in \bbR^2) $, consisting of eigenvalues $ \omega_- (\ul{k}) < \omega_0 (\ul{k}) < \omega_+ (\ul{k}) $ (cf.\,Fig.\,\ref{fig:bulk}):
\begin{equation} \label{eq:Omega+}
	\omega_{\pm} (\ul{k}) = \pm \sqrt{k^2 + \left( f - \nu k^2 \right)^2} \,, \hspace{10pt} \omega_0 (\ul{k}) = 0 \,.
\end{equation}
The bands are separated by gaps of size $ f $ at $ \underline{k} = (0,0) $. Let moreover 
\begin{align}
	\omega^+ (k_x) \coloneqq \inf_{k_y} \omega_+ (k_x,k_y) = \omega_+ (k_x,0) = \sqrt{k_x^2 + \left( f - \nu k_x^2 \right)^2} \label{eq:BandRim} 
\end{align}
denote the lower rim of the upper band.

Momentum space $ \bbR^2 \ni \ul{k} $ is unbounded, reflecting that translations give rise to a continuous symmetry group. As we shall see in Sect.\,\ref{subsec:BulkDetails}, the regularization provided by $ \nu $ ensures the eigenspaces of $ \omega_\sigma (\ul{k}) $ converge as $ \ul{k} \to \infty $, thus allowing for a description based on the $1$-point compactification $ S^2 $ of $ \bbR^2 $.

We shall call \textit{Bloch bundle} the trivial vector bundle $ E = S^2 \times \bbC^3 $. It decomposes as
\begin{equation} \label{eq:BlochBundle}
	E = \bigoplus_{\sigma = \pm, 0} E_\sigma \,,
\end{equation}
where 
\begin{equation} \label{eq:BandBundle}
	E_\sigma = \{ (\underline{k},\hat{\psi}) \vert \underline{k} \in S^2, \hat{\psi} \in \operatorname{ran} (P_\sigma (\underline{k})) \} \,,
\end{equation}
is the eigenbundle determined by the eigenprojections $ P_\sigma (\ul{k}) $ of the eigenvalues $ \omega_\sigma (\ul{k}) $, i.e.\,, $ E_{\sigma, \ul{k}} = \op{ran} (P_\sigma (\ul{k})) $, cf.\,\cite{GJT21} for details. The Chern numbers $ \op{ch} ( P_{\pm} ) \equiv C_{\pm} $, $ \op{ch} ( P_0 ) \equiv C_0 $ associated to the bands are proper topological invariants, in view of the compactness of $ S^2 $. They amount to
\begin{equation}
	C_{\pm} = \pm 2 \,, \qquad C_0 = 0 \,, \label{eq:BulkIndices}
\end{equation}
as shown in \cite{GJT21}, Prop. 2.1.

The Hamiltonian $H$ enjoys an antiunitary symmetry, formally akin to (even) particle-hole (pseudo-) symmetry (PHS) \cite{ZirnbauerPHS},
\begin{gather}
	\Xi H(\ul{k}) \Xi^{-1} = - H (- \ul{k}) \,, \nonumber \\
	\Xi = \id_3 \cdot K \,, \label{eq:PHS}
\end{gather} 
with $K$ denoting complex conjugation. Clearly, $ \Xi^2 = \id $. The Hamiltonian enjoys no further symmetry among those contemplated in the AZK table \cite{TenfoldWay, KitaevTable}. It thus sits in class $D$.

In the following, we will only focus on the upper band, the middle one being trivial, and $ E_- = \Xi E_+ $. We will refer to 
\begin{equation}
	C_+ = \op{ch} (P_+) \label{eq:BulkIndex}
\end{equation}
as the \textit{bulk index}. Its value, $ +2 $, will be compared to that of other indices, that arise by restriction to a half-plane.

\subsection{Edge picture and boundary conditions} \label{subsec:BCs}

Let us now introduce the edge picture. Details that are not needed for the statement of the results will be returned upon in Sects.\,\ref{subsec:ScattTheo}, \ref{subsec:DetailsBCs}. \vspace{1\baselineskip}

A boundary is introduced by restricting the physical space to the upper half-plane
\begin{equation}
	\bbR^2_+ \coloneqq \{ (x,y) \in \mathbb{R}^2 \vert y > 0 \} = \mathbb{R} \times \mathbb{R}_+ \,. \label{eq:HalfPlane}
\end{equation}
Since translation invariance is retained in the $x$-direction only, normal modes of the formal differential operator \eqref{eq:HamGen} are of the form
\begin{equation}
	\psi = \tilde{\psi}(y;k_x, \omega) \eul^{\im (k_x x - \omega t)} \,, \label{eq:PsiTilde}
\end{equation}
where $ \tilde{\psi} $ is in turn a normal mode of the formal differential operator
\begin{equation} \label{eq:EdgeHamPrecursor}
	\setlength{\tabcolsep}{12pt}
	H (k_x) =
	\begin{pmatrix}
		0 & k_x & - \im \partial_y \\
		k_x & 0 & - \im ( f + \nu (\partial_y^2 - k_x^2) ) \\
		- \im \partial_y & \im ( f + \nu (\partial_y^2 -k_x^2) ) & 0 
	\end{pmatrix} \,, \qquad (k_x \in \bbR) \,. 
\end{equation}
These modes are more general than the modes \eqref{eq:NormalModes} in that they can be evanescent, i.e.\,decaying transversally.

The formal differential operator \eqref{eq:EdgeHamPrecursor} is likewise promoted, as in the previous section, to a self-adjoint operator $ H^\# $ on the Hilbert space $ \cal{H}^{\#} = L^2 (\bbR^2_+)^{\otimes 3} $, by way of specifying suitable boundary conditions that define its domain $ \cal{D} (H^\#) $.

Not all boundary conditions correspond to a self-adjoint realization of $ H^{\#} $. Von Neumann's theory of self-adjoint extensions builds on the deficiency spaces at $ z = \pm \im $ of the (closed) operator realizing the differential operator $H$ with minimal domain. It then expresses the extensions in terms of unitary maps between the two spaces; cf.\,(\cite{RSII}, Thm.\,X.2).

We shall determine all self-adjoint boundary conditions that are
\begin{itemize}
	\item[(i)] local,
	
	\item[(ii)] invariant under translation along the boundary, and
	
	\item[(iii)] involve derivatives $ \partial_x $ of the wave function of order 1 at most.
\end{itemize}
Let us first focus on conditions (i, ii). As a result, the deficiency spaces are those of a minimal realization of $ H(k_x) $, cf.\,\eqref{eq:EdgeHamPrecursor}. The defect indices of the latter are finite, and in fact $ n_\pm = 2 $. The extensions of the formal differential operator \eqref{eq:EdgeHamPrecursor} are thus in some correspondence with unitary matrices of order 2, $ U(k_x) \in U(2) $ (cf.\,Sect.\,\ref{subsec:DetailsBCs} and App.\,\ref{app:SelfAdj}). The same extensions can also be expressed in terms of two boundary conditions involving $ \tilde{\psi} (y) $ and $ \tilde{\psi}' (y) $ $ (' = \partial_y) $ at $ y = 0 $
\begin{equation} \label{eq:BoundaryCondition}
	A(k_x) \Psi = 0 \,, \qquad \Psi \coloneqq 
	\begin{pmatrix}
		\tilde{\psi} (0) \\
		\tilde{\psi}' (0)
	\end{pmatrix} \in \bbC^6 \,,
\end{equation}
whence $ A(k_x) $ is a $ 2 \times 6 $ matrix that, far from being arbitrary, satisfies conditions (to be discussed in Sect.\,\ref{subsec:DetailsBCs}) arising from the requirement of self-adjointness.

Any such $ A(k_x) $ determines $ U(k_x) $, but not conversely. Rather, $ U(k_x) $ determines the orbit $ [A(k_x)] $ of $ A(k_x) $ under multiplication from the left by $ \op{GL} (2,\bbC) $,
\begin{equation}
	\tilde{A} (k_x) \sim A(k_x) \ \Leftrightarrow \ \tilde{A} (k_x) = G(k_x) A(k_x) \,, \qquad G(k_x) \in \op{GL} (2,\bbC) \,. \label{eq:Cosets}
\end{equation}
This reflects the fact that the two boundary conditions are unaltered when replaced by linearly independent linear combinations.

By including requirement (iii) above, the boundary conditions are of the general form
\begin{equation}
	A_0 \psi + A_x \partial_x \psi + A_y \partial_y \psi = 0 \,, \qquad (A_j \in \op{Mat}_{2 \times 3} (\bbC) \,, \ j = 0,x,y ) \,, \label{eq:LocalBC}
\end{equation}
with
\begin{equation}
	A(k_x) = (A_0 + \im k_x A_x, A_y) \,. \label{eq:A0AxAy}
\end{equation}
Eq.\,\eqref{eq:LocalBC} then implies that the dependence $ k_x \mapsto A(k_x) $ in \eqref{eq:BoundaryCondition} is affine. Alternatively, $ k_x \mapsto U(k_x) $ will turn out to be a fractional linear transformation with coefficients in $ \op{Mat}_{2} (\bbC) $. Moreover, a linear recombination of the rows of Eq.\,\eqref{eq:LocalBC} amounts to $ G(k_x) $, cf.\,\eqref{eq:Cosets}, being a constant matrix.
\begin{remark} \label{rem:Orbits}
	Both the context of Eq.\,\eqref{eq:Cosets}, which is pointwise in $ k_x $, and that of \eqref{eq:LocalBC}, which is global, induce a notion of orbit, $ [A] $. They are the orbits under local vs global gauge transformations, i.e.\,$ G(k_x) $ vs $G$, acting by left multiplication. The former orbits are of course larger, and decompose into the latter. We take the point of view that only $A$'s arising from \eqref{eq:LocalBC} are to be considered, but any two of them shall be considered equivalent in the sense of \eqref{eq:Cosets}. This reflects the fact that this notion characterizes a same half-plane Hamiltonian $ H^\# $.
\end{remark}
The complete, intrinsic characterization of the objects $ A,U,G $ so obtained will be given in Sec.\,\ref{subsec:DetailsBCs}. The following operative definition will be useful. The cross-references to App.\,\ref{app:SelfAdj} fill the details that were left open up to now, but can be skipped still.
\begin{definition}[\textit{Boundary condition and von Neumann unitary}]
	\label{def:OperativeBC}
	An affine function
	\begin{equation}
		A : \bbR \longrightarrow \op{Mat}_{2 \times 6} (\bbC) \,, \qquad k_x \longmapsto A(k_x)
	\end{equation}
	satisfying (\ref{eq:RankCondition}, \ref{eq:CondSelfAdjKx}) is called a \emph{boundary condition}. Likewise, the function
	\begin{equation}
		U : \bbR \longrightarrow U(2) \,, \qquad k_x \longmapsto U(k_x)
	\end{equation}
	obtained from $ A $ by (\ref{eq:AulA}, \ref{eq:JuxtaposeA1A2}, \ref{eq:UFormula}) is called the \emph{von Neumann unitary} of the boundary condition.
\end{definition}
\begin{remark}
	Conditions (\ref{eq:RankCondition}, \ref{eq:CondSelfAdjKx}) are equivalent to self-adjointness of $ H^\# $ in the context (i-iii).
\end{remark}
We partition the set of self-adjoint boundary condition into the following four families.
\begin{definition}[\textit{Families of boundary conditions}] \label{def:Families}
	A self-adjoint boundary condition $ k_x \mapsto A(k_x) $, cf.\,Def.\,\ref{def:OperativeBC}, with $ A (k_x) $ as in \eqref{eq:A0AxAy}, belongs to family
	\begin{itemize}[align=parleft, itemindent=2cm, labelsep=1cm]
		\item[DD:] iff $ \op{rk} A_y = 0 $;
		
		\item[ND/DN:] iff $ \op{rk} A_y = 1 $;
		
		\item[NN:] iff $ \op{rk} A_y = 2 $.
	\end{itemize}
\end{definition}
Clearly, the conditions are compatible with \eqref{eq:Cosets}. The middle case will be subdivided into two (overlapping) families in Lm.\,\ref{lem:Dictionary}. As will also be explained in Sect.\,\ref{subsec:DetailsBCs}, the two families are related by a substitution, by way of which results can be carried over between them. As they remain structurally unaffected, only the case ND will be pursued.

The labeling by symbols D and N comes about as follows. We distinguish between \textit{generalized Neumann} (N, a.k.a.\,Robin or mixed) and \textit{Dirichlet} (D) boundary conditions (BC), in that they do, respectively do not, involve $ (\partial_y f) (0) $ for some linear combination $f$ of the field components $ \eta, u, v $.

The two occurring boundary conditions will involve two linear combinations $ f_1, f_2 $ of the components $ u,v $ only. As explained, the boundary conditions can be replaced by linearly independent linear combinations thereof. While making use of this freedom, in order to get rid of boundary conditions of type N to the extent possible, the four families of Def.\,\ref{def:Families} correspond to
\begin{itemize}[align=parleft, itemindent=2cm, labelsep=1cm]
	\item[DD:] No derivatives $ \partial_y u $, $ \partial_y v $ appear in the BC;
	
	\item[ND/DN:] Only one linear combination 
	\begin{equation}
		a \partial_y u + b \partial_y v \,, \quad \big( (a,b) \neq 0 \big)
	\end{equation} 
	\hspace{2cm}appears in the BC;
	
	\item[NN:] Both BCs must contain linear combinations of $ \partial_y u $, $ \partial_y v $.
\end{itemize}

Let us provide examples of boundary conditions \eqref{eq:BoundaryCondition}, and in fact of types DD and DN respectively.
\begin{itemize}
	\item Dirichlet boundary conditions
	\begin{equation} \label{eq:Dirichlet}
		\begin{cases}
			u(y) \rvert_{y=0} = 0 \\
			v(y) \rvert_{y=0} = 0
		\end{cases}
		\,,  \qquad A_D (k_x) = 
		\begin{pmatrix}
			0 & 1 & 0 & 0 & 0 & 0 \\
			0 & 0 & 1 & 0 & 0 & 0 
		\end{pmatrix} \,.
	\end{equation}
	
	\item A family of boundary conditions parametrized by $ a \in \bbR $:
	\begin{equation} \label{eq:GrafBC}
		\begin{cases}
			v(y) \rvert_{y=0} = 0 & \\
			\partial_x u(y) + a \partial_y v(y) \rvert_{y=0} = 0
		\end{cases}
		\,, \qquad A (k_x) = 
		\begin{pmatrix}
			0 & 0 & 1 & 0 & 0 & 0 \\
			0 & \im k_x & 0 & 0 & 0 & a 
		\end{pmatrix} \,.
	\end{equation}
	Clearly, $ \partial_x $ acts as $ \im k_x $ on states \eqref{eq:PsiTilde}, whence the expression for $ A (k_x) $. See Figure \ref{fig:bulk}, right panel, for the spectrum of the corresponding Hamiltonian when $ a = 4 $.
\end{itemize}
\begin{figure}[hbt]
	\centering
	\includegraphics[width=0.7\linewidth]{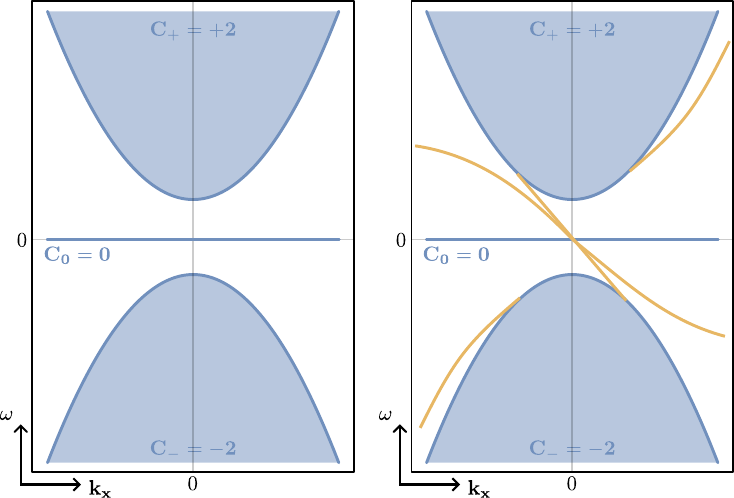}
	\caption{Left panel: Bulk spectrum, projected along the $ k_y $-direction, and Chern numbers of each band, cf.\,\eqref{eq:BulkIndices}. Right panel: Spectrum of $ H^\# $ with boundary conditions as in Eq.\,\eqref{eq:GrafBC} and $a=4$.}
	\label{fig:bulk}
\end{figure}
The above conditions (i-iii) will apply throughout the article. 
\begin{remark} \label{rem:LossSA}
	It should be moreover remarked that a given self-adjoint realization $ H^\# $ on the half-plane induces operators $ H^\# (k_x), \, (k_x \in \bbR) $ on the half-line, that are self-adjoint at all $ k_x $ except possibly at isolated points.
	
	Boundary conditions where such a failure occurs are themselves exceptional, and excluded from the further discussion, as a rule. 
\end{remark}
\begin{remark} \label{rem:TranslationsKx}
	Writing $ H^\# $ fiberwise, as in the previous remark, makes another fact apparent: The global half-plane Hamiltonian $ H^\# $ is left unchanged by a reparametrization $ k_x \mapsto \tilde{k}_x \equiv k_x + \tau $, $ (\tau \in \bbR) $. Indeed, it amounts to a relabeling of the fibers $ H^\# (k_x) $, cf.\,Rem.\,\ref{rem:LossSA}. Under this shift, a local boundary condition $ A(k_x) $, cf.\,\eqref{eq:A0AxAy}, transforms as 
	\begin{equation}
		A (k_x) \mapsto \tilde{A} (k_x) := (\tilde{A}_0 + \im k_x \tilde{A}_x \,, \ \tilde{A}_y) \,, \label{eq:KxTranslationBC}
	\end{equation}
	where
	\begin{equation}
		\tilde{A}_0 = A_0 + \im \tau A_x \,, \qquad \tilde{A}_x = A_x \,, \qquad \tilde{A}_y = A_y \,.
	\end{equation}
	From the point of view of global properties of $ H^\# $, two boundary conditions $ A, \tilde{A} $ related by \eqref{eq:KxTranslationBC} shall thus effectively be identified. More properly, in view of Rem.\,\ref{rem:Orbits}, two distinct orbits $ [A], [\tilde{A}] $ may still produce the same global operator $ H^\# $. 
\end{remark}
There is a further condition worth considering. In view of the fact that the bulk model enjoys particle-hole symmetry, it pays to single out those boundary conditions that preserve it. We thus define:
\begin{definition} \label{def:PHS}
	\begin{itemize}
		\item[(i)] (\textit{Particle-hole conjugation.}) Let $ H^\# $ be a self-adjoint realization, as described above and having domain $ \cal{D} (H^\#) $.
		
		We define its particle-hole conjugate by
		\begin{equation}
			H^\#_\Xi \coloneqq \Xi H^\# \Xi^{-1} \,,
		\end{equation} 
		where $ \Xi $ is as in \eqref{eq:PHS}. In particular, $ \psi \in \cal{D} (H^\#_\Xi) $ iff $ \Xi \psi \in \cal{D} (H^\#) $.
		
		\item[(ii)] (\textit{Particle-hole symmetry.}) $ H^\# $ is particle-hole symmetric if $ H^\# = H^\#_\Xi $. In particular
		\begin{equation}
			\psi \in \cal{D} (H^\#) \ \Longleftrightarrow \ \Xi \psi \in \cal{D} (H^\#) \,.
		\end{equation}
	\end{itemize}
\end{definition}
Let $ k_x \mapsto A(k_x) $ be the boundary condition, cf.\eqref{eq:LocalBC}, defining $ H^\# $. Then, $ \overline{A (-k_x)} $ defines $ H^\#_\Xi $, because $ A(k_x) \Psi (k_x) = 0 $ implies $ \overline{A (-k_x)} \ \overline{\Psi (k_x)} = 0 $, cf.\,(\ref{eq:BoundaryCondition}, \ref{eq:A0AxAy}).

Hence $ H^\# $ is particle-hole symmetric iff
\begin{equation}
	U(k_x) = \overline{U (-k_x)} \,. \label{eq:USymmetry}
\end{equation}
Indeed, symmetry states that the two maps, $ A(k_x) $ and $ \overline{A (-k_x)} $, determine the same boundary condition for $ H^\# (k_x) $ at any $ k_x $, namely that their orbits are the same point-wise in $ k_x $, i.e.\,, \eqref{eq:USymmetry}.

We will specialize to the symmetric case on and off, for each family of boundary conditions separately.

\subsection{Edge index and associated integers} \label{subsec:Integers}

As mentioned earlier, the bulk model enjoys \textit{continuous} translational symmetry and its momentum space is thus unbounded, at least at first. That sets the model apart from most models in condensed matter physics, where translations are discrete, being those of a lattice, and momentum space compact as a result, and typically a torus. In Sect.\,\ref{subsec:Bulk} that issue was addressed and its remedy indicated, namely compactification of momentum space.

A similar issue shows up in the edge problem, i.e.\,when restricting the model to a half-space. Then only the longitudinal momentum $ k_x $ remains a good quantum number, no matter whether translations are continuous or discrete, at least if the boundary conditions share that symmetry. In the discrete case, however, the range of $ k_x $ is compact to begin with and the spectrum is of the form seen in Fig.\,\ref{fig:mergers} (first panel): It consists of bands (shaded regions) and curves (in yellow). The energy ranges of the former give rise to the bulk spectrum, collectively in $k_x$; not so the latter, which represent the dispersion relations of states propagating along the edge. When considered at fixed $ k_x $, edge states correspond to bound states of $ H^\# (k_x) $. 
\begin{figure}[hbt]
	\centering
	\includegraphics[width=0.9\linewidth]{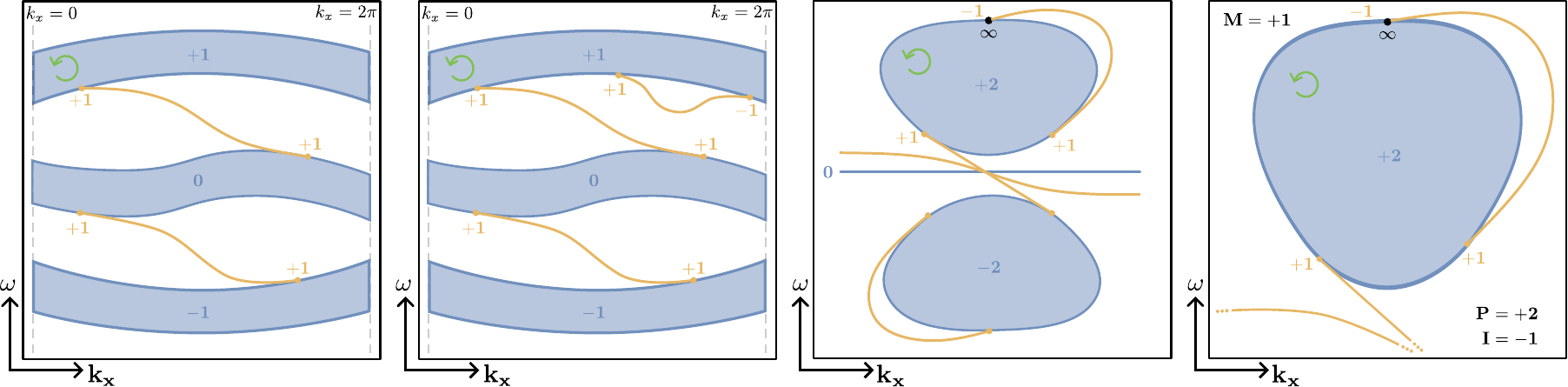}
	\caption{Panel 1: Typical edge spectrum of a model with discrete translation symmetry. The green circle represents the natural orientation for the boundary of the bands, inherited from the $ (k_x, \omega) $-plane. Panel 2: Same spectrum with addition of a \enquote{topologically trivial} edge state. Panel~3: Compactified spectrum of the SWM, with boundary conditions as in \eqref{eq:GrafBC} with $ a=4 $. Panel~4: Proper and improper mergers with the top band for the spectrum of Panel 3. Bulk-edge correspondence, as by Eq.\,\eqref{eq:BEC}, \textit{does not hold}: $ \op{ch} (P) = 2 $, $ N = 1 $.}
	\label{fig:mergers}
\end{figure}

Bulk-edge correspondence \cite{Hatsugai93PRL,Hatsugai93PRB} relates the Chern number $ \op{ch} (P) $ of a band, cf.\,\eqref{eq:BandBundle}, with the mergers of edge states with that band:
\begin{equation}
	\op{ch} (P) = N \,, \label{eq:BEC}
\end{equation}
where $N$ is the signed count of states leaving (+), resp.\,joining (-) the bulk spectrum, as meant w.r.t.\,the natural orientation of its boundary. That orientation, to be clear, is induced from the one of the plane $ (k_x,\omega) $, see panels of Fig.\,\ref{fig:mergers} in green, on the boundaries of the shaded regions. Equivalently, if $ k_x $ is taken to run in positive direction, it parametrizes the upper and lower boundary of the band oppositely, as compared to its natural orientation. With that convention,
\begin{equation}
	N = N_+ - N_- \,,
\end{equation}
where $ N_+ $ counts both the edge states leaving the boundary at the bottom and those joining at the top; similarly for $ N_- $.

The second panel of Fig.\,\ref{fig:mergers} shows a cancellation between a state leaving and then joining a band. Likewise, such a cancellation should be preserved in the SWM. This, and more, can be achieved by a certain compactification of the $ (k_x, \omega) $-plane, illustrated in panels three and four of Fig.\,\ref{fig:mergers}. In contrast to Fig.\,\ref{fig:bulk}, the infinity inside the bulk spectrum is compactified to a single point (the point at infinity $\infty$), whereas directions moving to infinity inside the gap, either at $ k_x \to + \infty $ or $ - \infty $, give rise to a continuum of points at infinity, labeled by $ 0 < h < + \infty $, if they attain a horizontal asymptote of that height. In lack of a horizontal asymptote, convergence is to $ \infty $.

So prepared, we may address edge eigenvalues connected to the top band in relation to their behavior in the finite or at infinity. We say:
\begin{itemize}
	\item An edge eigenvalue merging with that band at finite momentum $ k_x $ has a \textit{proper merger};
	
	\item An edge eigenvalue diverging for $ k_x \to \infty $ has an \textit{improper merger};
	
	\item An edge eigenvalue having a horizontal asymptote at $ k_x \to + \infty $ or $ - \infty $ has an \textit{escape}.
\end{itemize}
Viewed through the compactification, proper and improper mergers look alike, and shall simply be called \textit{mergers}. We remark that in \cite{GJT21} a somewhat different terminology was used.
\begin{definition}[\textit{Number of proper and improper mergers, number of escapes}]
	We denote by $P,I,E$ the signed number of proper mergers, improper mergers and escapes, respectively. Let moreover
	\begin{equation}
		M \coloneqq P + I  \label{eq:NumMergers}
	\end{equation}
	denote the signed number of \textit{mergers}. \label{def:Integers}
\end{definition} 

Let us introduce one last integer of interest. We recall that any (self-adjoint) boundary condition gives rise to a map $ U: \bbR \to U(2) $, $ k_x \mapsto U(k_x) $ with limits $ U(+\infty) = U(-\infty) $. That prompts the following definition.
\begin{definition}[\textit{Boundary winding}]
	For $U$ as above, we call $ B(U) $ its winding number, i.e.
	\begin{equation}
		B (U) = + \frac{1}{2 \pi \im} \int_{-\infty}^{+ \infty} \diff k_x \op{tr} \left( U^{-1} (k_x) \partial_{k_x} U(k_x) \right) \,. \label{eq:BoundWinding}
	\end{equation} \label{def:BoundWinding}
\end{definition}
In the next section, our main results are reported, expressed in terms of the integers $ P,I,E,B $ only. Given our choice of edge index, we set:
\begin{definition} \label{def:BEC}
	Bulk-edge correspondence (BEC) holds true if 
	\begin{equation}
		M = C_+ (= +2) \,. \label{eq:BECNumbers}
	\end{equation}
\end{definition}
Accordingly, violations occur for $ M \neq 2 $.
\begin{remark} \label{rem:IntVector}
	Given that the half-space model is determined by the boundary condition, all four integers $ P,I,E,B $ are functions of just $U$, i.e.\,, $ \cal{V} = \cal{V} (U) $, with
	\begin{equation}
		\cal{V} \coloneqq \left( P,I,E,B \right) \,. \label{eq:IntVector}
	\end{equation}
	Moreover, $\cal{V} $ is a property of $ (H^\# (k_x))_{k_x \in \bbR} $, and thus invariant under translations of $ k_x $, cf.\,Rem.\,\ref{rem:TranslationsKx}. Equivalently, $ \cal{V} $ is a function of $U$, upon identification of the unitaries that differ by translation of $ k_x $.
\end{remark}

\section{Overview of results} \label{sec:Results}

This section reports the main results of this work and addresses the questions \ref{Q1}--\ref{Q4} raised in the introduction. At first, \textit{spectral events} that occur while continuously changing the boundary condition will be defined and listed. They are linked to the integers $ P,I,E,B $ in Defs.\,\ref{def:Integers}, \ref{def:BoundWinding}, and more properly to their changes. As such, they matter for Question \ref{Q3}. The number $M$ of mergers, cf.\,\eqref{eq:NumMergers}, turns out to be more stable than others, as boundary conditions change; indeed, there is only one kind of change that prevents its invariance. This addresses \ref{Q2}. The more technical Question \ref{Q1} is addressed for its larger part in Sects.\,\ref{subsec:DetailsBCs}, \ref{sec:DetailedResults} and App.\,\ref{app:SelfAdj}. Question \ref{Q4} will be dealt with by charting the values of $ P,I,E,B $ in a phase diagram of boundary conditions. Violations of BEC are typical in all families except for DD, where they do not appear at all. \vspace{1\baselineskip}   

The first result connects changes in the value of our integers to \textit{spectral events}, while exploring the manifold of self-adjoint boundary conditions.
\begin{figure}[h]
	\centering
	\includegraphics[width=0.9\linewidth]{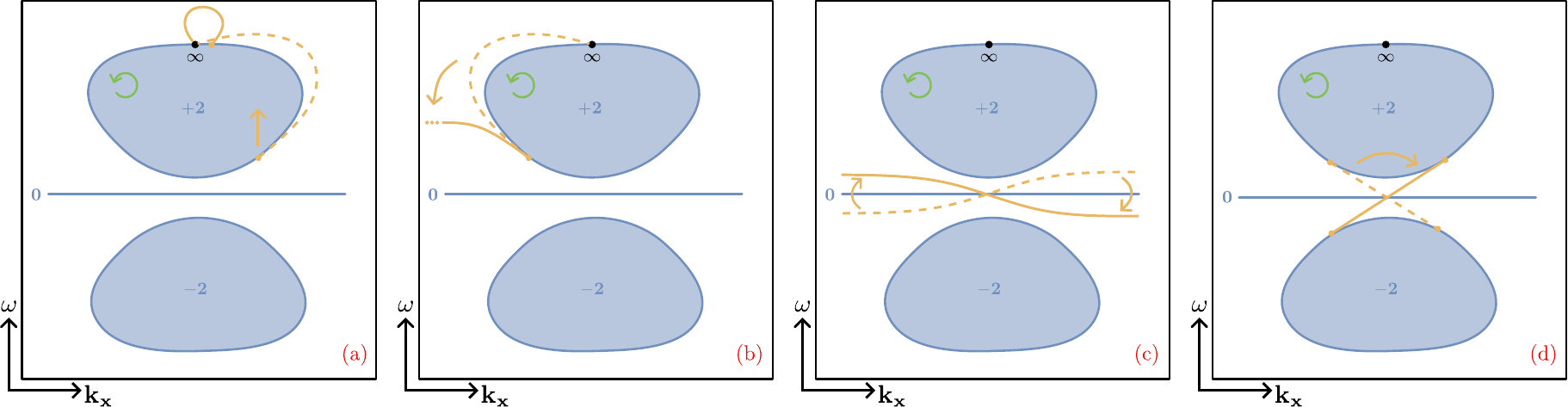}
	\caption{The four spectral events (a--d) described in Lemma \ref{thm:transitions}.}
	\label{fig:Transitions}
\end{figure}
\begin{theorem} \label{thm:transitions}
	Let the manifold of self-adjoint boundary conditions be partitioned into subsets, or regions, $R$ according to the values of $ \cal{V} $, cf.\,\eqref{eq:IntVector}. They are \textit{regular}, in the sense that they satisfy $ \partial \mathring{R} = \partial R $. Any transition between regions can be achieved by composition of elementary ones, which occur across boundaries between just two regions, possibly upon arbitrarily small deformation. Transitions are described as follows in terms of \textit{spectral events} (see Fig.\,\ref{fig:Transitions}):
	\begin{itemize}
		\item[(a)] A merging point reaches or leaves $ \infty $; as it does so, $ P $ and $ I $ change at once, but oppositely: $ \Delta P = - \Delta I = \pm 1 $ ($ \Delta M = 0 $);
		
		\item[(b)] An asymptotically parabolic states turns flat or vice-versa; $ I $ and $ E $ change at once, but oppositely: $ \Delta I = - \Delta E = \pm 1 $ ($ \Delta M = \pm 1 $);
		
		\item[(c)] A one-sided flat state disappears from the gap by crossing the zero-energy band, while another one appears at the opposite end; $E$ changes, $ | \Delta E | = +2 $.
	\end{itemize}
	Within the above manifold, we single out the exceptional ones, where self-adjointness fails at a point $ k_x $ in the sense of Rem.\,\ref{rem:LossSA}, and consider an elementary transition across them. Then
	\begin{itemize}
		\item[(d)] On approach, an edge eigenvalue turns vertical at the point $ k_x $ of failure; $B$ is ill-defined, yet it does not change.
	\end{itemize}
	Integers not mentioned in relation to a given event are left unchanged.
\end{theorem}
\begin{remark} \label{rem:WeirdETrans}
	The theorem does not rule out further spectral events, not contemplated above, that do not correspond to a transition between level sets of $ \cal{V} $. An example is a simultaneous appearance (or disappearance) of two flat states at opposite ends. Hence $ \Delta E = 0 $.
\end{remark}
\begin{remark} \label{rem:NonElementary}
	Within the restricted manifold of boundary conditions corresponding to the particle-hole symmetric case, cf.\,Def.\,\ref{def:PHS}, transitions between regions exhibit more than one spectral event, thus failing to be elementary. More than that, the restoration of that property by slight deformation may require that the symmetry is forgone. However, the property $ \partial \mathring{R} = \partial R $ of the regions survives the restriction.
\end{remark}
An answer to Question \ref{Q3} follows as a corollary.
\begin{corollary} \label{cor:ViolationMechanism}
	Let $ U_1,  U_2 $ be two distinct boundary conditions, and let $ M_j = M (U_j) $, $ (j = 1,2) $ be the corresponding indices \eqref{eq:NumMergers}. If $ M_1 \neq M_2 $, then any path joining $ U_1 $ to $ U_2 $ must undergo at least one transition of type (b) (cf.\,Thm.\,\ref{thm:transitions}).
\end{corollary}
In plain language, parabolic-to-flat transitions are the only spectral events capable of altering $M$. Regions on the BC manifold where BEC is respected or violated must thus be separated by \enquote{surfaces} of type (b) transitions.

Further details on existence and prevalence of BEC are gained by separately charting $ \cal{V} (U) $ for all $ U $'s within each of the families DD, ND, DN, NN, starting with DD. In this specific case, no charting is actually needed, as illustrated by the following Proposition.
\begin{proposition}[Index charts, DD case] \label{prop:Map1}
	Any boundary condition $ A $ (as by Def.\,\ref{def:OperativeBC}) within the family DD belongs to
	\begin{equation}
		[ A_D ] \,,
	\end{equation}
	where $ A_D $ are the Dirichlet boundary conditions, cf.\,\eqref{eq:Dirichlet}, and $ [A_D] $ their orbit, cf.\,\eqref{eq:Cosets}. As a result, the family corresponds to a single von Neumann unitary $U$. 
	
	The value of $ \cal{V} (U) $, cf.\,Eq.\,\eqref{eq:IntVector}, for that unitary is
	\begin{equation}
		\cal{V} (U) = (2,0,-1,0) \,,
	\end{equation}
	whence $ M = +2 = C_+ $ and BEC holds, cf.\,(\ref{eq:NumMergers},\ref{eq:BECNumbers},\ref{eq:IntVector}).
\end{proposition}
\begin{figure}[hbt]
	\centering
	\includegraphics[width=0.9\linewidth]{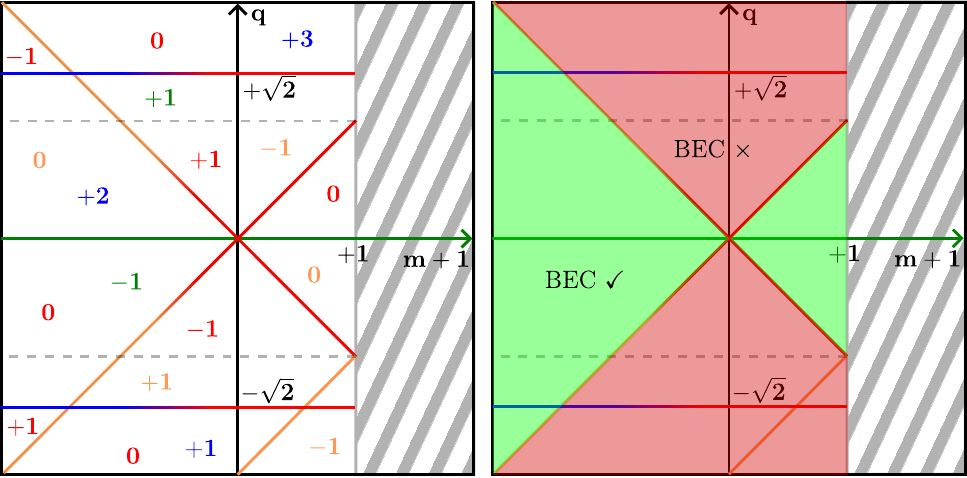}
	\caption{Map of integers for family ND. Drawn in blue, red, orange, green are $ P,I,E,B $, respectively. To each integer are associated transition lines. Lines of a specific colour pertain to one index only, like the orange $ \{ q = m-1 \} $ line where $ E $ alone varies. By contrast, if two indices change at once, the corresponding line interpolates between the two colors pertaining each index. The right panel highlights were BEC holds or fails.}
	\label{fig:IndicesII}
\end{figure}
\begin{proposition}[Index charts, ND case] \label{prop:Map2}
	The integer tuple $ \cal{V} $, cf.\,\eqref{eq:IntVector}, depends on the von Neumann unitary $U$ only through some real parameters $(m,q) \in (-\infty, 0] \times \bbR $ determined by the boundary condition, whence $ \cal{V} (U) = \cal{V} (m,q) $. A boundary condition is particle-hole symmetric iff $ m=-1 $. The values of $ \cal{V} $ are charted in Fig.\,\ref{fig:IndicesII} (left panel).
	
	Bulk-edge correspondence, namely $ M = +2 $, is true within the sector
	\begin{equation}
		\cal{C} \coloneqq \{ |q| < |m+1| \}
	\end{equation}
	(cf.\,right panel of Fig.\,\ref{fig:IndicesII}), and violated outside of it. In particular, BEC is always false on the particle-hole symmetric submanifold $ \{ m = -1 \} $.
\end{proposition}
A comparison with \cite{GJT21} is in order. The boundary conditions of the 1-parameter family considered there are all of type DN and particle-hole symmetric, whence BEC is false throughout that family by our definition. That definition however differs from theirs, which was based on $P$ rather than $M$.

In case NN, a graphical representation akin to Figure \ref{fig:IndicesII} is not feasible, because of the sheer dimensionality of the associated BC manifold. Yet, restriction to the particle-hole symmetric submanifold makes a drawing possible.
\begin{figure}[hbt]
	\centering
	\includegraphics[width=0.9\linewidth]{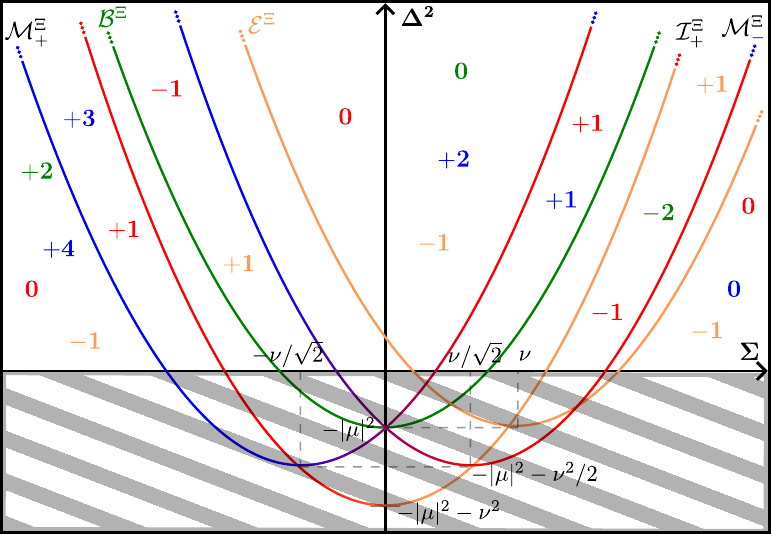}
	\caption{Map of indices for the PHS submanifold within the family NN. Color coding and transition lines follow the conventions of Fig.\,\ref{fig:IndicesII}. The transitions moreover illustrate Rem.\,\ref{rem:NonElementary}}.
	\label{fig:IndicesIV}
\end{figure}
\begin{proposition}[Index charts, NN case with symmetry] \label{prop:Map4}
	Within the particle-hole symmetric submanifold of NN, $ \cal{V} $ depends on $U$ only through parameters $ (\Sigma,\Delta^2, \mu) \in \bbR \times \bbR_+ \times \im \bbR $, whence $ \cal{V} (U) = \cal{V} (\Sigma, \Delta^2, \mu) $.
	
	Bulk-edge correspondence, namely $ M_+ = +2 $, is true within the region 
	\begin{equation}
		\{ \Delta^2 > \cal{I}^\Xi \coloneqq \Sigma^2 - | \mu |^2 - \nu^2 \} \,,
	\end{equation}
	and violated outside of it.
\end{proposition}
Question \ref{Q4} is finally answered by specifying a notion of \textit{typicality}. We say that a property is typical on the manifold of boundary conditions if it holds true with positive measure. The notion is actually proper to the measure class of the Lebesgue measure induced by any local chart.
\begin{theorem}[Typicality of BEC and violations thereof]
	\label{thm:Typicality}
	Within the families ND and NN, both BEC and its violation are typical, namely both $ M = C_+ = 2 $ and its opposite $ M \neq C_+ $ occur on subsets of positive measure.
\end{theorem}
In summary: Boundary conditions of the family DD exhibit bulk-edge correspondence, but the latter is typically violated in all the other families ND, DN and NN.  

In the following sections the above results will be made more precise, e.g.\,by determining the explicit dependence of $ \cal{V} $ on the boundary condition. It is in that version that they will be proven.

\section{Details of the setup} \label{sec:PreliminaryDetails}

\subsection{Momentum-space compactification and bulk eigensections} \label{subsec:BulkDetails}

In this section, we review \cite{GJT21} how the odd-viscous regularization enables the compactification $ \bbR^2 \mapsto S^2 $ of momentum space, mentioned in Sect.\,\ref{subsec:Bulk}. We moreover provide explicit sections $ \hat{\psi}^\zeta $ of the bundle $ E_+ $ of interest, associated to the positive eigenvalue $ \omega_+ $ of the bulk model. Such sections are regular everywhere but in a point $ \zeta \in S^2 $, hence the superscript. \vspace{1\baselineskip}

Above Eq.\,\eqref{eq:BlochBundle}, we argued that eigenspaces associated to $ \omega_\sigma (\ul{k}) \,, \ (\sigma = \pm,0) $ converge as $ k = | \ul{k} | \to \infty $, thus allowing a compactification of momentum space. That claim is justified below.

Consider the fiber $ H (\ul{k}) $, cf.\,Eq.\,\eqref{eq:BulkHam}, of the bulk Hamiltonian $H$, cf.\,\eqref{eq:HamGen}. It can be rewritten as
\begin{equation}
	H (\ul{k}) = \vec{d} (k) \cdot \vec{S} = (k_x, k_y, f -  \nu k^2) \cdot \vec{S} \,,
\end{equation}
where $ \vec{S} = (S_1, S_2, S_3) $, and the matrices
\begin{equation}
	S_1 = 
	\begin{pmatrix}
		0 & 1 & 0 \\
		1 & 0 & 0 \\
		0 & 0 & 0
	\end{pmatrix} \,, \qquad
	S_2 = 
	\begin{pmatrix}
		0 & 0 & 1 \\
		0 & 0 & 0 \\
		1 & 0 & 0
	\end{pmatrix} \,, \qquad
	S_3 = 
	\begin{pmatrix}
		0 & 0 & 0 \\
		0 & 0 & - \im \\
		0 & \im & 0
	\end{pmatrix} \,,
\end{equation}
are an irreducible spin-$1$ representation. The fiber $ H(\ul{k}) $ shares eigenprojections $ P_\pm (\ul{k}), P_0 (\ul{k}) $ with its flattened counterpart
\begin{equation}
	H' (\ul{k}) = \vec{e} (\ul{k}) \cdot \vec{S} \,, \qquad \vec{e} (\ul{k}) \coloneqq \frac{\vec{d} (\ul{k})}{| \vec{d} (\ul{k}) |} = \frac{1}{\omega_+ (\ul{k})} (k_x, k_y, f -  \nu k^2) \,. \label{eq:FlatHam}
\end{equation}
The eigenprojections read
\begin{equation}
	P_\pm (\ul{k}) = \frac{(\vec{e} (\ul{k}) \cdot \vec{S})^2 \pm \vec{e} (\ul{k}) \cdot \vec{S}}{2} \,, \qquad P_0 (\ul{k}) = \id - (\vec{e} (\ul{k}) \cdot \vec{S})^2 \,,
\end{equation}
and project onto the eigenspace of $ \omega_\pm (\ul{k}) $, $ \omega_0 (\ul{k}) $, cf.\,Eq.\,\eqref{eq:Omega+}, respectively.

The map $ \vec{e}: \ul{k} \in \bbR^2 \to S^2 $ is convergent for $ k = | \ul{k} | \to \infty $, and more precisely to
\begin{equation}
	\lim_{k \to \infty} \vec{e} (\ul{k}) = (0,0,-1) \,, 
\end{equation}
thanks to the odd-viscous regularization provided by the $ \nu k^2 $ term in \eqref{eq:FlatHam}. We stress that the limit is attained irrespective of the direction along which $ \ul{k} $ is sent to infinity. Accordingly, $ P_\sigma (\ul{k}) \,, \ (\sigma = \pm, 0) $ attains a well-defined limit for $ k \to \infty $, and so do the fibers
\begin{equation}
	E_{\sigma, \ul{k}} = \op{ran} P_\sigma (\ul{k})
\end{equation}
of the eigenbundle $ E_\sigma $, cf.\,Eq.\,\eqref{eq:BandBundle}. Along with the 1-point compactification of the plane $ \bbR^2 \ni \ul{k} $ to the (Riemann) sphere $ S^2 \cong \bbR^2 \cup \{ \infty \} $, the Bloch bundle extends as well.

Next, we shift our focus to the eigenbundle $ E_+ $ of interest, and provide sections that are global except for a single point.

An eigensection of $ E_+ $ is a smooth map $ \varphi: S^2 \to \bbC^2 $ that solves
\begin{equation}
	P_+ (\ul{k}) \varphi (\ul{k}) = \varphi (\ul{k}) \label{eq:SectionCondition}
\end{equation}
and does not vanish, where defined. We present two eigensections, namely
\begin{equation}
	\hat{\psi}^{\infty} (\ul{k}) = \frac{1}{k_x - \im k_y} \hat{\psi} (\ul{k}) \,, \qquad \hat{\psi}^{0} (\ul{k}) = \frac{1}{k_x + \im k_y} \hat{\psi} (\ul{k}) \,, \label{eq:PsiInfPsiZero}
\end{equation}
where
\begin{equation}
	\hat{\psi} (\ul{k}) \coloneqq
	\begin{pmatrix}
		k^2 / \omega_+ \\
		k_x - \im k_y q \\
		k_y + \im k_x q
	\end{pmatrix}
	\,, \qquad q (\ul{k}) \coloneqq \frac{f - \nu k^2}{\omega_+ (\ul{k})} \,. \label{eq:PsiHatAndQ}
\end{equation}
It is readily verified that they satisfy \eqref{eq:SectionCondition}. We have
\begin{equation}
	\begin{aligned}
		\hat{\psi} (\ul{k}) &=
		\begin{pmatrix}
			0 \\
			k_x + \im k_y \\
			k_y - \im k_x
		\end{pmatrix}
		+ o(|\ul{k}|^{0}) \,, \qquad (\ul{k} \to \infty) \,, \\
		\hat{\psi} (\ul{k}) &=
		\begin{pmatrix}
			0 \\
			k_x - \im k_y \\
			k_y + \im k_x
		\end{pmatrix}
		+ o(|\ul{k}|^2) \,, \qquad (\ul{k} \to 0) \,,
	\end{aligned}
\end{equation}
because $ q(\infty) = -1 $ and $ q(0) = 1 $. Hence, \eqref{eq:PsiInfPsiZero} extend smoothly to $ \ul{k} \in S^2 \setminus\{ \zeta \} $ with $ \zeta = \infty, 0 $, respectively. Indeed, and contrary to appearance, $ \hat{\psi}^\infty $ is smooth at $ \ul{k} = 0 $
\begin{equation}
	\hat{\psi}^\infty (0) = 
	\begin{pmatrix}
		0 \\
		1 \\
		\im
	\end{pmatrix} \,.
\end{equation}
Likewise,
\begin{equation}
	\hat{\psi}^0 (\infty) = 
	\begin{pmatrix}
		0 \\
		1 \\
		-\im
	\end{pmatrix} \,.
\end{equation}
We note that the two patches cover $ S^2 $ and that the transition function $t$,
\begin{equation}
	\hat{\psi}^\infty = t \hat{\psi}^0 \,, \qquad t = \frac{k_x + \im k_y}{k_x - \im k_y} \,,
\end{equation}
has winding $2$, coinciding with $ C_+ $ as expected.

\subsection{Scattering theory of shallow water waves} \label{subsec:ScattTheo}

This section is placed in the context where physical space is the upper half-plane $ \bbR^2_+ $, cf.\,Eq.\,\eqref{eq:HalfPlane}. The scattering of waves at the boundary $ y= 0 $ matters in various ways \cite{GJT21}. In this section, we review how solutions of the edge eigenvalue problem are found as \textit{scattering states}, namely as linear combinations of incoming, outgoing and evanescent \enquote{bulk} waves. At fixed boundary conditions, longitudinal momentum $k_x$, and incoming transverse momentum $ \kappa $, the (unique) scattering solution induces a \textit{scattering amplitude} $ S(k_x, \kappa) $, mapping incoming to outgoing waves. Eigenvalues $ \omega (k_x) $ in the discrete spectrum of $ H^\# (k_x) $ are found as poles of $S$ when analytically continued to the upper half-plane in $ \kappa \in \bbC $. The corresponding \textit{bound states} arise as outgoing (evanescent) waves without incoming counterpart. As $ k_x $ varies, the bound states are viewed as edge states of dispersion $ k_x \mapsto \omega (k_x) $.
\vspace{1\baselineskip}

Let a boundary condition $A$, cf.\,Def.\,\ref{def:OperativeBC}, be given. Scattering states are obtained as normal modes \eqref{eq:PsiTilde} of the formal differential operator $ H(k_x) $, cf.\,\eqref{eq:EdgeHamPrecursor}, that are asymptotic to incoming and outgoing plane waves. In the present context, the transverse momentum $ k_y $ is no longer a good quantum number. However, $ k_y $ remains a well-defined property of the asymptotic states $ \psi_{\op{in}} $ and $ \psi_{\op{out}} $, as they are of the form \eqref{eq:NormalModes}, and hence determined by their amplitude in $ E_{+,\ul{k}} $. We thus retain a parameter $ \kappa > 0 $ that describes the momenta $ - \kappa $ and $ + \kappa $ of the incoming and outgoing parts, $ \psi_{\op{in}} $ and $ \psi_{\op{out}} $, respectively. They are related by elastic reflection,
\begin{equation}
	s : (k_x, - \kappa) \longmapsto (k_x, \kappa) \,,
\end{equation}
because they then share the same frequency, $ \omega_+ (k_x, -\kappa) = \omega_+ (k_x, \kappa) \equiv \omega $, cf.\,\eqref{eq:Omega+}.

Moreover, there are two more solutions of $ \omega_+ (k_x, k_y) = \omega $, to be found in the analytic continuation of $ \omega_+ $. They are $ k_y = \pm \kappa_{\op{ev}} $, with $ \kappa_{\op{ev}} $ in the complex upper half-plane, and are related to $ \kappa > 0 $ by
\begin{equation}
	\kappa_{\op{ev}} (k_x, \kappa) = \im \sqrt{\kappa^2 + 2 k_x^2 + \frac{1 - 2 \nu f}{\nu^2}} \ \in \im \mathbb{R}_{+} \,. \label{eq:KappaEv}
\end{equation}
The two solutions are purely imaginary, and correspond to evanescent and divergent waves at $ y \to \infty $, respectively. In particular, $ \psi_{\op{ev}} $ participates in the scattering state $ \psi_s $ along with $ \psi_{\op{in}} $ and $ \psi_{\op{out}} $.

We summarize:
\begin{definition}[\textit{Scattering state}]
	\label{def:ScattState}
	For $ k_x \in \mathbb{R} $ and $ \kappa > 0 $, a \emph{scattering state} is a solution $ \psi_s = \tilde{\psi}_s (y; k_x, \kappa) e^{\im (k_x x - \omega t)} $, cf.\,\eqref{eq:PsiTilde}, of $ H (k_x) \tilde{\psi}_s (y; k_x, \kappa) = \omega(k_x, \kappa) \tilde{\psi}_s (y; k_x, \kappa)$ and of the form
	\begin{equation} \label{eq:ScattState}
		\psi_s \coloneqq \psi_{\op{in}} + \psi_{\op{out}} + \psi_{\op{ev}} \,,
	\end{equation}
	that moreover satisfies the boundary conditions. (The solution exists for every self-adjoint boundary condition $A$, and is unique up to multiples.)
\end{definition}
The scattering state \eqref{eq:ScattState} defines a \textit{scattering map}, acting between fibers $ E_{+, \ul{k}} = \op{ran} P_+ (\ul{k}) $ of the eigenbundle $ E_+ $ in Eq.\,\eqref{eq:BandBundle}, associated to the band $ \sigma = + \,$:
\begin{align}
	\mathscr{S} : E_{+,(k_x, -\kappa)} &\longrightarrow E_{+,(k_x, \kappa)} \,, \qquad (\kappa > 0) \nonumber \\
	\psi_{\op{in}} &\longmapsto \psi_{\op{out}} \label{eq:ScattMap} \,.
\end{align}
Schematically, the following diagram is commutative:
\begin{equation}
	\begin{tikzcd}
		E_+ \arrow{r}{\mathscr{S}} \arrow{d} & E_+ \arrow{d} \\
		\bbR^2 \arrow{r}{s} & \bbR^2 
	\end{tikzcd} \,. \label{eq:CommutingDiagram}
\end{equation}
Underlying the definition of $ \mathscr{S} $ is the identification of the amplitudes in $E_+$ with the waves seen in \eqref{eq:ScattState}.

In view of the fact that $ E_+ $ is nontrivial, sections will have to be local, cf.\,Sect.\,\ref{subsec:BulkDetails}. Let $ U_{\op{in}}, U_{\op{out}} \subset \bbR^2 $ be open subsets equipped with nowhere vanishing sections $ \hat{\psi}_{\op{in}}, \hat{\psi}_{\op{out}} $ of $ E_+ $. To be clear, $ U_{\op{in}} $ and $ U_{\op{out}} $ are not required to be subsets of the natural half-spaces (or hemispheres)
\begin{equation}
	S^2_{\op{out/in}} = \{ (k_x, \kappa) \ | \ \pm \kappa > 0 \} \,, \label{eq:Hemispheres}
\end{equation}
nor are they required to be related by the reflection $s$; even less $ \hat{\psi}_{\op{in}} $, $ \hat{\psi}_{\op{out}} $ to be related by \eqref{eq:ScattMap}. However, in case $ s(U_{\op{in}}) $ and $ U_{\op{out}} $ overlap, then a \textit{scattering amplitude} $ S = S (k_x, \kappa) \in \bbC \setminus \{0\} $ is defined for $ (k_x, \kappa) \in s(U_{\op{in}}) \cap U_{\op{out}} $ by
\begin{equation}
	\mathscr{S} \big( \hat{\psi}_{\op{in}} (k_x, -\kappa) \big) = S(k_x, \kappa) \hat{\psi}_{\op{out}} (k_x, \kappa) \,. \label{eq:SDef}
\end{equation}
An evanescent amplitude $T$ can likewise be defined, by considering an open set $ U_{\op{ev}} \subset \bbR \times \im \bbR $ of evanescent momenta $ (k_x, \kappa_{\op{ev}}) $. The role of the map $s$ is now taken by $ (k_x, \kappa) \mapsto (k_x, \kappa_{\op{ev}}) $, as in \eqref{eq:KappaEv}.

Bound states of the edge eigenvalue problem are now found by analytic continuation in $ \kappa \in \bbC $ of the scattering amplitude. Provided that the sections $ \hat{\psi}_\bullet $ above are analytic as well, $ \psi_{\op{out}} $ and $ \psi_{\op{in}} $ turn into exponentially decaying or growing solutions, as $ y \to + \infty $. More precisely,  $ \psi_{\op{in}} \coloneqq \hat{\psi} (k_x, - \kappa) e^{- \im \kappa y} $ is decaying for $ \op{Im} (\kappa) < 0 $, and the opposite holds for $ \psi_{\op{out}} $. The state \eqref{eq:ScattState} determines a bound state iff none of its terms is diverging in $ y \to + \infty $. That criterion is most useful for $ (k_x, \kappa) $ away from $ \infty $ and small $ \kappa \neq 0 $, thus exploring a complex neighborhood of the boundary of the bulk spectrum. There, the expression under the square root in \eqref{eq:KappaEv} remains bounded away from zero. Hence the criterion reduces to $ S(k_x, \kappa) = 0 $ and $ \op{Im} (\kappa) < 0 $. Equivalently, a bound state of energy $ \omega(k_x, \kappa) $ in the gap occurs iff
\begin{equation}
	S^{-1} (k_x, \kappa) = 0 \quad \op{and} \quad \op{Im} (\kappa) > 0 \,. \label{eq:BoundState}
\end{equation}
That amounts to the stance that \enquote{bound states are outgoing states with no incoming counterpart}.

\subsection{Topological content of the scattering amplitude and bulk-edge correspondence} \label{subsec:TopScatt}

In the previous section, we introduced the scattering amplitude based on overlapping open sets $ s(U_{\op{in}}), U_{\op{out}} $. Here, we will make more specific choices, such as $ U_{\op{out}} = U_{\op{in}} $.

First, however, we take $ U_{\op{out}} = S^2_{\op{out}} $ and $U_{\op{in}}$ to be the open set containing the closure of $ S^2_{\op{in}} $, which thus slightly protrudes into $ S^2_{\op{out}} $ and includes $\infty$. In particular, $ U_{\op{out}}, U_{\op{in}} $ form a cover of $ S^2 $. Both sets are depicted in Fig.\,\ref{fig:Levinson}, along with a loop $ \cal{C} $ running within their intersection and having the same orientation as $ \partial U_{ \op{out} } $.
\begin{figure}[hbt]
	\centering
	\includegraphics[width=0.9\linewidth]{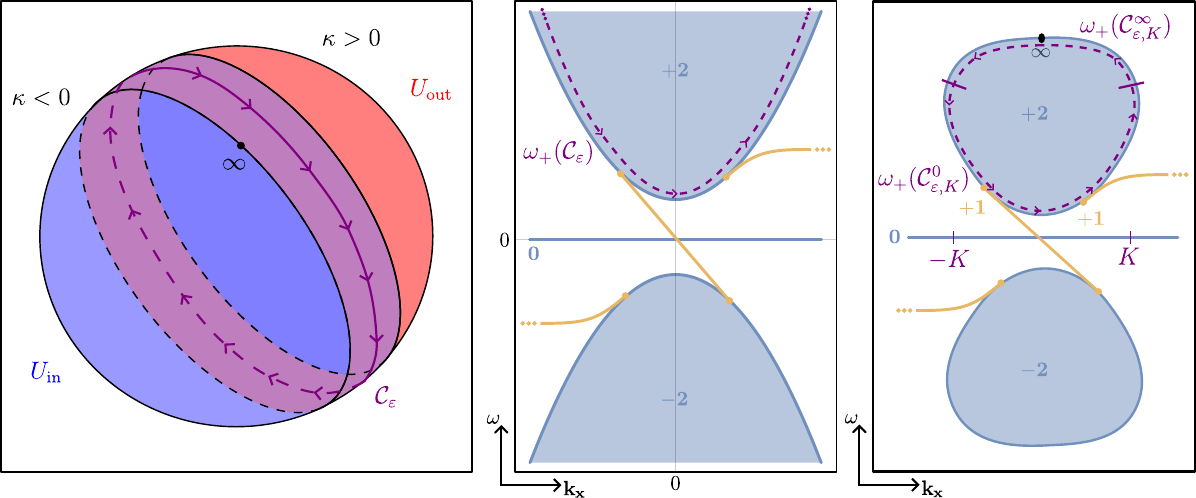}
	\caption{Left panel: compactified momentum space, patches $ U_{\op{in}} $ (blue) and $ U_{\op{out}} $ (red), contour $ \cal{C}_\epsilon $ (green) approximating the $ \{ k_y = 0 \} $ meridian as $ \epsilon \to 0 $. Central panel: image under $ \omega_+ $ of the loop $ \cal{C}_\epsilon $. Right panel: Finite and infinite portions of $ \cal{C}_\epsilon $, in the compactified $ (k_x, \omega) $-plane. \textcolor{red}{Modify figure: $ \varepsilon \mapsto \epsilon $!}}
	\label{fig:Levinson}
\end{figure}

Let $ \hat{\psi}_{ \op{in} } $ be a section on $ U_{ \op{in} } $. Then a section $ \hat{\psi}_{ \op{out} } $ is naturally defined on $ U_{ \op{out} } $ by \eqref{eq:ScattMap},
\begin{equation}
	\mathscr{S} (\hat{\psi}_{ \op{in} } (k_x, - \kappa)) = \hat{\psi}_{ \op{out} } (k_x, \kappa) \,, \qquad (\kappa > 0) \,.
\end{equation}
Since $ s(U_{\op{in}}) $ and $ U_{ \op{in} } $ overlap, we may take $ U_{ \op{in} } $ in the role of $ U_{ \op{out} } $ in \eqref{eq:SDef}, and obtain
\begin{equation}
	\mathscr{S} (\hat{\psi}_{ \op{in} } (k_x, - \kappa)) = S (k_x, \kappa) \hat{\psi}_{ \op{in} } (k_x, \kappa)
\end{equation}
for $ (k_x, \kappa) \in s(U_{\op{in}}) \cap U_{ \op{in} } = U_{ \op{out} } \cap U_{ \op{in} } $. Thus
\begin{equation}
	\hat{\psi}_{ \op{out} } (k_x, \kappa) = S(k_x, \kappa) \hat{\psi}_{ \op{in} } (k_x, \kappa) \label{eq:OperativeSDef}
\end{equation}
there, which exhibits the scattering amplitude as the transition function between local bundle charts of $ E_+ $. We conclude the following.
\begin{proposition}[\textit{Thm.\,2.7 of \cite{GJT21}, first item}]
	\label{prop:BulkScatteringCorrespondence}
	The scattering amplitude $S$, cf.\,\eqref{eq:SDef}, relates to the Chern number $ C_+ $ of $ E_+ $ by
	\begin{equation} \label{eq:SChernGM}
		\ds C_+ = \frac{1}{2 \pi \im} \oint_{\mathcal{C}} S^{-1} \mathrm{d} S =: W \,,
	\end{equation}
	where $ \cal{C} $ is any loop as described above.
\end{proposition}
At the same time, the scattering matrix $ S $ is sensitive to branches of the discrete spectrum joining the upper band, as illustrated in this version of \textit{Levinson's theorem}, originally developed in \cite{GP13} and later adapted to unbounded Hamiltonians, particularly to the SWM, in \cite{GJT21}.
\begin{theorem}[\textit{Levinson's relative theorem}]
	\label{thm:Levinson}
	Let $ \epsilon > 0 $, and pick (finite) $ k^{(1)}_x < k^{(2)}_x $ that do not correspond to a merger of an edge mode branch with the bulk spectrum of $ H^\# $.	Then
	\begin{equation} \label{eq:LevFin}
		\lim_{\epsilon \rightarrow 0} \operatorname{arg} S (k_x, \epsilon) \Big\vert_{k^{(1)}_x}^{k^{(2)}_x} = 2 \pi n (k^{(1)}_x, k^{(2)}_x) \,,
	\end{equation}
	where $ \operatorname{arg} $ denotes a continuous argument and $ n(k^{(1)}_x,k^{(2)}_x) $ is the signed number of edge mode branches emerging ($+$) or disappearing ($-$) at the lower band limit between $ k^{(1)}_x $ and $ k^{(2)}_x $, as $ k_x $ increases.
\end{theorem}
In simple terms, as we hover over the bottom of the positive band, the phase of $S$ jumps by $ + 2 \pi $ ($ - 2 \pi $) for each edge eigenvalue leaving (joining) the band. The theorem may be termed as a \textit{relative version} of Levinson's theorem \cite{RSIII}, since it compares the phases of the scattering matrix at threshold for two values of a parameter (here $ k_x $), rather than at two ends of an energy interval.

Eq.\,\eqref{eq:LevFin} accounts for the part of the winding \eqref{eq:SChernGM} that occurs because of proper mergers. More precisely, we consider a sequence of loops $ \cal{C}_\epsilon $ that tend to the \enquote{meridian} $ \{ (k_x, k_y) \ | \ k_y = 0 \} \cup \{ \infty \} $ in the limit $ \epsilon \to 0 $, cf.\,Fig.\,\ref{fig:Levinson}. Let
\begin{equation}
	\cal{C}_\epsilon = \cal{C}^0_{\epsilon,K} \cup \cal{C}^\infty_{\epsilon,K} \label{eq:ContourSplitting}
\end{equation}
be a partition into two arcs such that $ \cal{C}^0_{\epsilon,K} $ converges to the part $ |k_x| \leq K $ of the meridian, and $ \cal{C}^\infty_{\varepsilon,K} $ to the complementary part containing $ \{ \infty \} $. For any given boundary conditions, proper mergers do not accumulate at $ k_x \to \infty $. Thus
\begin{equation}
	\ds W_0 \coloneqq \lim_{\epsilon \rightarrow 0} \frac{1}{2 \pi \im} \int_{\mathcal{C}_{\epsilon,K}^{0}} S^{-1} \mathrm{d} S \,, \qquad W_\infty \coloneqq \lim_{\epsilon \rightarrow 0} \frac{1}{2 \pi \im} \int_{\mathcal{C}_{\epsilon,K}^{\infty}} S^{-1} \mathrm{d} S \,, \qquad W = W_0 + W_\infty \,, \label{eq:W0WInfty}
\end{equation}
are well-defined, are integers, and do not depend on $ K $, if sufficiently large.

Eq.\,\eqref{eq:LevFin} states
\begin{equation}
	W_0 = P \,.
\end{equation}
Recalling Eq.\,\eqref{eq:SChernGM} and $ W = W_0 + W_\infty $, we moreover conclude
\begin{equation}
	C_+ = W = W_0 + W_\infty = P + W_\infty \,. \label{eq:CPlusW0WInfty}
\end{equation}
Our Def.\,\ref{def:BEC} of bulk-edge correspondence, $ M = C_+ $, where $ M = I+P $ is the number of mergers, thus amounts to
\begin{equation}
	I = W_\infty \,.
\end{equation}
In other words, violations of BEC occur when the number of improper mergers $I$ differs from the winding of $ S $ at infinite momentum. Such events do occur.

\subsection{Details on self-adjoint boundary conditions and particle-hole symmetry thereof} \label{subsec:DetailsBCs}

We shall characterize the families DD, ND, DN, NN of self-adjoint boundary conditions (cf.\,Def.\,\ref{def:Families}) and their particle-hole symmetric subsets. Later, in Sect.\,\ref{sec:DetailedResults}, the main results of the paper will be formulated in the terms introduced here. The proofs of the statements in this section are found in App.\,\ref{app:SelfAdj}. \vspace{1\baselineskip}

Boundary conditions $ A $, cf.\,Def.\,\ref{def:OperativeBC}, are affine linear maps $ \bbR \ni k_x \mapsto A(k_x) \in \op{Mat}_{2 \times 6} (\bbC) $. As is proven in \cite{GJT21}, Lm.\,B.1-3, and as we shall review in App.\,\ref{app:SelfAdj}, self-adjointness of $ H^\# (k_x)$, see below Eq.\,\eqref{eq:EdgeHamPrecursor}, imposes various constraints on the entries of the $ 2 \times 6 $ matrix $A (k_x)$. For short, they are of two kinds, linear and non-linear. Quite generally \cite{KS99}, self-adjoint boundary conditions correspond to subspaces $V$ of boundary values $ \Psi = (\psi, \psi') $, cf.\,\eqref{eq:BoundaryCondition}, that are maximally isotropic w.r.t.\,some matrix $ \hat{\Omega} $, i.e.\,$ \hat{\Omega} V = V^\perp $.

In the present context, $ V \subset \bbC^6 $, but the rank of $ \hat{\Omega} $ is not maximal, i.e., not $6$, but $ \op{rk} \hat{\Omega} = 4 $. This comes about because the deficiency indices of the half-line problem are $ \langle 2,2 \rangle $, rather than $ \langle 3,3 \rangle $ as one might instead suspect in view of $3$ components of the field and the order 2 of the differential operator \eqref{eq:HamGen}. Note, indeed, that this differential operator is of first order in $ \eta $, but of second in $ u,v $. That accounts for the fact that the $ 2 \times 6 $ matrix $ A $ can be expressed by a $ 2 \times 4 $ matrix $ \ul{A} $, which solves the linear constraint. The non-linear constraint, now on $ \ul{A} $, comes from the isotropy.  

The $ 2 \times 4 $ matrix $ \ul{A} $ is related to the $ 2 \times 6 $ $A$ by
\begin{equation}
	A = \ul{A} M \,, \label{eq:AulA}
\end{equation}
where
\begin{equation}
	M =
	\begin{pmatrix}
		1 & 0 & 0 & 0 & - \nu & 0 \\
		0 & 1 & 0 & 0 & 0 & 0 \\
		0 & 0 & 1 & 0 & 0 & 0 \\
		0 & 0 & 0 & 0 & 0 & \nu
	\end{pmatrix} \,, \label{eq:MatM}
\end{equation}
as shall be justified in App.\,\ref{app:SelfAdj}. In this new notation, the boundary condition \eqref{eq:BoundaryCondition} is written as
\begin{equation} 
	\ul{A} (k_x) M \Psi = 0 \,. \label{eq:NewBC}
\end{equation}
By locality of $ A $, cf.\,\eqref{eq:LocalBC}, and \eqref{eq:AulA}, $ \ul{A} $ must take the form
\begin{gather} 
	\ul{A} (k_x) = \ul{A}_{c} + k_x \ul{A}_k \,, \nonumber\\
	\ul{A}_{c} = (a_1, \ a_1', \ a_2', \ a_2), \hspace{10pt} \ul{A}_k = (0, \ b_1, \ b_2, \ 0) \, \label{eq:NewA}
\end{gather}
and we shall also write it as
\begin{equation}
	\ul{A} (k_x) = \left( A_1 (k_x), A_2 (k_x) \right), \label{eq:JuxtaposeA1A2}
\end{equation}
with $ (\cdot, \cdot) $ denoting horizontal juxtaposition and $ A_1, A_2 $ square $ 2 \times 2 $ matrices given by 
\begin{gather} 
	A_i (k_x) = A^0_i + k_x B_i \,, \nonumber \\
	A^0_1 = (a_1, \ a_1') \,, \hspace{10pt} A^0_2 = (a_2', \ a_2) \,, \hspace{10pt} B_1 = (0, \ b_1) \,, \hspace{10pt} B_2 = (b_2, \ 0) \,, \label{eq:A1A2}
\end{gather}
$ a_i, a_i', b_i \in \bbC^2 $, $ (i = 1,2) $. 

The specific structure \eqref{eq:NewA} of the $ 2 \times 4 $ matrix
\begin{equation}
	\ul{A} (k_x) = (a_1, \ a_1' + k_x b_1, \ a_2' + k_x b_2, \ a_2) \label{eq:ExplicitUlA}
\end{equation}
determines the general form of the original $ 2 \times 6 $ boundary condition $ A $ by
\begin{equation}
	A (k_x) = \ul{A} (k_x) M = (a_1, \ a_1' + k_x b_1, \ a_2' + k_x b_2, \ 0, \ - \nu a_1, \ \nu a_2) \,, \label{eq:GeneralTwoBySix}
\end{equation} 
cf.\,\eqref{eq:MatM}. Note incidentally that the vanishing entry in \eqref{eq:GeneralTwoBySix} states that no boundary condition involves $ \partial_y \eta $, as was stated following Def.\,\ref{def:Families}, and in line with the above comment on deficiency indices.

By \eqref{eq:LocalBC}, the local boundary condition $ A(k_x) $ may moreover be written as
\begin{equation}
	A (k_x) = (A_0 + \im k_x A_x, A_y) \,. \label{eq:LocalBCAgain}
\end{equation}
Comparing Eqs.\,(\ref{eq:GeneralTwoBySix}, \ref{eq:LocalBCAgain}) then yields
\begin{equation}
	A_0 = (a_1, \ a_1', \ a_2') \,, \qquad A_x = (0, \ b_1, \ b_2) \,, \qquad A_y = (0, \ - \nu a_1, \ \nu a_2) \,. \label{eq:ExplicitA0AxAy}
\end{equation}

The definition of families, Def.\,\ref{def:Families}, is a distinction of cases. It affects the matrix $ A_y $ only, and thus the vectors $ a_i $, $ (i=1,2) $, but not $ a_i', b_i $. In the next lemma, we state the implications for $ a_1, a_2 $. In doing so, we separate the case ND/DN into two, depending on whether $ \op{rk} A_y \neq 0 $ is ensured by $ a_1 \neq 0 $ or $ a_2 \neq 0 $.
\begin{lemma} \label{lem:Dictionary}
	Let $ A(k_x) = (A_0 + \im k_x A_x , A_y) $ be a $ 2 \times 6 $ matrix such that $ A = \ul{A} M $ for some $ \ul{A} (k_x) $ as in \eqref{eq:ExplicitUlA}. If the associated boundary condition
	\begin{equation}
		(A_0 + A_x \partial_x + A_y \partial_y) \psi = 0\,,
	\end{equation} 
	cf.\,Eq.\,\eqref{eq:LocalBC}, belongs to family
	\begin{itemize}[itemindent=1em]
		\item[DD:] then $ a_1 = a_2 = 0 $;
		
		\item[ND:] then $ a_1 \neq 0 $ and $ a_2 = \alpha a_1 $ for some $ \alpha \in \bbC $;
		
		\item[DN:] then $ a_2 \neq 0 $ and $ a_1 = \alpha a_2 $ for some $ \alpha \in \bbC $;
		
		\item[NN:] then $ a_1, a_2 $ are linearly independent.
	\end{itemize}
\end{lemma}
The non-linear constraint on $ \ul{A} $ remains to be solved. It says:
\begin{equation} \label{eq:CondSelfAdjKxNew}
	\ul{A} (k_x) \Omega \ul{A}^* (k_x) = 0 \,, \hspace{12pt}
	\Omega =
	\begin{pmatrix}
		0 & \id_2 \\
		\id_2 & 0
	\end{pmatrix} \,,
\end{equation}
almost everywhere in $ k_x $, see App.\ref{app:SelfAdj}. Equivalently, using \eqref{eq:JuxtaposeA1A2},
\begin{equation} \label{eq:CondOnA}
	A_1 (k_x) A_2^* (k_x) + A_2 (k_x) A_1^* (k_x) = 0
\end{equation}
a.e.\,in $ k_x $, which is in turn to say
\begin{equation} \label{eq:CondSelfAdj2New}
	\begin{gathered}
		A^0_1 A^{0*}_2 + A^0_2 A^{0*}_1 = 0 \,, \\
		B_1 A^{0*}_2 + A^0_1 B_{2}^* + A^0_2 B_{1}^* + B_2 A^{0*}_1 = 0 \,,
	\end{gathered}
\end{equation}
cf.\,\eqref{eq:A1A2}. Those conditions translate into others for $ a_i', b_i $, $ (i=1,2) $. Remarkably, they amount to real linear equations. The next proposition states them, and in fact separately for the different families of BCs, cf.\,Lm.\,\ref{lem:Dictionary}.
\begin{proposition}[\textit{Families of boundary conditions}] \label{prop:BCs}
	Let $ \ul{A}: k_x \to \op{Mat}_{2 \times 4} (\bbC) $ be as in \eqref{eq:NewA}. Depending on the cases DD through NN, see Lm.\,\ref{lem:Dictionary}, Eq.\,\eqref{eq:CondSelfAdj2New} amounts to
	\begin{enumerate}[itemindent=1em]
		\item[DD:] $ a_1', \ a_2', \ b_1, \ b_2 \in \bbC^2 $ (arbitrary);
		
		\item[ND:] For some $ \lambda, \lambda' \in \bbR $,
		\begin{equation}
			\begin{array}{lcl}
				\bar{\alpha} a_1' + a_2' & = & \im \lambda' a_1  \\
				\bar{\alpha} b_1 + b_2 & = & \im \lambda a_1  \,.
			\end{array} \label{eq:ConditionsND}
		\end{equation}
		
		\item[DN:] Same as ND, but with $ a_1 $ and $ a_2 $ interchanged.
		
		\item[NN:] For some $ \mu, \mu' \in \bbC $ and $ \lambda_1, \lambda_2, \lambda_1', \lambda_2' \in \bbR $,
		\begin{equation}
			\begin{array}{lcll}
				a_1' & = & \mu' a_1 + \im \lambda_1' a_2 \\
				a_2' & = & - \bar{\mu}' a_2 + \im \lambda_2' a_1 
			\end{array} \label{eq:ConditionsNN1}
		\end{equation}
		and
		\begin{equation}
			\begin{array}{lcl}
				b_1 & = & \mu a_1 + \im \lambda_1 a_2 \\
				b_2 & = & - \bar{\mu} a_2 + \im \lambda_2 a_1 \,.
			\end{array} \label{eq:ConditionsNN2}
		\end{equation}
	\end{enumerate}
\end{proposition}
We shall also provide an explicit characterization of the particle-hole symmetric (cf.\,Def.\,\ref{def:PHS}) subsets of those families.
\begin{proposition}[\textit{Particle-hole symmetric families of boundary conditions}] \label{prop:PHSBCs}
	The map $ \ul{A}: k_x \to \op{Mat}_{2 \times 4} (\bbC) $ of \eqref{eq:NewA} encodes a particle-hole symmetric and self-adjoint boundary condition iff one of the following is true.
	\begin{enumerate}[itemindent=1em]
		\item[$ \op{DD}_{\Xi} $:] $ a_1 = a_2 = 0 $.
		
		\item[$ \op{ND}_{\Xi} $:] $ a_1 \neq 0 $, $ a_2 = \alpha a_1 $ for some $ \alpha \in \bbR $ and
		\begin{equation}
			\begin{array}{lcl}
				\alpha a_1' + a_2' & = & 0 \\
				\alpha b_1 + b_2 & = & \im \lambda a_1 \,,
			\end{array}
		\end{equation}
		for some $ \lambda \in \bbR $.
		
		\item[$ \op{DN}_{\Xi} $:] Same as $ \op{ND}_{\Xi} $, but with $ a_1 $ and $ a_2 $ interchanged.
		
		\item[$ \op{NN}_{\Xi} $:] $ a_1, a_2 $ linearly independent and
		\begin{equation}
			\begin{array}{lcl}
				a_1' & = & \mu' a_1 \\
				a_2' & = & - \mu' a_2 \\
				b_1 & = & \mu a_1 + \im \lambda_1 a_2 \\
				b_2 & = & \mu a_2 + \im \lambda_2 a_1 \,,
			\end{array} \label{eq:NNPHS}
		\end{equation}
		for some $ \mu' \in \bbR $, $ \mu \in \im \bbR $ and $ \lambda_1, \lambda_2 \in \bbR $.
	\end{enumerate}
\end{proposition}
We now elucidate how \textit{von Neumann unitaries} $ U $, cf.\,Def.\,\ref{def:OperativeBC}, relate to orbits (under left-action of $ \op{GL} (2, \bbC) $) of boundary conditions $ \ul{A} $. First of all, notice that $ \ul{A} $ and $ A $ are in one-to-one correspondence, upon imposing self-adjointness. Then, so are their orbits $ [A], [\ul{A}] $, cf.\,Eq.\,\eqref{eq:Cosets}.

The von Neumann unitary $ U $ associated with orbit $ [\ul{A}] $ (equivalently $ [A] $) is given by
\begin{equation} \label{eq:UFormula}
	\begin{array}{llcl}
		U : & \bbR & \longmapsto & U(2) \\
		& k_x & \longmapsto & U(k_x) = (A_1 (k_x) + A_2 (k_x))^{-1} (A_1 (k_x) - A_2 (k_x)) \,,
	\end{array} 
\end{equation}
where $ A_1, A_2 $ are as in Eq.\,\eqref{eq:JuxtaposeA1A2}. Eq.\,\eqref{eq:UFormula} yields a well-defined map $ \ul{A} \to U $, that descends to a map $ \upsilon: [\ul{A}] \mapsto U $. The map $ \upsilon $ is moreover one-to-one. These fine points are dealt with in App.\,\ref{app:SelfAdj}.

As we argued below Eq.\,\eqref{eq:Cosets}, two representatives of the same orbit $ [\ul{A}] $ encode for the same boundary conditions. The next proposition describes the orbits, and it does so for each family of boundary conditions at a time. Moreover, it does so in terms of the unitaries that are in bijective correspondence with those orbits, cf.\,Rem.\,\ref{rem:Orbits}.
\begin{proposition}[\textit{Families of von Neumann unitaries}] \label{prop:UBCs}
	(i) Let $ \ul{A} $ be a boundary condition parametrized as in Prop.\,\ref{prop:BCs} and Lm.\,\ref{lem:Dictionary}. The corresponding curves in $ U(2) $, cf.\,\eqref{eq:UFormula}, are of the form
	\begin{equation}
		U (k_x) = J + \hat{U} (k_x) \,, \qquad
		J = 
		\begin{pmatrix}
			-1 & 0 \\
			0 & 1
		\end{pmatrix} \,, \label{eq:JSummand}
	\end{equation}
	with $ \hat{U} $ depending on families as follows.
	\begin{enumerate}[itemindent=1em]
		\item[$ \op{DD} $:] $ \hat{U} (k_x) = 0 $.
		
		\item[$ \op{ND} $:] In terms of the parameters $ \alpha \in \bbC $ and $ \lambda, \lambda' \in \bbR $,
		\begin{equation}
			\hat{U} (k_x) =
			\frac{2}{1 + | \alpha|^2 + \im (k_x \lambda + \lambda')}
			\begin{pmatrix}
				1 & - \alpha \\
				\bar{\alpha} & - |\alpha|^2
			\end{pmatrix}  \,. \label{eq:UND}
		\end{equation}
		
		\item[$ \op{NN} $:] In terms of the parameters $ \mu, \mu' \in \bbC $ and $ \lambda_1, \lambda_2, \lambda_1', \lambda_2' \in \bbR $,
		\begin{equation}
			\begin{gathered}
				\hat{U} (k_x) = - \frac{2}{\det \tilde{U}} \, \tilde{U} (k_x) \,, \\
				\tilde{U} (k_x) = 
				\begin{pmatrix}
					1 + \im (\lambda_1' + k_x \lambda_1) & \mu' + k_x \mu \\
					\bar{\mu}' + k_x \bar{\mu} & -1 -\im (\lambda_2' + k_x \lambda_2)
				\end{pmatrix} \,.
			\end{gathered}
			\label{eq:UNN}
		\end{equation}
	\end{enumerate}
	The family DN is obtained by the substitution detailed in Prop.\,\ref{prop:BCs}.
	
	(ii) The particle-hole symmetric subfamilies are given by the same formal expressions, with restrictions on the parameters:
	\begin{enumerate}[itemindent=1em]
		\item[$ \op{DD}_{\Xi} $:] No changes, DD is always particle-hole symmetric;
		
		\item[$ \op{ND}_{\Xi} $:] $ \lambda' = 0 $ and $ \alpha \in \bbR $;
		
		\item[$ \op{NN}_{\Xi} $:] $ \lambda_1' = \lambda_2' = 0 $, $ \mu \in \im \bbR $ and $ \mu' \in \bbR $.
	\end{enumerate}
\end{proposition}
\begin{remark} \label{rem:ConsequenceTranslationKx}
	Recall Rem.\,\ref{rem:TranslationsKx}. Shifts $ k_x \mapsto k_x + \tau \,, \ (\tau \in \bbR) $ induce corresponding shifts $ \lambda' \mapsto \lambda' + \tau \lambda $ in family ND. The parameter $ \lambda' $ can thus be set to $ \lambda' = 0 $, if $ \lambda \neq 0 $, with no effect on $ \cal{V} $. Similarly, in family NN, shifts $ \lambda_i' \mapsto \lambda_i' + \tau \lambda_i $, $ (i=1,2) $, $ \mu' \mapsto \mu' + \tau \mu $ are induced and they are expected to leave $ \cal{V} $ invariant. The former fact is explicitly seen in Prop.\,\ref{prop:NewMap2}, the latter in Prop.\,\ref{prop:NewMap4}.
\end{remark}

\section{Results} \label{sec:DetailedResults}

In this section, we deepen and extend some of the results of Sect.\,\ref{sec:Results}. More precisely, we deepen Prop.\,\ref{prop:Map2} (ND family) in that we relate the there not yet specified parameters $ m,q $ to the more fundamental ones ($ \alpha \in \bbC $, $ \lambda, \lambda' \in \bbR $), labeling the boundary conditions of that family; or rather their orbits, which is all that matters, cf.\,Prop.\,\ref{prop:BCs} and Prop.\,\ref{prop:UBCs}. Likewise, we deepen Prop.\,\ref{prop:Map4} (NN family) by relating the parameters $ \Sigma, \Delta^2 $ to the parameters $ \mu, \mu' \in \bbC $, $ \lambda_i, \lambda_i' \in \bbR $ $ (i=1,2) $, labeling the orbits of that family. Moreover, we extend Prop.\,\ref{prop:Map4}, in that we treat all boundary conditions of the family NN, rather than just the particle-hole symmetric ones. As a result, the charts mapping the integer tuple $ \cal{V} $ extend accordingly.

By contrast, Prop.\,\ref{prop:Map1} (family DD) does not require further details. The same applies to other results (Thm.\,\ref{thm:transitions}, Cor.\,\ref{cor:ViolationMechanism}, Thm.\,\ref{thm:Typicality}) spanning across families. \vspace{1\baselineskip}

As stated, we begin by rephrasing Prop.\,\ref{prop:Map2}.
\begin{proposition}[\textit{Prop.\,\ref{prop:Map2}, ND family, full statement}] \label{prop:NewMap2}
	Let $ \ul{A}: k_x \to \ul{A} (k_x) $ be a boundary condition in family ND, parametrized as in Prop.\,\ref{prop:BCs}. Let moreover
	\begin{equation}
		m \coloneqq - \frac{| 1 + \im \alpha |^2}{| 1- \im \alpha |^2} \in (-\infty, 0] \,, \qquad q \coloneqq - \frac{2 \lambda}{\nu | 1- \im \alpha |^2} \in \bbR \,. \label{eq:MAndQ}
	\end{equation}
	Then, the entries of the integer tuple $ \cal{V} = (P,I,E,B) $, cf.\,\eqref{eq:IntVector}, take the following values.
	\begin{align}
		P &= 
		\begin{cases}
			1 \,, & \quad q \leq - \sqrt{2} \\
			2 \,, & \quad - \sqrt{2} < q < \sqrt{2} \\
			3 \,, & \quad q \geq \sqrt{2} \,,
		\end{cases} \label{eq:PND} \\
		I &=
		\begin{cases}
			-1 \,, & \quad ( \sqrt{2} < q < - (m+1) ) \ \lor \ ( - \sqrt{2} < q < -|m+1| ) \\
			0 \,, & \quad ( |q| < |m+1| \ \land \ |q| < \sqrt{2} ) \ \lor \ ( |q| > |m+1| \ \land \ |q| > \sqrt{2} ) \\
			+1 \,, & \quad ( m+1 < q < - \sqrt{2} ) \ \lor \ ( |m+1| < q < \sqrt{2} ) \,,
		\end{cases} \label{eq:IND} \\
		E &= 
		\begin{cases}
			-1 \,, & \quad (q < m-1) \ \lor \ (q > |m+1|) \\
			0 \,, & \quad - |m+1| < q < |m+1| \\
			+1 \,, & \quad m-1 < q < - |m+1| \,,
		\end{cases} \label{eq:EValueND} \\
		B &= 
		\begin{cases}
			\op{sgn} q \,, & \quad q \neq 0 \\
			0 \,, & \quad q=0 \,.
		\end{cases} \label{eq:BND}
	\end{align}
\end{proposition}
The proof of this result, and of the following one, will be given in Sect.\,\ref{sec:Proofs}.
\begin{remark} \label{rem:GenuineParsND}
	As they should, the integers do depend on the boundary conditions \eqref{eq:ExplicitUlA} only through their orbits, and thus on $ \alpha \in \bbC $, $ \lambda, \lambda' \in \bbR $ at most, cf.\,Rem.\,\ref{rem:Orbits} and \ref{eq:UND}. Actually, they do not depend on $ \lambda' $, as argued in Rem.\,\ref{rem:ConsequenceTranslationKx}.
\end{remark}
Prop.\,\ref{prop:Map4} can similarly be rewritten, moreover considering the entire family NN, rather than its particle-hole symmetric subset.
\begin{proposition}[\textit{Prop.\,\ref{prop:Map4}, NN family, full statement}] \label{prop:NewMap4}
	Let $ \ul{A}: k_x \to \ul{A} (k_x) $ be a boundary condition in family NN, parametrized as in Prop.\,\ref{prop:BCs}. Let moreover
	\begin{equation}
		\Sigma \coloneqq \frac{\lambda_1 + \lambda_2}{2} \,, \qquad \Delta \coloneqq \frac{\lambda_1 - \lambda_2}{2} \,, \qquad \mu = \op{Re} \mu + \im \op{Im} \mu \equiv \mu_R + \im \mu_I \,. \label{eq:SigmaAndDelta}
	\end{equation}
	The integers $ P,I,E $, cf.\,\eqref{eq:IntVector}, depend only on $ \mu \in \bbC $, $ \Sigma, \Delta \in \bbR $ and take the following values.
	\begin{itemize}
		\item Let 
		\begin{equation}
			\cal{M}_{\pm} \coloneqq \Sigma^2 - | \mu |^2 \pm \nu \sqrt{2} (\mu_R + \Sigma) \,. \label{eq:MPlusMinus}
		\end{equation}
		Then
		\begin{equation}
			P = 
			\begin{cases}
				+2 \,, & \quad \Delta^2 > \cal{M}_+ \,, \ \cal{M}_- \\
				+1 \,, & \quad \cal{M}_- < \Delta^2 < \cal{M}_+ \\
				+3 \,, & \quad \cal{M}_+ < \Delta^2 < \cal{M}_- \\
				0 \,, & \quad \Delta^2 < \cal{M}_- < \cal{M}_+ \\
				+4 \,, & \quad \Delta^2 < \cal{M}_+ < \cal{M}_- \,.
			\end{cases} \label{eq:PNN}
		\end{equation}
		
		\item Let
		\begin{equation}
			\cal{I}_\pm \coloneqq \Sigma^2 - | \mu |^2 \pm 2 \nu \mu_R - \nu^2 \,. \label{eq:IPlusMinus}
		\end{equation}
		The value of $I$ stays constant within some regions of parameter space $ \bbR^4 \ni (\mu_R, \mu_I, \Sigma, \Delta) $. A first property distinguishing them is the value of $ \mu_R + \Sigma $, which leads to four tables. Within each of them, further inequalities determine the actual regions. Combinations of inequalities that are empty are denoted by $ \varnothing $, or are omitted. The values of $I$ are listed as entries.
		\begin{center}
			\begin{tabular}{ c"c|c|c}
				\xrowht[()]{10pt}
				$ \mu_R + \Sigma < - \nu $ & $ \Delta^2 < \cal{M}_+ < \cal{M}_- $ & $ \cal{M}_+ < \Delta^2 < \cal{M}_- $ & $ \Delta^2 > \cal{M}_- > \cal{M}_+ $ \\
				\thickhline\xrowht{10pt}
				$ \Delta^2 < \cal{I}_{\pm} $ & $0$ & $+1$ & $\varnothing$ \\
				\hline\xrowht{10pt}
				$ \cal{I}_- < \Delta^2 < \cal{I}_+ $ & $-1$ & $0$ & $+1$ \\
				\hline\xrowht{10pt}
				$ \cal{I}_+ < \Delta^2 < \cal{I}_- $ & $\varnothing$ & $0$ & $\varnothing$ \\
				\hline\xrowht{10pt}
				$ \Delta^2 > \cal{I}_\pm $ & $\varnothing$ & $-1$ & $0$ \\
				\multicolumn{4}{c}{\vspace{1\baselineskip}} \\
				\xrowht[()]{10pt}
				$ - \nu < \mu_R + \Sigma < 0 $ & $ \Delta^2 < \cal{M}_+ < \cal{M}_- $ & $ \cal{M}_+ < \Delta^2 < \cal{M}_- $ & $ \Delta^2 > \cal{M}_- > \cal{M}_+ $ \\
				\thickhline\xrowht{10pt}
				$ \Delta^2 < \cal{I}_{\pm} $ & $\varnothing$ & $\varnothing$ & $\varnothing$  \\
				\hline\xrowht{10pt}
				$ \cal{I}_- < \Delta^2 < \cal{I}_+ $ & $-1$ & $0$ & $+1$ \\
				\hline\xrowht{10pt}
				$ \cal{I}_+ < \Delta^2 < \cal{I}_- $ & $\varnothing$ & $-2$ & $-1$ \\
				\hline\xrowht{10pt}
				$ \Delta^2 > \cal{I}_\pm $ & $\varnothing$ & $-1$ & $0$ 
			\end{tabular}
			\label{tab:MergeNN1AndNN2}
		\end{center}
		\begin{center}
			\begin{tabular}{ c"c|c|c}
				\xrowht[()]{10pt}
				$ 0 < \mu_R + \Sigma < \nu $ & $ \Delta^2 < \cal{M}_- < \cal{M}_+ $ & $ \cal{M}_- < \Delta^2 < \cal{M}_+ $ & $ \Delta^2 > \cal{M}_+ > \cal{M}_- $ \\
				\thickhline\xrowht{10pt}
				$ \Delta^2 < \cal{I}_{\pm} $ & $\varnothing$ & $\varnothing$ & $\varnothing$ \\
				\hline\xrowht{10pt}
				$ \cal{I}_- < \Delta^2 < \cal{I}_+ $ & $\varnothing$ & $+2$ & $+1$ \\
				\hline\xrowht{10pt}
				$ \cal{I}_+ < \Delta^2 < \cal{I}_- $ & $+1$ & $0$ & $-1$ \\
				\hline\xrowht{10pt}
				$ \Delta^2 > \cal{I}_\pm $ & $\varnothing$ & $+1$ & $0$ \\
				\multicolumn{4}{c}{\vspace{1\baselineskip}} \\
				\xrowht[()]{10pt}
				$ \mu_R + \Sigma > \nu $ & $ \Delta^2 < \cal{M}_- < \cal{M}_+ $ & $ \cal{M}_- < \Delta^2 < \cal{M}_+ $ & $ \Delta^2 > \cal{M}_+ > \cal{M}_- $ \\
				\thickhline\xrowht{10pt}
				$ \Delta^2 < \cal{I}_{\pm} $ & $0$ & $-1$ & $\varnothing$  \\
				\hline\xrowht{10pt}
				$ \cal{I}_- < \Delta^2 < \cal{I}_+ $ & $\varnothing$ & $0$ & $\varnothing$ \\
				\hline\xrowht{10pt}
				$ \cal{I}_+ < \Delta^2 < \cal{I}_- $ & $+1$ & $0$ & $-1$ \\
				\hline\xrowht{10pt}
				$ \Delta^2 > \cal{I}_\pm $ & $\varnothing$ & $+1$ & $0$ 
			\end{tabular}
			\captionof{table}{Values of $I$ in different regions within family NN.}
			\label{tab:MergeNN3AndNN4}
		\end{center}

		\item Let
		\begin{equation}
			\cal{E} \coloneqq (\Sigma - \nu)^2 - | \mu |^2 \,. \label{eq:EPlusMinus}
		\end{equation}
		Then, the values of $E$, for each region of parameter space $ \bbR^4 \ni (\mu_R, \mu_I, \Sigma, \Delta) $ where it stays constant, are listed in the following table.
		\begin{center}
			\begin{tabular}{c"c|c}
				\xrowht[()]{10pt}
				& $ \Delta^2 > \cal{E} $ & $ \Delta^2 < \cal{E} $ \\
				\thickhline\xrowht{10pt}
				$ \Delta^2 < \cal{I}_{\pm} $ & $+1$ & $-1$ \\
				\hline\xrowht{10pt}
				$ \cal{I}_- < \Delta^2 < \cal{I}_+ $ & $0$ & $0$ \\
				\hline\xrowht{10pt}
				$ \cal{I}_+ < \Delta^2 < \cal{I}_- $ & $0$ & $0$\\
				\hline\xrowht{10pt}
				$ \Delta^2 > \cal{I}_\pm $ & $-1$ & $+1$ 
			\end{tabular}
			\captionof{table}{Values of $E$ in different regions within family NN.}
			\label{tab:EscapesNN}
		\end{center}
	\end{itemize}
	The winding $B$ of the von Neumann unitary $U$, cf.\,\eqref{eq:UNN}, takes values as follows. Let
	\begin{equation}
		\cal{B} \coloneqq \Sigma^2 - |\mu|^2 \,. \label{eq:DefB}
	\end{equation}
	Then, if $ \Delta^2 \neq \cal{B} $ and $ \Sigma \neq 0 $ (general case),
	\begin{equation}
		B =
		\begin{cases}
			0 \,, & \quad \Delta^2 > \cal{B} \\
			-2 \op{sgn} \Sigma \,, & \quad \Delta^2 < \cal{B} \,. \\
		\end{cases} \label{eq:BNN}
	\end{equation}
	By contrast, in the exceptional case where \textbullet\ $ \Delta^2 = \cal{B} $, \textbullet\ $ \Sigma \neq 0 $ or $ \lambda_2 \lambda_1' + \lambda_1 \lambda_2' \neq \overline{\mu} \mu' + \mu \overline{\mu}' $, and
	\begin{equation}
		d_{21} \coloneqq - (\lambda_1 + \lambda_2) (1 - \lambda_1' \lambda_2' + | \mu'|^2) + (\lambda_1' + \lambda_2') (\overline{\mu} \mu' + \mu \overline{\mu}' - \lambda_2 \lambda_1' - \lambda_1 \lambda_2') \neq 0 \,, \label{eq:d21NNAgain}
	\end{equation}
	where we recall $ (\lambda_1, \lambda_2) = (\Sigma + \Delta, \Sigma - \Delta) $, then
	\begin{equation}
		B = \op{sgn} d_{21} \,.
	\end{equation} 
	Finally, if
	\begin{equation}
		\mu = \lambda_1 = \lambda_2 = 0 
	\end{equation}
	or
	\begin{equation}
		\lambda_2 = - \lambda_1 \quad \op{and} \quad ( \mu \neq 0 \ \op{or} \ \lambda_1 \neq 0 ) \,,
	\end{equation}
	then $ B = 0 $. In the cases that are not covered, $B$ is ill-defined.
\end{proposition}
\begin{remark}
	A remark analogous to Rem.\,\ref{rem:GenuineParsND} applies here, too. The integers depend at most on the orbit of the boundary condition, and thus on $ \lambda_i, \lambda_i' $, $ (i = 1,2) $, $ \mu $, $ \mu' $, cf.\,Eq.\,\ref{eq:UNN}; and actually do not change under the shifts mentioned in Rem.\,\ref{rem:ConsequenceTranslationKx}. One indeed verifies that if $ \Delta^2 = \cal{B} $, i.e.\,$ |\mu|^2 = \lambda_1 \lambda_2 $, then \eqref{eq:d21NNAgain} is left invariant by them. \label{rem:GenuineParsNN}
\end{remark}
\begin{remark}
	Since Prop.\,\ref{prop:NewMap4} is an extension of Prop.\,\ref{prop:Map4}, the latter can be recovered from the former by restriction to the particle-hole symmetric subset of family NN. In practice, this is achieved by imposing the constraints on $ \mu, \mu' $ and $ \lambda_i, \lambda_i' \ (i=1,2) $ detailed in Prop.\,\ref{prop:UBCs}, point (ii). The thresholds $ \cal{M}_\pm, \, \cal{I}_\pm, \, \cal{E}, \, \cal{B} $ are specialized to PHS, thus becoming
	\begin{equation}
		\cal{M}^\Xi_\pm = \Sigma^2 - | \mu |^2 \pm \nu \sqrt{2} \Sigma \,, \qquad \cal{I}^\Xi_\pm \equiv \cal{I}^\Xi = \Sigma^2 - | \mu |^2 - \nu^2 \,, \qquad \cal{E}^\Xi = \cal{E} \,, \qquad \cal{B}^\Xi = \cal{B} \,. \label{eq:PHSTransNN}
	\end{equation}
	Fig.\,\ref{fig:IndicesIV} is then recovered by plotting the \textit{transition surfaces}
	\begin{equation}
		\{ \Delta^2 = \cal{M}^\Xi_\pm, \, \Delta^2 = \cal{I}^\Xi , \, \Delta^2 = \cal{E}^\Xi , \, \Delta^2 = \cal{B}^\Xi \}
	\end{equation}
	and decorating each region, defined by such surfaces, with the values of $ (P,I,E,B) $ found in the proposition above.
\end{remark}

\section{Evaluating the integers $ P $, $I$, $E$, $B$} \label{sec:Derivation}

\subsection{Number of proper mergers $P$} \label{subsec:ComputingP}

This section and the upcoming Sects.\,\ref{subsec:ParabolicStates}, \ref{subsec:FlatStates} and \ref{subsec:BoundaryWinding} illustrate how to calculate each entry of the integer tuple $ \cal{V} $, given a choice $ \ul{A} $ of $ 2 \times 4 $ self-adjoint boundary condition, cf.\,\eqref{eq:ExplicitUlA}. Here, we focus on the number of proper mergers $P$, cf.\,Def.\,\ref{def:Integers}. In view of the scattering matrix $S$ playing the role of a transition function between incoming and outgoing sections of $ E_+ $ (cf.\,Prop.\,\ref{prop:BulkScatteringCorrespondence}), $P$ shall be computed as
\begin{equation}
	P = C_+ - W_\infty = 2 - W_\infty \,,
\end{equation}
cf.\,\eqref{eq:CPlusW0WInfty}, where
\begin{equation}
	W_\infty = \lim_{\epsilon \rightarrow 0} \frac{1}{2 \pi \im} \int_{\mathcal{C}_{\epsilon,K}^{\infty}} S^{-1} \mathrm{d} S \,,
\end{equation}
cf.\,\eqref{eq:W0WInfty}, represents the winding of $S$ along a contour $ \mathcal{C}_{\epsilon,K}^{\infty} $, entirely contained in a neighborhood of $ \{ \infty \} $ and approximating the $ \{ k_y = 0 \} $ meridian of the compactified momentum plane as $ \epsilon \to 0 $.

Computing the winding number $ W_\infty $, and in turn $P$, thus requires: Connecting a self-adjoint boundary condition $ \ul{A} $ to its associated scattering matrix $S$; Expanding the latter around $ \{ \infty \} $, i.e., the north pole of the compactified momentum plane; Extracting from that asymptotic behavior the winding $ W_\infty $ along $ \mathcal{C}_{\epsilon,K}^{\infty} $. The first two steps, in this very order, are detailed in the paragraphs below. The third point, by contrast, is deferred to Sect.\,\ref{sec:Proofs}. \vspace{1\baselineskip}

We start by reviewing \cite{GJT21} how $ S $ is obtained from a given $ 2 \times 4 $ self-adjoint boundary condition $ k_x \mapsto \ul{A} (k_x) $. Recall the definition \ref{def:ScattState} of a scattering state
\begin{equation}
	\psi_s = \tilde{\psi}_s (y; k_x, \kappa) \eul^{\im (k_x x - \omega t)}
\end{equation}
of energy $ \omega = \omega_+ (k_x, \kappa) $. The scattering and evanescent amplitudes $ S $ and $T$, cf.\,Eq.\,\eqref{eq:SDef}, allow the following rewriting
\begin{align}
	\tilde{\psi}_s &\equiv \tilde{\psi}_{\op{in}} + \tilde{\psi}_{\op{out}} + \tilde{\psi}_{\op{ev}} \nonumber \\ 
	&=\hat{\psi}_{\op{in}} (k_x, - \kappa) \eul^{- \im \kappa y} + S(k_x, \kappa) \hat{\psi}_{\op{out}} (k_x, \kappa) \eul^{\im \kappa y} + T(k_x, \kappa) \hat{\psi}_{\op{ev}} (k_x, \kappa_{\op{ev}}) \eul^{\im \kappa_{\op{ev}} y} \,, \label{eq:ExplicitScattState}
\end{align}
where sections $ \hat{\psi}_j $, $ (j = \op{in}, \, \op{out}, \, \op{ev}) $ of $ E_+ $, cf.\,\eqref{eq:BandBundle}, have been chosen, and in particular so that $ \hat{\psi}_{\op{in}} : S^2 \to \bbC^3 $ is normalized for all $ (k_x, \kappa) \in S^2 $. The expressions for $S,T$, but not their topological content, depend on the choice of sections.

The scattering amplitude $S$ is completely determined by imposing the boundary conditions $ \ul{A} $ on the scattering state $ \psi_s $:
\begin{equation}
	\ul{A} M \Psi_s = 0 \,, \qquad \Psi_s = 
	\begin{pmatrix}
		\tilde{\psi}_s \\
		\partial_y \tilde{\psi}_s
	\end{pmatrix} \Big\rvert_{y=0} \,, \label{eq:BCScattState}
\end{equation}
cf.\,(\ref{eq:BoundaryCondition}, \ref{eq:AulA}). The resulting form of $S$ is specified in the next lemma.
\begin{lemma}
	Given sections $ \hat{\psi}_j = (\eta_j, \, u_j, \, v_j) $, $ (j = \op{in}, \, \op{out}, \, \op{ev}) $ of $ E_+ $, cf.\,\eqref{eq:BandBundle}, the scattering amplitude $S$ takes the form
	\begin{equation}
		S = - \frac{\det (\ul{A} V_{\op{in}})}{\det (\ul{A} V_{\op{out}})} \,, \label{eq:SwithAV}
	\end{equation}
	with $ V_{\op{in/out}} = (V_1^{\op{in/out}}, V_2^{\op{in/out}})^T $ and
	\begin{equation}
		V_1^{\op{in/out}} =
		\begin{pmatrix}
			\eta_{\op{i/o}} \pm \im \nu \kappa u_{\op{i/o}} & \eta_{\op{e}} - \im \nu \kappa_{\op{ev}} u_{\op{e}} \\
			u_{\op{i/o}} & u_{\op{e}} 
		\end{pmatrix}\,, \qquad 
		V_2^{\op{in/out}} = 
		\begin{pmatrix}
			v_{\op{i/o}} & v_{\op{e}} \\
			\mp \im \nu \kappa v_{\op{i/o}} & \im \nu \kappa_{\op{ev}} v_{\op{e}}  
		\end{pmatrix}\,. \label{eq:VinVout}
	\end{equation}
\end{lemma}
\begin{proof}
	Let
	\begin{equation}
		\Psi_j \coloneqq 
		\begin{pmatrix}
			\tilde{\psi}_j \\
			\partial_y \tilde{\psi}_j
		\end{pmatrix} \Big\rvert_{y=0} \,, \qquad (j = \op{in}, \, \op{out}, \, \op{ev}) \,,
	\end{equation}
	namely
	\begin{equation}
		\Psi_{\op{in/out}} = 
		\begin{pmatrix}
			\hat{\psi}_{\op{in/out}} \\
			\mp \im \kappa \hat{\psi}_{\op{in/out}}
		\end{pmatrix} \,, \qquad 
		\Psi_{\op{ev}} = 
		\begin{pmatrix}
			\hat{\psi}_{\op{ev}} \\
			\im \kappa_{\op{ev}} \hat{\psi}_{\op{ev}}
		\end{pmatrix} \,.
	\end{equation}
	By linearity, Eq.\,\eqref{eq:BCScattState} is equivalent to
	\begin{equation}
		\ul{A} M \big ( \Psi_{\op{in}} + S \Psi_{\op{out}} + T \Psi_{\op{ev}} ) = 0 \,,
	\end{equation}
	which we rearrange as
	\begin{equation}
		(\ul{A} M \Psi_{\op{out}}, \, \ul{A} M \Psi_{\op{ev}})
		\begin{pmatrix}
			S \\
			T
		\end{pmatrix}
		=
		- \ul{A} M \Psi_{\op{in}} \,,
	\end{equation}
	where $ \ul{A} M \Psi_{j} \ (j = \op{in}, \, \op{out}, \, \op{ev}) $ are $ 2 \times 1 $ matrices. Then, by Cramer's rule
	\begin{equation}
		S = - \frac{\det ( \ul{A} M \Psi_{\op{in}}, \, \ul{A} M \Psi_{\op{ev}} )}{\det ( \ul{A} M \Psi_{\op{out}}, \, \ul{A} M \Psi_{\op{ev}} )} \,.
	\end{equation}
	The result follows by observing
	\begin{equation}
		V_{\op{in/out}} = M ( \Psi_{\op{in/out}}, \, \Psi_{\op{ev}} ) \,,
	\end{equation}
	with $M$ as in \eqref{eq:MatM}.
\end{proof}
We further specialize the choice of sections to
\begin{equation}
	\hat{\psi}_{\op{in}} (k_x, - \kappa) = \hat{\psi}^0 (k_x, -\kappa) \,, \quad \hat{\psi}_{\op{out}} (k_x, \kappa) = \hat{\psi}^0 (k_x, \kappa) \,, \quad \hat{\psi}_{\op{ev}} (k_x, \kappa_{\op{ev}}) = \hat{\psi}^\infty (k_x, \kappa_{\op{ev}}) \,, \label{eq:ChoiceOfSections}
\end{equation}
cf.\,\eqref{eq:PsiInfPsiZero}, so as to sit in the case $ \hat{\psi}_{\op{out}} (k_x, \kappa) = \hat{\psi}_{\op{in}} (s (k_x, \kappa)) $ where $S$ acts as a transition function of the bundle $ E_+ $ (see Sect.\,\ref{subsec:TopScatt} for details). Eq.\,\eqref{eq:SwithAV} then takes the form
\begin{equation}
	S (k_x, \kappa) = - \frac{g (k_x, -\kappa)}{g (k_x, \kappa)} \,, \label{eq:SWithG}
\end{equation}
where the \textit{Jost function} $ g (k_x, \kappa) $ is given by
\begin{equation}
	g (k_x, \kappa) \coloneqq \op{det} (\ul{A} V)\,, \label{eq:gFromV}
\end{equation}
with $ V = V_{\op{out}} $, cf.\,\eqref{eq:VinVout}, evaluated at
\begin{equation}
	l_{\op{o}} \coloneqq l^0 (k_x, \kappa) \,, \qquad l_{\op{e}} \coloneqq l^\infty (k_x, \kappa_{\op{ev}}) \,, \qquad (l = \eta, u, v) \,. \label{eq:DefEtaOEtaE}
\end{equation}
\begin{remark} \label{rem:ProportionalSections}
	Notice that sections \eqref{eq:ChoiceOfSections} can be replaced by non-zero multiples thereof without changing the zeros of the Jost function $g$, cf.\,\eqref{eq:gFromV}. Those zeros are in turn relevant because they determine edge states, cf.\,\eqref{eq:BoundState}.
\end{remark}

Only the first few orders of the expansion of $ g(k_x, \kappa) $ around $ (k_x, \kappa) \to \infty $ are relevant for the winding number $ W_\infty $. We perform that expansion by first bringing $ \{ \infty \} $ to $0$ via the (orientation preserving) coordinate change
\begin{equation}
	(k_x , \kappa) = \Big( \frac{\cos \varphi}{\varepsilon} , - \frac{\sin \varphi}{\varepsilon} \Big) \,, \label{eq:ChangeOfVariables}
\end{equation}
and then expanding around $ \varepsilon \to 0 $. The leading orders of the resulting series are reported in the next proposition, separately for the families DD, ND, NN of Def.\,\ref{def:Families}. As customary, case DN is not pursued, cf.\,comment right below Def.\,\ref{def:Families}.
\begin{proposition}
	\label{prop:GEps}
	Consider the families DD, ND and NN of Def.\,\ref{def:Families}, moreover parametrized as in Lm.\,\ref{lem:Dictionary} and Prop.\,\ref{prop:BCs}. The Jost function $ g (\varepsilon, \varphi) $, cf.\,Eq.\,\eqref{eq:gFromV}, has, for each family, the following expansion around $ \varepsilon = 0 $.
	\begin{equation}
		\label{eq:Expansions}
		\begin{array}{llcl}
			\op{DD:} & g (\varepsilon, \varphi) & = & 2 \im \left( a_1' + \frac{\cos \varphi}{\varepsilon} b_1 \right) \wedge \left( a_2' + \frac{\cos \varphi}{\varepsilon} b_2 \right) + o (1) \\
			\op{ND:} & g (\varepsilon, \varphi) & = & \varepsilon^{-2} \left( 2 \lambda \cos \varphi + \im \nu | 1 + \im \alpha |^2 \sin \varphi - \nu | 1 - \im \alpha |^2 \sqrt{1 + \cos^2 \varphi} + 2 \lambda' \varepsilon \right) \cdot \\
			& & & a_1 \wedge \left( \cos \varphi b_1 + \varepsilon a_1' \right) + o(1) \\
			\op{NN:} & g (\varepsilon, \varphi) & = & \varepsilon^{-2} (a_1 \wedge a_2) [ A_1 (\varepsilon, \varphi) B_2 (\varepsilon, \varphi) - A_2 (\varepsilon, \varphi) B_1 (\varepsilon, \varphi) ] + o(1) \,,
		\end{array} 
	\end{equation}
	where, in case NN,
	\begin{equation}
		\label{eq:A1A2B1B2}
		\begin{aligned}
			A_1 (\varepsilon, \varphi) &= \varepsilon (\mu' + \lambda_2') + \im \nu \sin \varphi + (\mu + \lambda_2) \cos \varphi \,,  \\
			A_2 (\varepsilon, \varphi) &= \varepsilon (\mu' - \lambda_2') + \nu \sqrt{1 + \cos^2 \varphi} + (\mu - \lambda_2) \cos \varphi \,,  \\
			B_1 (\varepsilon, \varphi) &= \im \varepsilon (\lambda_1' + \bar{\mu}') - \nu \sin \varphi + \im (\lambda_1 + \bar{\mu}) \cos \varphi \,,  \\
			B_2 (\varepsilon, \varphi) &= \im \varepsilon (\lambda_1' - \bar{\mu}') - \im \nu \sqrt{1 + \cos^2 \varphi} + \im (\lambda_1 - \bar{\mu}) \cos \varphi \,. 
		\end{aligned}
	\end{equation}
\end{proposition}
The proof of this statement is deferred to App.\,\ref{app:ExpLambda}.

As the last item of this section, we connect the winding of $ S $ along the contour $ \cal{C}_{\epsilon, K}^\infty $ appearing in \eqref{eq:W0WInfty} to that of $ g (\varepsilon, \varphi) $ along a circumference of fixed radius around the point at infinity $ \varepsilon = 0 $.
\begin{lemma}
	\label{cl:NvsW}
	Let $ \ul{\lambda} \coloneqq (\lambda_x, \lambda_y) \ni \bbR^2_{\ul{\lambda}}$ be the reciprocal momentum variables, and let $ \tau: \bbR^2_{\ul{\lambda}} \to \bbR^2_{\ul{k}} $ be the coordinate change
	\begin{equation}
		\tau: (\lambda_x, \lambda_y) \mapsto \Big( \frac{\lambda_x}{\lambda^2} , \, - \frac{\lambda_y}{\lambda^2} \Big) = (k_x, k_y) \,, \qquad \lambda = \sqrt{\lambda_x^2 + \lambda_y^2} \,.
	\end{equation}
	Then, there exists an oriented contour
	\begin{equation}
		\gamma_R = \gamma_R ([0, 2 \pi]) \,, \qquad \gamma_R (\varphi) = (R \cos \varphi, R \sin \varphi) \in \bbR^2_{\ul{\lambda}}
	\end{equation}
	such that
	\begin{equation}
		W_\infty = \lim_{\epsilon \rightarrow 0} \frac{1}{2 \pi \im} \int_{\mathcal{C}_{\epsilon,K}^{\infty}} S^{-1} \mathrm{d} S = \lim_{R \rightarrow 0} \frac{1}{2 \pi \im} \int_{\gamma_R} g^{-1} \mathrm{d} g \,, \label{eq:WInftyEquality}
	\end{equation}
	where $g$ is the Jost function of (\ref{eq:SWithG}, \ref{eq:gFromV}).
\end{lemma}
\begin{proof}
	We split the proof into two parts. First, we claim and prove that the winding of $ S $ along a certain contour $ \gamma_\epsilon \subset \bbR^2_{\ul{\lambda}} $ equals the winding of $g$ along $ \gamma_R $. Then, we show that $ \gamma_\epsilon $ and the original contour $\mathcal{C}_{\epsilon,K}^{\infty}$ achieve the same limit as $ \epsilon \to 0 $.
	
	Consider, for some $ \epsilon, T >0 $, the oriented contour
	\begin{equation}
		\gamma_\epsilon = \gamma_{\epsilon} ([-T,T]) \,, \qquad \gamma_\epsilon (t) = (-t, -\epsilon) \in \bbR^2_{\ul{\lambda}} \,,
	\end{equation}
	corresponding to a straight line, oriented right to left, in the $ \lambda_y < 0 $ half-plane of Fig.\,\ref{fig:WInftyEqMinN}. The winding of $S$ along $ \gamma_\epsilon $ is given by
	\begin{equation}
		\frac{1}{2 \pi \im} \int_{\gamma_\epsilon} \mathrm{d} \log S \,.
	\end{equation}
	By \eqref{eq:SWithG}, we observe that
	\begin{equation}
		\diff \log S (\ul{\lambda}) = \diff \log g (\lambda_x, -\lambda_y) -\diff \log g (\lambda_x, \lambda_y) \,.
	\end{equation}
	Then
	\begin{align}
		\frac{1}{2 \pi \im} \int_{\gamma_\epsilon} \mathrm{d} \log S &= \frac{1}{2 \pi \im} \Big( \int_{\gamma_\epsilon} \diff \log g (\lambda_x, -\lambda_y) - \int_{\gamma_\epsilon} \diff \log g (\lambda_x, \lambda_y) \Big) \nonumber \\
		&= \frac{1}{2 \pi \im} \Big( \int_{\gamma_{-\epsilon}} \diff \log g (\lambda_x, \lambda_y) - \int_{\gamma_\epsilon} \diff \log g (\lambda_x, \lambda_y) \Big) \nonumber \\
		&= \frac{1}{2 \pi \im} \int_{\gamma_{-\epsilon} - \gamma_{\epsilon}} \diff \log g \,.
	\end{align}
	As $ \epsilon \to 0 $, the difference of contours $ \gamma_{-\epsilon} - \gamma_{\epsilon} $ becomes topologically equivalent to a positively oriented circle $ \gamma_R $ (of small radius $ R > 0 $) around the origin of $ \bbR^2_{\ul{\lambda}} $, cf.\,Fig.\,\ref{fig:WInftyEqMinN}. Thus,
	\begin{equation}
		W_\infty = \frac{1}{2 \pi \im} \int_{\gamma_\epsilon} \mathrm{d} \log S = \frac{1}{2 \pi \im} \int_{\gamma_R} \diff \log g \,,
	\end{equation}
	which is the first claim.
	\begin{figure}[hbt]
		\centering
		\includegraphics[width=0.65\linewidth]{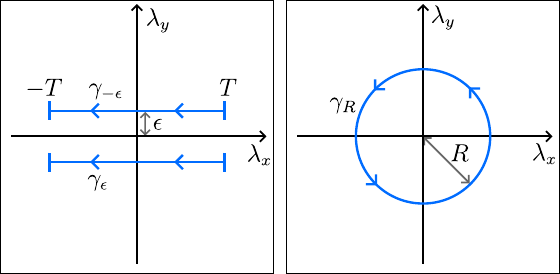}
		\caption{Winding contours at infinity. On the left, the $ \gamma_{\pm \epsilon} $ straight contours. On the right, the topologically equivalent, positively oriented circle $ \gamma_R $.} 
		\label{fig:WInftyEqMinN}
	\end{figure}
	
	For the second part of the proof, we recall that $ \mathcal{C}_{\epsilon,K}^{\infty} $ can be any contour, oriented in the direction of increasing $ k_x $, which approximates the infinite portion of the $ \{ k_y = 0 \} $ \enquote{meridian} as $ \epsilon \to 0 $:
	\begin{equation}
		\lim_{\epsilon \to 0} \mathcal{C}_{\epsilon,K}^{\infty} = \{ |k_x| > K \} \times \{ k_y = 0 \} \subset \bbR^2_{\ul{k}} \,.
	\end{equation}
	The candidate contour $ \gamma_\epsilon $ is thus suitable for the computation of $ W_\infty $ if
	\begin{equation}
		\lim_{\epsilon \to 0} \tau \circ \gamma_\epsilon ([-T,T]) = \{ |k_x| > K \} \times \{ k_y = 0 \} \,.
	\end{equation}
	However, 
	\begin{equation}
		\lim_{\epsilon \to 0} \tau \circ \gamma_\epsilon (t) = \lim_{\epsilon \to 0} \Big( \frac{-t}{t^2 + \epsilon^2} , \, \frac{\epsilon}{t^2 + \epsilon^2} \Big) = \Big( -\frac{1}{t}, 0 \Big) \,,
	\end{equation}
	whence the desired limit is achieved by picking $ T = K^{-1} $.
\end{proof}

\subsection{Number of improper mergers $I$} \label{subsec:ParabolicStates}

In this section, we show how to compute $I$, the number of improper mergers, given the boundary conditions. At first, we notice that $I$ is tantamount to the signed number of asymptotically parabolic edge eigenvalues. They correspond to poles of the scattering matrix $S$ in a region $ \omega \gtrsim k_x^2 $, $ k_x \to \infty $ of the $ (k_x, \omega) $-plane. There, both evanescent momenta $ (\kappa, \kappa_{\op{ev}}) $ diverge, and it pays to replace them by a point $ ( k_{y+}, k_{y-} ) $ on a Riemann surface. Working on this surface, we express poles of $S$ as zeros of $ G (k_x, k_{y\pm}) $, a close relative of the Jost function $ g $, cf.\,\eqref{eq:SWithG}. The required parabolic behavior of the associated edge states imposes the Ansatz $ k_{y\pm} = \pm \im c_\pm k_x $, $ c_\pm \in \bbR $, so that $ G(k_x, k_{y\pm} = \pm \im c_\pm k_x) \equiv G(c_\pm) = 0 $ defines a curve in the plane $ (c_+, c_-) $, whose type depends on the families DD, ND, NN of the boundary conditions. Further constraints restrict eigenvalues to lie on certain arcs $ \cal{A}_\uparrow $, $ \cal{A}_\downarrow $ of a circle in the $ (c_+,c_-) $-plane, and the value of $I$ is thus geometrically found as the signed number of intersections between the curve and the arcs. \vspace{1\baselineskip}

Above Def.\,\ref{def:Integers}, $I$ was defined as the signed number of edge eigenvalues diverging for $ k_x \to \infty $. Such eigenvalues are asymptotically quadratic in $ k_x \to \infty $: Faster divergences would make the eigenvalue hit the positive bulk band at finite $ k_x $, whereas slower ones would imply that the eigenvalue and band edge never actually meet. The problem of finding $I$ thus reduces to identifying all \textit{asymptotically parabolic edge eigenvalues} $ \omega (k_x) $, and counting them with sign. More precisely, this section considers a region
\begin{equation}
	\omega \gtrsim k_x^2 \label{eq:RegionI}
\end{equation}
of the $ (k_x, \omega) $-plane, meaning $ \omega \geq c k_x^2 $ for any fixed $ c > 0 $. Other regions, for example $ \omega \ll k_x^2 $ will be considered later.

As discussed above Eq.\,\eqref{eq:BoundState}, the edge states of interest shall be found as solutions, in $ \kappa $, of
\begin{equation}
	S^{-1} (k_x, \kappa) = 0 \label{eq:Poles}
\end{equation}
at $ |k_x| \to \infty $, moreover with energy
\begin{equation}
	\omega = \omega_+ (k_x, \kappa) \label{eq:OmegaEqOmegaPlus}
\end{equation}
lying in the desired region \eqref{eq:RegionI} and $ \op{Im} (\kappa) > 0 $. Those edge states involve two evanescent states of rates $ \kappa, \kappa_{\op{ev}} $ that determine each other, cf.\,\eqref{eq:KappaEv}. In the regimes \eqref{eq:RegionI} they will both diverge and it pays to treat them on the same footing.

Fix $ \omega $ and ignore $ k_x $ for the time being. Let $ X_\pm $ be the two solutions, in $X$, of
\begin{equation}
	\omega^2 = X + (f - \nu X)^2 \,, \label{eq:BranchFreeOmega}
\end{equation}
namely
\begin{equation}
	X_{\pm} = \frac{- (1 - 2 \nu f) \pm \sqrt{\Delta}}{2 \nu^2} \,, \qquad (\Delta = 1 - 4 \nu f + 4 \omega^2 \nu^2) \,, \label{eq:XPlusMinusExplicit}
\end{equation}
where we stress that $ \Delta > 0 $ by $ 4 \nu f < 1 $ (cf.\,comment below Eq.\,\eqref{eq:SWM}). The solutions are related by Vieta's formulae
\begin{gather}
	X_+ + X_- = - \frac{1 - 2 \nu f}{\nu^2} \,, \label{eq:FirstVieta} \\  
	X_+ X_- = \frac{f^2 - \omega^2}{\nu^2} \,. \label{eq:SecondVieta}
\end{gather}
A convenient way to view these equations is as follows. Eq.\,\eqref{eq:BranchFreeOmega} determines a map $ \omega \mapsto (X_+, X_-) $ with $ X_+ > X_- $ for $ \omega $ real. For $ \omega \to \infty $, we moreover have
\begin{equation}
	X_\pm = \pm \frac{\omega}{\nu} + \cal{O} (\omega^{-1}) \,. \label{eq:XOmega}
\end{equation}
Conversely, pairs $ (X_+, X_-) $ are constrained by \eqref{eq:FirstVieta}, one of Vieta's formulae, and $ \omega $ is then determined by the other. By extension, $ (k_x, \omega) $ determine $ (k_{y+}, k_{y-}) $ by
\begin{equation}
	k_x^2 + k_{y \pm}^2 = X_\pm \,. \label{eq:DefKYPlusMinus}
\end{equation}
We thus find
\begin{equation}
	k_{y+} = \sqrt{X_+ - k_x^2} \eqqcolon \kappa \,, \qquad k_{y-} = \im \sqrt{k_x^2 - X_-} = \kappa_{\op{ev}} (k_x, \kappa) \,, \label{eq:AnalyticContinuation}
\end{equation}
up to signs. Notice that $ \kappa \in \im \bbR $ iff $ X_+ < k_x^2 $, a condition equivalent to $ \omega $ sitting in a spectral gap: $ \omega^2 < \big( \omega^+ (k_x) \big)^2 $, cf.\,\eqref{eq:BandRim}. By contrast, $ \kappa_{\op{ev}} \in \im \bbR_+ $ is granted, even outside of the spectral gap, by $ X_- < 0 $, $ (\omega \in \bbR) $, cf.\,\eqref{eq:XPlusMinusExplicit}. Moreover,
\begin{equation}
	k_{y+}^2 - k_{y-}^2 = X_+ - X_- = \frac{2 \omega}{\nu} + \cal{O} (\omega^{-1}) \,. \label{eq:KyPMDiff}
\end{equation}
Again and conversely, we consider points $ (k_x, k_{y+}, k_{y-}) $, where $ k_{y\pm} $ are constrained by (\ref{eq:DefKYPlusMinus}, \ref{eq:FirstVieta}) and $ \omega $ is determined by \eqref{eq:SecondVieta}. Eq.\,\eqref{eq:DefKYPlusMinus} allows for two branches of $ k_{y\pm} $ each. At first we shall not distinguish between them, thus viewing $ (k_{y+}, k_{y-}) $ as a point on a Riemann surface depending on $ k_x $. 

We now consider functions on that surface, such as the scattering amplitude $ S (k_x, k_{y+}, k_{y-}) $ $ \equiv S (k_x, k_{y+}) $ whose poles we wish to identify, cf.\,\eqref{eq:Poles}. We moreover discuss such zeros only in the generic case where the index $B$ is well-defined. That entails that all the fibers $ H^\# (k_x) $ are self-adjoint one by one, cf.\,Rem.\,\ref{rem:LossSA}, and actually including at $ k_x = \pm \infty $ in a natural sense. In turn, that means that the Jost functions $g$, related to $S$ by \eqref{eq:SWithG}, have an easily identifiable asymptotics at $ k_x \to \pm \infty $. It allows to simplify the quotient
\begin{equation}
	S = - g (k_x, -k_{y+}) / g (k_x, k_{y+}) = - G (k_x, -k_{y+}) / G (k_x, k_{y+}) \,, \label{eq:Simplification}
\end{equation}
in favor of simpler expressions $ G(k_x, k_{y \pm}) $, that are found, for families DD-ND-NN, as follows. The expansions of Prop.\,\ref{prop:GEps} are still formally valid, upon replacement $ (\kappa, \kappa_{\op{ev}}) \rightsquigarrow (k_{y+}, k_{y-}) $, because $ |k_x| \to \infty $ was there explored along a parabola $ \cal{C}_\epsilon $, cf.\,central panel of Fig.\,\ref{fig:Levinson}, that lies within the region $ \omega \gtrsim k_x^2 $, here under study. However, there $ k_{y+} = k_{y+} (\omega) \in \bbR $ because $ \omega $ was in the bulk band; here, $ k_{y+} \in \im \bbR $ because $ \omega $ is chosen in the gap. The desired $G$'s are obtained, family by family, from those expansions by retaining the leading orders only and performing the simplifications alluded to in \eqref{eq:Simplification}.
\begin{proposition} \label{prop:Varieties}
	Consider the families DD, ND and NN of Def.\,\ref{def:Families}, moreover parametrized as in Lm.\,\ref{lem:Dictionary} and Prop.\,\ref{prop:BCs}. Poles of the scattering amplitude $S$, determined by the boundary conditions through (\ref{eq:SWithG}, \ref{eq:gFromV}), in region \eqref{eq:RegionI} as $ |k_x| \to \infty $ are encoded by $ G (k_x, k_{y+}) = 0 $, with $G$ as listed below, separately for each family.
	\begin{equation}
		\label{eq:Varieties}
		\begin{array}{llcl}
			\op{DD:} & G (k_x, k_{y+}) & = & 1 \,, \\
			\op{ND:} & G (k_x, k_{y+}) & = & 2 \lambda k_x - \im \nu | 1 + \im \alpha |^2 k_{y+} + \im \nu | 1 - \im \alpha |^2 k_{y-} \,, \\
			\op{NN:} & G (k_x, k_{y+}) & = & 2 (|\mu|^2 - \lambda_1 \lambda_2) k_x^2 - \im \nu (2 \mu_R - \lambda_1 - \lambda_2) k_x k_{y+} \\
			& & & - \im \nu (2 \mu_R + \lambda_1 + \lambda_2) k_x k_{y-} - 2 \nu^2 k_{y+} k_{y-} \,,
		\end{array}
	\end{equation}
	where $ \mu_R \coloneqq \op{Re} \mu $.
\end{proposition}
\begin{proof}
	We start by retaining the leading order in $ \varepsilon $ of the expansions \eqref{eq:Expansions}:
	\begin{equation}
		\label{eq:LeadingOrders}
		\begin{array}{llcl}
			\op{DD:} & g (\varepsilon, \varphi) & \simeq & 2 \im \left( \frac{\cos \varphi}{\varepsilon} b_1 \right) \wedge \left( \frac{\cos \varphi}{\varepsilon} b_2 \right) \,, \\
			\op{ND:} & g (\varepsilon, \varphi) & \simeq & a_1 \wedge \left( \frac{\cos \varphi}{\varepsilon} b_1 \right) \left( 2 \lambda \frac{\cos \varphi}{\varepsilon} + \im \nu | 1 + \im \alpha |^2 \frac{\sin \varphi}{\varepsilon} - \nu | 1 - \im \alpha |^2 \frac{\sqrt{1 + \cos^2 \varphi}}{\varepsilon} \right) \,, \\
			\op{NN:} & g (\varepsilon, \varphi) & \simeq & (a_1 \wedge a_2) \Big[ \left( \im \nu \frac{\sin \varphi}{\varepsilon} + (\mu + \lambda_2) \frac{\cos \varphi}{\varepsilon} \right) \left( - \im \nu \frac{\sqrt{1 + \cos^2 \varphi}}{\varepsilon} + (\lambda_1 - \bar{\mu}) \frac{\cos \varphi}{\varepsilon} \right) \\
			& & & - \left( \nu \frac{\sqrt{1 + \cos^2 \varphi}}{\varepsilon} + (\mu - \lambda_2) \frac{\cos \varphi}{\varepsilon} \right) \left( - \nu \frac{\sin \varphi}{\varepsilon} + \im (\lambda_1 + \bar{\mu}) \frac{\cos \varphi}{\varepsilon} \right) \Big] \,.
		\end{array} 
	\end{equation}
	We pass to a notation in terms of $ k_{y \pm} $ by
	\begin{equation}
		\frac{\cos \varphi}{\varepsilon} \to k_x \,, \qquad - \frac{\sin \varphi}{\varepsilon} \to k_{y+} \,, \qquad \im \frac{\sqrt{1 + \cos^2 \varphi}}{\varepsilon} \to k_{y-} \,.
	\end{equation}
	This results in
	\begin{equation}
		\label{eq:BranchFreeLeadOrd}
		\begin{array}{llcl}
			\op{DD:} & g (k_x, \pm k_{y+}) & \simeq & 2 \im k_x^2 ( b_1 \wedge b_2 ) \,, \\
			\op{ND:} & g (k_x, \pm k_{y+}) & \simeq & k_x (a_1 \wedge b_1) \big( 2 \lambda k_x \mp \im \nu | 1 + \im \alpha |^2 k_{y+} + \im \nu | 1 - \im \alpha |^2 k_{y-} \big) \,, \\
			\op{NN:} & g (k_x, \pm k_{y+}) & \simeq & - \im (a_1 \wedge a_2) \big( 2 (|\mu|^2 - \lambda_1 \lambda_2) k_x^2 \mp \im \nu (2 \mu_R - \lambda_1 -\lambda_2) k_x k_{y+} \\
			& & & + \im \nu (2 \mu_R + \lambda_1 + \lambda_2) k_x k_{y-} \mp 2 \nu^2 k_{y+} k_{y-} \big) \,.
		\end{array} 
	\end{equation}
	By the assumption of self-adjointness at infinity, equivalently $B$ well-defined there, the prefactors $ 2 \im k_x^2 ( b_1 \wedge b_2 ) $, $ k_x (a_1 \wedge b_1) $, $ - \im (a_1 \wedge a_2) $ are non-zero in family DD, ND and NN, respectively. They can thus be simplified in the spirit of \eqref{eq:Simplification}, leading to the expressions of $G$ listed in \eqref{eq:Varieties}.
\end{proof}
For all families DD-ND-NN, $ G $ is homogeneous in $ k_x, \ k_{y+}, \ k_{y-} $. This prompts the Ansatz
\begin{equation}
	k_{y \pm} = \pm \im c_{\pm} k_x \,, \label{eq:Ansatz}
\end{equation}
with $ c_\pm \in \bbR $, for solutions $ k_{y \pm} $ of $ G = 0 $ at $ |k_x| \to \infty $. The parameters $ c_\pm $ are constrained by consistency with (\ref{eq:BranchFreeOmega}, \ref{eq:DefKYPlusMinus}). In particular, solutions $ X_\pm $ of \eqref{eq:BranchFreeOmega} are related to one another by the first Vieta formula, cf.\eqref{eq:FirstVieta},
\begin{equation}
	X_+ + X_- = - \frac{1 - 2 \nu f}{\nu^2} \ \longleftrightarrow \ 2 k_x^2 + k_{y+}^2 + k_{y-}^2 = - \frac{1 - 2 \nu f}{\nu^2} \,.
\end{equation}
At $ |k_x| \to \infty $, the above and our Ansatz imply
\begin{equation}
	c_+^2 + c_-^2 = 2 \,. \label{eq:ConstraintCircle}
\end{equation}
Moreover, at $ \omega \to + \infty $, Eq.\,\eqref{eq:XOmega} is rewritten, through \eqref{eq:DefKYPlusMinus} and in terms of our Ansatz, as
\begin{equation}
	k_x^2 (1 - c_\pm^2 ) = \pm \frac{\omega}{\nu} \,,
\end{equation}
from which the further constraint 
\begin{equation}
	| c_\pm | \lessgtr 1 \,. \label{eq:ConstraintRectangle}
\end{equation}
Finally, the branches corresponding to bound states have $ \op{Im} (k_{y \pm}) > 0 $, namely $ c_+ > 0 $, $ c_- < 0 $ ($ c_+ < 0 $, $ c_- > 0 $) at $ k_x \to + \infty $ ($ k_x \to - \infty $). All in all, asymptotically parabolic edge eigenvalues at $ k_x \to + \infty $, or $ k_x \to - \infty $, are found as points $ (c_+, c_-) \in \bbR^2 $ belonging to the algebraic variety
\begin{equation}
	\{ (c_+, c_-) \ | \ G(k_x, k_{y\pm} = \pm \im c_\pm k_x) = 0 \} \,,
\end{equation}
and additionally satisfying the constraints
\begin{align}
	(c_+, c_-) \in \cal{A}_{\downarrow} &= \{ (c_+, c_-) \ | \ c_+^2 + c_-^2 = 2 \,, \ |c_\pm| \lessgtr 1  \,, \ c_\pm \gtrless 0 \} \nonumber \\
	&\equiv \{ (c_+, c_-) = \sqrt{2} ( \cos \theta, \sin \theta ) \ | \ \theta \in [- \pi/2, - \pi / 4] \} \label{eq:UpperArc}
\end{align}
for $ k_x \to + \infty $, or
\begin{align}
	(c_+, c_-) \in \cal{A}_{\uparrow} &= \{ (c_+, c_-) \ | \ c_+^2 + c_-^2 = 2 \,, \ |c_\pm| \lessgtr 1  \,, \ c_\pm \lessgtr 0 \} \nonumber \\
	&\equiv \{ (c_+, c_-) = \sqrt{2} ( \cos \theta, \sin \theta ) \ | \ \theta \in [\pi/2, 3 \pi / 4] \} \label{eq:LowerArc}
\end{align}
for $ k_x \to - \infty $.
\begin{figure}[hbt]
	\centering
	\includegraphics[width=0.33\linewidth]{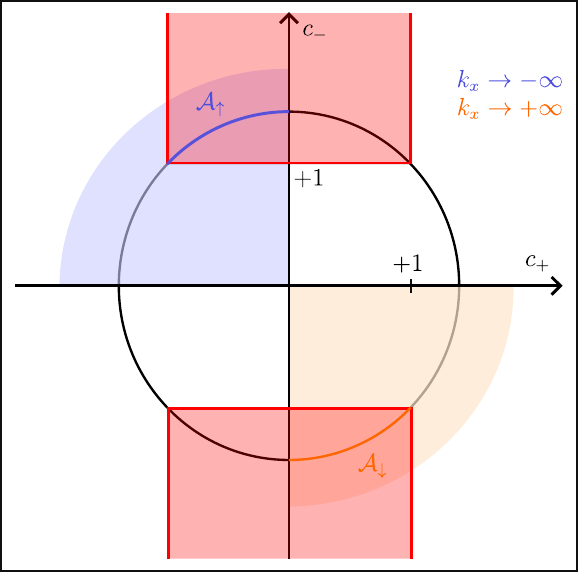}
	\caption{Graphical representation of the constraints of Eqs.\,(\ref{eq:ConstraintCircle}, \ref{eq:ConstraintRectangle}, \ref{eq:UpperArc}, \ref{eq:LowerArc}).}
	\label{fig:DetectionI}
\end{figure}

Solutions of $ G(k_x, k_{y\pm} = \pm \im c_\pm k_x) = 0 $ correspond to curves in the $ (c_+, c_-) $-plane, differing between families DD-ND-NN as detailed in the lemma below.
\begin{lemma} \label{lem:Curves}
	Consider the families DD-ND-NN of boundary conditions, cf.\,Def.\,\ref{def:Families}, parametrized as in Lm.\,\ref{lem:Dictionary} and Prop.\,\ref{prop:BCs}. For each family, let $ G (k_x, k_{y+}) $ be as in Prop.\,\ref{prop:Varieties}. Let moreover 
	\begin{equation}
		G (k_x, c_\pm) \coloneqq G ( k_x, k_{y\pm} = \pm \im c_\pm k_x ) \,.
	\end{equation}
	Then, solutions $ (c_+, c_-) $ of $ G (k_x, c_\pm) = 0 $ for $|k_x| \to \infty$ belong to:
	\begin{enumerate}
		\item[DD:] The empty set $ \varnothing $;
		
		\item[ND:] The line
		\begin{equation}
			c_- = - \frac{| 1 + \im \alpha |^2}{| 1 - \im \alpha |^2} c_+ - \frac{2 \lambda}{\nu |1 - \im \alpha|^2} = m c_+ + q \,, \label{eq:CurveND}
		\end{equation}
		with $ m,q $ as in \eqref{eq:MAndQ};
		
		\item[NN:] The hyperbola
		\begin{equation}
			c_- = \frac{2 (|\mu|^2 - \lambda_1 \lambda_2) + \nu (2 \mu_R - \lambda_1 - \lambda_2) c_+}{\nu (2 \mu_R + \lambda_1 + \lambda_2) + 2 \nu^2 c_+} \,, \label{eq:CurveNN}
		\end{equation}
		where again $ \mu_R = \op{Re} \mu $.
	\end{enumerate} 
\end{lemma}
\begin{proof}
	The claims follow by substitution of \eqref{eq:Ansatz} into \eqref{eq:Varieties} and elementary algebra.
\end{proof}
The task of computing the number $I$ of \textit{improper mergers} was thus reduced to counting the intersections (with sign) of the curves $ c_- (c_+) $ of Lm.\,\ref{lem:Curves} with the arcs $ \cal{A}_{\uparrow}, \, \cal{A}_{\downarrow} $. By the conventions illustrated in the third and fourth panel of Fig.\,\ref{fig:mergers}, intersections with $ \cal{A}_{\uparrow} $ ($ \cal{A}_{\downarrow} $), corresponding to parabolic bound states at $ k_x \to - \infty $ ($ k_x \to + \infty $), contribute $ +1 $ ($ -1 $) to the total number $I$ of \textit{improper mergers}. In summary,
\begin{equation}
	I = ( \text{\# of intersections of} \ c_-(c_+) \ \text{with} \ \cal{A}_{\downarrow} ) - ( \text{\# of intersections of} \ c_-(c_+) \ \text{with} \ \cal{A}_{\uparrow} ) \,. \label{eq:GeometricI}
\end{equation}
The value of $I$ in the families DD-ND-NN will be computed in Sect.\,\ref{sec:Proofs}, making use of \eqref{eq:GeometricI}.

\subsection{Number of escapes $E$} \label{subsec:FlatStates}

This section is similar to the previous one, in that we find edge states at $ |k_x| \to \infty $ as zeros of the Jost function $g$. It however differs from it, because we require $ |\omega| = |\omega (k_x) | \ll k_x^2 $, $ |k_x| \to \infty $. Since this regime has very small overlap with that of Sect.\,\ref{subsec:ParabolicStates}, cf.\,\eqref{eq:RegionI}, the expansions thereby found for $g$ are generally not valid. With no assumptions other than $ |\omega| \ll k_x^2 $, whence also $ X \ll k_x^2 $ by \eqref{eq:BranchFreeOmega}, we expand $ g $ around $ \delta \coloneqq X/k_x^2 \to 0 $, and study the zeros of its leading order for different behaviors of $ |\omega| $. If $ |\omega| \to 0, \infty $ as $ |k_x| \to \infty $, $ g=0 $ is true only exceptionally in the space of boundary conditions. By contrast, asymptotically flat solutions $ \omega \to \omega_{\op{a}, \pm} $ at $ k_x \to \pm \infty $ are typical. The height $ \omega_{\op{a}, \pm} $ of those asymptotes is separately computed and listed for families DD, ND and NN. \vspace{1\baselineskip}

Like in the previous section, we wish to find solutions $ \kappa $ with $ \op{Im} (\kappa) > 0 $ of
\begin{equation}
	S^{-1} (k_x, \kappa) = 0
\end{equation}
at $ |k_x| \to \infty $, corresponding to eigenstates of energy
\begin{equation}
	\omega = \omega_+ (k_x, \kappa) \,.
\end{equation}
Moreover, we again see $ S \equiv S (k_x, k_{y+}) $ as a function on the Riemann surface of $ (k_{y+}, k_{y-}) $, cf.\,\eqref{eq:DefKYPlusMinus}. Unlike before, however, we focus on edge eigenvalues $ \omega \equiv \omega (k_x) $ such that
\begin{equation}
	| \omega | \ll k_x^2 \,, \qquad (|k_x| \to \infty) \,. \label{eq:RegionE}
\end{equation}
This region only slightly overlaps \eqref{eq:RegionI}. The union of the two is moreover the whole $ (k_x, \omega) $-plane, ensuring that we explore all of the gap. Again in contrast with the previous section, we fix the branches of $ k_{y\pm} $ that grant $ \op{Im} (k_{y \pm}) > 0 $ from the very beginning:
\begin{equation}
	k_{y \pm} = \im |k_x| \sqrt{1 - \frac{X_{\pm}}{k_x^2}} \,.
\end{equation}
Consider $ X $ as in \eqref{eq:BranchFreeOmega}. By \eqref{eq:RegionE}, we have
\begin{equation}
	\delta \coloneqq \frac{X}{k_x^2} \to 0 \,, \qquad (|k_x| \to \infty) \,,
\end{equation}
whence
\begin{equation}
	k_{y \pm} = \im |k_x| \sqrt{1 - \frac{X_{\pm}}{k_x^2}} = \im |k_x| \Big( 1 - \frac{X_{\pm}}{2 k_x^2} + \cal{O} (\delta^2) \Big) \,. \label{eq:KyPMFlat}
\end{equation}
Since the regions \eqref{eq:RegionI} and \eqref{eq:RegionE} have a large symmetric difference, we cannot rely on the expansions of Sect.\,\ref{subsec:ParabolicStates} to determine poles of $S$ at $ |k_x| \to \infty $, corresponding to edge states. We thus proceed as follows.

Eq.\,\eqref{eq:SWithG} allows looking at zeros $ g (k_x, k_{y+}) = 0 $ of the Jost function $g$, rather than at poles of $S$. The condition \eqref{eq:RegionE} on $ \omega $ is enforced by expanding $ g (k_x, k_{y+} (\delta)) $ around $ \delta = 0 $, and retaining the leading order. Within it, we then look for zeros at $ |k_x| \to \infty $. These turn out to change depending on the  asymptotic behavior (in $ k_x $) of $ \omega $, and we thus find it advantageous to read $ g (k_x, k_{y+} (\omega)) \equiv g (k_x, \omega) = 0 $ as an implicit equation for $ \omega $ rather than $ k_{y+} $, cf.\,discussion below Eq.\,\eqref{eq:KyPMDiff}. If $ |\omega| \to 0, \infty $, $ g = 0 $ has solutions only exceptionally in the space of boundary conditions. By contrast, if $ \omega \to \omega_{\op{a}, \pm} < \infty $ at $ k_x \to \pm \infty $, solutions of $ g (k_x, \omega) =0 $ exist almost everywhere in the space of BCs. These claims are formalized in the next proposition, where we also provide the values $ \omega_{\op{a}, \pm} $ of the horizontal asymptotes, separately for each of the families DD, ND and NN.
\begin{proposition} \label{prop:EscapeHeight}
	Consider the families of boundary conditions DD, ND and NN (cf.\,Def.\,\ref{def:Families}), parametrized as in Lm.\,\ref{lem:Dictionary} and Prop.\,\ref{prop:BCs}. At $ |k_x| \to \infty $, edge eigenvalues with dispersion $ |\omega| \ll k_x^2 $ exist iff one of the following holds.
	\begin{itemize}
		\item DD. $ \omega \to \omega_{\op{a}, \pm} < \infty $ at $k_x \to \pm \infty$ with
		\begin{equation}
			\omega_{\op{a}, \pm} = \pm \frac{1}{2 \nu} \,. \label{eq:AsymptoteDD}
		\end{equation}
		
		\item ND. $ |\omega| \to \infty $ at $k_x \to \pm \infty$ and
		\begin{equation}
			\lambda \mp 2 \nu \alpha_I = 0 \ \longleftrightarrow \ q = \pm (m+1) \,; \label{eq:SubparInftyND}
		\end{equation}
		or $ |\omega| \to 0 $ at $ |k_x| \to \infty $ and
		\begin{equation}
			\lambda - \nu (1 + |\alpha|^2) = 0 \ \longleftrightarrow \ q = m-1 \,; \label{eq:ZeroInftyND}
		\end{equation}
		or $ \omega \to \omega_{\op{a}, \pm} < \infty $ at $ k_x \to \pm \infty $ with
		\begin{equation}
			\omega_{\op{a}, \pm} = \pm \frac{\lambda - \nu (1 + | \alpha |^2)}{2 \nu ( \lambda \mp 2 \nu \alpha_I )} \,. \label{eq:AsymptoteDN}
		\end{equation}
		The parameters $ m, q $ are as in \eqref{eq:MAndQ}.
		
		\item NN. $ |\omega| \to \infty $ at $k_x \to \pm \infty$ and
		\begin{equation}
			\Delta^2 = \cal{I}_\mp \,; \label{eq:SubparInftyNN}
		\end{equation}
		or $ |\omega| \to 0 $ at $ |k_x| \to \infty $ and
		\begin{equation}
			\Delta^2 = \cal{E} \,; \label{eq:ZeroInftyNN}
		\end{equation}
		or $ \omega \to \omega_{\op{a}, \pm} < \infty $ at $ k_x \to \pm \infty $ with
		\begin{equation}
			\omega_{\op{a}, \pm} = \frac{|\mu|^2 - \nu^2 + \nu (\lambda_1 + \lambda_2) - \lambda_1 \lambda_2}{4 \nu^2 \mu_R \pm 2 \nu ( |\mu|^2 + \nu^2 - \lambda_1 \lambda_2 )} \,. \label{eq:AsymptoteNN}
		\end{equation}
		See (\ref{eq:SigmaAndDelta}, \ref{eq:IPlusMinus}, \ref{eq:EPlusMinus}) for the definitions of $ \Delta $, $ \cal{I}_\pm $, $ \cal{E} $.
	\end{itemize}
	The asymptotically flat states contribute to $E$ provided $ \omega_{\op{a}, \pm} > 0 $.
\end{proposition}
In the interest of readability, the proof of this result is deferred to App.\,\ref{app:ProofEscapes}.

Computing the number of escapes $E$, given some boundary conditions, was thus reduced to the following counting exercise. A flat asymptote of height $ \omega_{\op{a}, +} > 0 $ ($ \omega_{\op{a},-} > 0$) at $ k_x \to + \infty $ ($ k_x \to - \infty $) contributes $-1$ ($ +1 $) to $E$. This convention comes about by consistency with the counting of parabolic states, cf.\,Fig.\,\ref{fig:mergers}. By contrast, eigenvalues $ \omega_{\op{a}, \pm} < 0 $ give no contribution to $E$, because they do not pertain the positive bulk gap of interest. In summary,
\begin{equation}
	E = ( \text{\# of asymptotes} \ \omega_{\op{a},-} > 0 ) - ( \text{\# of asymptotes} \ \omega_{\op{a},+} > 0 ) \,. \label{eq:PracticalE}
\end{equation}
The value of $E$ stemming from \eqref{eq:PracticalE} is displayed, for each of the families DD, ND, NN, in Sect.\,\ref{sec:Proofs}, and leads to the proof of Eq.\,\eqref{eq:EValueND} in Prop.\,\ref{prop:NewMap2} and Tab.\,\ref{tab:EscapesNN} in Prop.\,\ref{prop:NewMap4}.

\subsection{Boundary winding $B$} \label{subsec:BoundaryWinding}

In this section, we discuss how to compute the winding number $B$ (dubbed \textit{boundary winding}, cf.\,Def.\,\ref{def:BoundWinding}) of the unitary $U$ associated with a given orbit $ [\ul{A}] $ of self-adjoint boundary conditions. Generally, Prop.\,\ref{prop:BCs} holds, and computing $B$ amounts to applying formula \eqref{eq:BoundWinding}. However, there exist exceptional (as in, \textit{atypical} in the sense of Thm.\,\ref{thm:Typicality}) boundary conditions, where $B$ takes values that differ from those given by the integral formula \eqref{eq:BoundWinding} applied to \eqref{eq:UND}, \eqref{eq:UNN}. Below, we tackle both general and exceptional cases. In the interest of readability, most proofs are deferred to App.\,\ref{app:ProofB}. \vspace{1\baselineskip}

Recall that
\begin{equation}
	U = (A_1 + A_2)^{-1} (A_1 - A_2 ) \eqqcolon A_+^{-1} A_- \,, \label{eq:APlusAMinus}
\end{equation}
cf.\,(\ref{eq:UFormula}, \ref{eq:JuxtaposeA1A2}). Observe moreover that $ \det U : k_x \mapsto \det ( U(k_x)) $ describes a curve in the complex plane. Its winding about the origin, which we denote by $ N (\det U) $, is such that
\begin{equation}
	B(U) = N (\det U) \,,
\end{equation}
as is seen by \eqref{eq:BoundWinding} and 
\begin{equation}
		\diff (\log \circ \det U) = \op{tr} (U^{-1} \diff U) \,,
\end{equation}
for any unitary $U$ in $ U(n) \ (n \in \bbN) $. By known properties of winding numbers and \eqref{eq:APlusAMinus}
\begin{equation}
	B (U) = N (\det U) = N (\det A_-) - N (\det A_+) \eqqcolon N (P_-) - N(P_+) \,, \label{eq:PPlusPMinus}
\end{equation}
where $ P_\pm \coloneqq \det A_\pm $. Moreover, by the affine linear structure of
\begin{equation}
	\ul{A} : k_x \mapsto (A_1 (k_x), \, A_2 (k_x)) \,,
\end{equation}
$ P_\pm $ is a polynomial of order (at most) two in $ k_x $, namely a parabola. For the exceptional boundary conditions mentioned in the introductory paragraph, the parabola $ P_\pm $ however degenerates to a line, or even a single point in the complex plane $ \bbC  \ni P_\pm $.

Before computing $ P_\pm $ for boundary conditions $ \ul{A} $ in families DD, ND, NN (cf.\,Def.\,\ref{def:Families}), let us lay down general formulae for the winding of a (possibly degenerate) parabola 
\begin{equation}
	P (k) = c_0 k^2 + c_1 k + c_2 \,, \label{eq:PofK}
\end{equation}
where $ c_0, \, c_1, \, c_2 \in \bbC $, around the origin of its complex plane. Let
\begin{equation}
	d_{10} \coloneqq \op{Im} (c_1 \overline{c_0}) \,, \qquad d_{20} \coloneqq \op{Im} (c_2 \overline{c_0}) \,, \qquad d_{21} \coloneqq \op{Im} (c_2 \overline{c_1}) \,, \label{eq:Dij}
\end{equation}
and
\begin{equation}
	c(P) \coloneqq d_{20}^2 - d_{21} d_{10} \,. \label{eq:cofP}
\end{equation}
These quantities are of importance in the lemmas below.
\begin{lemma}[I: Curve traced] \label{lem:CurveTraced}
	The curve $ \bbR \to \bbC $, $ k \mapsto P(k) $ traces out
	\begin{itemize}
		\item[(a)] a point, if $ c_0 = c_1 = 0 $;
		
		\item[(b)] a straight line, if $ c_0 = 0 $, $ c_1 \neq 0 $;
		
		\item[(c)] a half-line traced both ways if $ c_0 \neq 0 $, $ d_{10} = 0 $;
		
		\item[(d)] a parabola, if otherwise; i.e., if $ c_0 \neq 0 $, $ d_{10} \neq 0 $.
	\end{itemize}
	In each case, the converse also applies.
\end{lemma}
\begin{lemma}[II: Avoidance of origin] \label{lem:Avoidance}
	The conditions for 
	\begin{equation}
		P(k) \neq 0 \,, \qquad (k \in \bbR) \label{eq:GeneralAvoidance}
	\end{equation}
	are, depending on the above cases,
	\begin{itemize}
		\item[(a)] $ c_2 \neq 0 $;
		
		\item[(b)] $ d_{21} \neq 0 $;
		
		\item[(c)] \textbullet\ $ d_{20} \neq 0 $, or \textbullet\ $ d_{20} = 0 $ and
		\begin{equation}
			4 c_2 \overline{c_0} |c_0|^2 > (c_1 \overline{c_0})^2 \,; \label{eq:AvoidanceC}
		\end{equation}
		
		\item[(d)] $ c(P) \neq 0 $.
	\end{itemize}
\end{lemma}
\begin{remark}
	Ad (c): In the second case, both sides of \eqref{eq:AvoidanceC} are real by $ d_{10} = 0 $, $ d_{20} = 0 $.
\end{remark}
If the origin is avoided, cf.\,\eqref{eq:GeneralAvoidance}, the winding number $ N(P) $ about $ z=0 $ is defined in the obvious (counter clock-wise) sense. It takes a half-integer value if the curve happens to be a straight line.
\begin{lemma}[III: Winding about origin] \label{lem:Winding}
	If the conditions I, II are met, the winding number is
	\begin{alignat}{4}
		& \text{(a)} \quad && N(P)  && = \hspace{2pt} && 0 \,; \nonumber \\
		& \text{(b)} \quad && N(P)  && = \hspace{2pt} && - \frac{1}{2} \op{sgn} d_{21} \,; \label{eq:WindingLine} \\
		& \text{(c)} \quad && N(P)  && = \hspace{2pt} && 0 \,; \nonumber \\
		& \text{(d)} \quad && N(P)  && = \hspace{2pt} && \begin{cases} 0 \,, & (c(P) > 0) \\ - \op{sgn} d_{10} \,, & (c(P) < 0) \,. \end{cases} \label{eq:WindingParabola}
	\end{alignat}
\end{lemma}
The next lemma connects the coefficients $ c_j^{\pm} \ (j=0,1,2)$ of $ P_\pm $ to the boundary conditions $ \ul{A} $ of interest.
\begin{lemma} \label{lem:CoefficientsOfPpm}
	Given self-adjoint boundary conditions $ \ul{A} $ parametrized as in \eqref{eq:ExplicitUlA}, the coefficients $ c_j^{\pm} \ (j=0,1,2)$ of $ P_\pm $, cf.\,\eqref{eq:PPlusPMinus}, are
	\begin{equation}
		\begin{aligned}
			c_0^\pm &= \pm b_2 \wedge b_1 \,, \\
			c_1^{\pm} &= b_2 \wedge (a_2 \pm a_1') + (a_1 \pm a_2') \wedge b_1 \,, \\
			c_2^\pm &= (a_1 \pm a_2') \wedge (a_1' \pm a_2) \,. \label{eq:CoefficientsOfPpm}
		\end{aligned}
	\end{equation}
\end{lemma}
\begin{proof}
	A straightforward computation based on (\ref{eq:JuxtaposeA1A2}, \ref{eq:ExplicitUlA}),
	\begin{gather}
		A_1 = (a_1, \, a_1' + k_x b_1) \,, \qquad A_2 = (a_2' + k_x b_2, a_2) \\
		A_\pm = \big( (a_1 \pm a_2') \pm k_x b_2, \, (\pm a_2 + a_1') + k_x b_1 \big) \,,
	\end{gather}
	and the observation \eqref{eq:DetTrick} that $ \op{det} (v_1, v_2) = v_1 \wedge v_2 $, $ v_i \in \bbC^2 $.
\end{proof}
The lemmas above, starting with the last one, can now be specialized to the families DD, ND, NN of boundary conditions, cf.\,Def.\,\ref{def:Families}. As customary, DN is not treated because of the substitution detailed in Prop.\,\ref{prop:BCs}.
\begin{lemma}
	When spelled out for families DD, ND, NN, Eq.\,\eqref{eq:CoefficientsOfPpm} reads
	\begin{itemize}
		\item[DD:]
		\begin{equation}
			\begin{aligned}
				c_0^\pm &= \pm b_2 \wedge b_1 \,, \\
				c_1^{\pm} &= \pm ( b_2 \wedge a_1' + a_2' \wedge b_1) \,, \\
				c_2^\pm &= \pm a_2' \wedge a_1' \,; \label{eq:CoefficientsOfPpmDD}
			\end{aligned}
		\end{equation}
		
		\item[ND:] 
		\begin{equation}
			\begin{aligned}
				c_0^\pm &= \pm \im \lambda a_1 \wedge b_1 \,, \\
				c_1^{\pm} &= \pm \im \lambda a_1 \wedge a_1' + (1 + |\alpha|^2 \pm \im \lambda') a_1 \wedge b_1 \,, \\
				c_2^\pm &= (1 + |\alpha|^2 \pm \im \lambda') a_1 \wedge a_1' \,; \label{eq:CoefficientsOfPpmND}
			\end{aligned}
		\end{equation}
		
		\item[NN:] 
		\begin{equation}
			\begin{aligned}
				c_0^\pm &= \pm (|\mu|^2 - \lambda_1 \lambda_2) a_1 \wedge a_2 \,, \\
				c_1^{\pm} &= \big( \im (\lambda_1 + \lambda_2) \mp (\lambda_2 \lambda_1' + \lambda_1 \lambda_2') \pm (\overline{\mu} \mu' + \mu \overline{\mu}') \big) a_1 \wedge a_2 \,, \\
				c_2^\pm &= \pm \big( (1 \pm \im \lambda_2') (1 \pm \im \lambda_1') + |\mu'|^2 \big) a_1 \wedge a_2 \,. \label{eq:CoefficientsOfPpmNN}
			\end{aligned}
		\end{equation}
	\end{itemize}
\end{lemma}
\begin{proof}
	An immediate consequence of \eqref{eq:CoefficientsOfPpm} and Prop.\,\ref{prop:BCs}.
\end{proof}
\begin{proposition}[Boundary winding $B$, family DD] \label{prop:BDD}
	Let $ \ul{A} $ be a self-adjoint boundary condition in family DD, cf.\,Def.\,\ref{def:Families}, moreover parametrized as in Lm.\,\ref{lem:Dictionary}, Prop.\,\ref{prop:BCs}. Let $U$ be the associated von Neumann unitary \eqref{eq:UFormula}, and let $ P_\pm $ be as in \eqref{eq:PPlusPMinus}. The index $ B(U) = N(P_-) - N(P_+) $ is well-defined whenever the curves $ P_\pm $ avoid the origin, cf.\,Lm.\,\ref{lem:Avoidance}, in which case
	\begin{equation}
		B(U) = 0 \,.
	\end{equation}
\end{proposition}
\begin{proof}
	Eq.\,\eqref{eq:PPlusPMinus} applies. By $ c_j^- = - c_j^+ $ we have $ P_- (k_x) = - P_+ (k_x) $, whence $ N(P_-) = N(P_+) $, whenever one of them (and hence both) are well-defined. Thus $ B(U) = 0 $ follows.
\end{proof}
The just stated result for family DD does not come as a surprise, since $ U(k_x) = J $ is a constant, cf.\,Prop.\,\ref{prop:UBCs}. 

Family ND is of greater complexity. There, Eq.\,\eqref{eq:CoefficientsOfPpmND} implies
\begin{equation}
	\label{eq:PTilde}
	\begin{aligned}
		P_\pm (k_x) &= \pm \im \lambda (a_1 \wedge b_1) k_x^2 + \big( \pm \im \lambda (a_1 \wedge a_1') + (1 + |\alpha|^2 \pm \im \lambda') (a_1 \wedge b_1) \big) k_x \\
		&+ (1 + |\alpha|^2 \pm \im \lambda') (a_1 \wedge a_1') \\
		&= \big( (a_1 \wedge a_1') + k_x (a_1 \wedge b_1) \big) \big( \pm \im \lambda k_x + (1 + |\alpha|^2 \pm \im \lambda') \big) \\
		&\equiv \big( (a_1 \wedge a_1') + k_x (a_1 \wedge b_1) \big) \tilde{P}_\pm (k_x) \,. 
	\end{aligned}
\end{equation}
The first factor in the last member of the equation above does not contribute to the winding $ B(U) $ by
\begin{equation}
	B(U) = N(P_-) - N(P_+) = N(\tilde{P}_-) - N(\tilde{P}_+) \,. \label{eq:BDifferenceTilde}
\end{equation}
Its zeros are however of relevance when discussing avoidance of the origin. In fact, we shall show that $ \tilde{P}_\pm $ have no zeros $ k_x \in \bbR $, implying that crossings of the origin, and in turn failure of self-adjointess, only take place at the fiber $ k_x $ solving
\begin{equation}
	(a_1 \wedge a_1') + k_x (a_1 \wedge b_1) = 0 \,,
\end{equation}
cf.\,Lm.\,\ref{lem:MaxRank} for details. We thus focus on $ \tilde{P}_\pm $ only, and in fact only on the $+$ variant, dropping subscripts. The $-$ variant is recovered by the replacement
\begin{equation}
	(\lambda, \lambda') \rightsquigarrow (- \lambda, - \lambda') \,. \label{eq:ReplacementND}
\end{equation}
Below, 
\begin{equation}
	\tilde{P}_+ (k_x) \equiv \tilde{P} (k_x) = \im \lambda k_x + (1 + |\alpha|^2 + \im \lambda') \label{eq:ExplicitPTilde}
\end{equation}
shall be seen as a special case of \eqref{eq:PofK} with
\begin{equation}
	\label{eq:CoeffOfPTilde}
	\begin{aligned}
		c_0 &= 0 \,, \\
		c_1 &= \im \lambda \,, \\
		c_2 &= 1 + | \alpha |^2 + \im \lambda' \,, 
	\end{aligned}
\end{equation}
whence $ d_{10} = d_{20} = 0 $ and
\begin{equation}
	d_{21} = - \lambda (1 + |\alpha|^2) \,. \label{eq:d21ND}
\end{equation}
\begin{proposition}[Winding of $P$, family ND] \label{prop:BND}
	Let $ \ul{A} $ be a self-adjoint boundary condition in family ND, cf.\,Def.\,\ref{def:Families}, moreover parametrized as in Lm.\,\ref{lem:Dictionary}, Prop.\,\ref{prop:BCs}. Then, the polynomial $ \tilde{P}_+ \equiv \tilde{P} $, cf.\,\eqref{eq:ExplicitPTilde}, satisfies:
	
	I. (\textit{Curve traced}) The curve $ k_x \mapsto \tilde{P} (k_x) $ traces out
	\begin{itemize}
		\item[(a)] a point iff $ \lambda = 0 $;
		
		\item[(b)] a straight line iff $ \lambda \neq 0 $.
	\end{itemize}
	
	II. (\textit{Avoidance of origin}) Depending on the cases in I, the curve avoids the origin, cf.\,\eqref{eq:GeneralAvoidance}, according to
	\begin{itemize}
		\item[(a)] always;
		
		\item[(b)] always.
	\end{itemize}
	
	III. (\textit{Winding about origin}) If the conditions I, II are met, the winding number is as in Lm.\,\ref{lem:Winding}, cf.\,(\ref{eq:WindingLine}, \ref{eq:WindingParabola}), with the restatement
	\begin{equation}
		N(P) = \frac{1}{2} \op{sgn} (\lambda) \label{eq:NofPND}
	\end{equation}
	in case (b).
\end{proposition}
\begin{proof}
	Given (\ref{eq:CoeffOfPTilde}, \ref{eq:d21ND}), I and III are immediate restatements of Lm.\,\ref{lem:CurveTraced}-\ref{lem:Winding}. II: (a) $ c_2 \neq 0 $ because $ \op{Re} (c_2) = 1 + |\alpha|^2 \neq 0 $ at all times. (b) $ d_{21} = - \lambda (1 + |\alpha|^2) \neq 0 $ by assumption $ \lambda \neq 0 $.
\end{proof}
We now present the specialization of Lms.\,\ref{lem:CurveTraced}, \ref{lem:Avoidance}, \ref{lem:Winding} to family NN. Just like in case ND, we focus on $ P_+ (k_x) = \op{det} A_+ (k_x) $. Analogous results for $ P_- = \op{det} A_- (k_x) $ are recovered via the substitution detailed prior to the proof of the next proposition, which shall be found in App.\,\ref{app:ProofB}.
\begin{proposition}[Winding of $P$, family NN] \label{prop:BNN}
	Let $ \ul{A} $ be a self-adjoint boundary condition in family NN, cf.\,Def.\,\ref{def:Families}, moreover parametrized as in Lm.\,\ref{lem:Dictionary}, Prop.\,\ref{prop:BCs}. The polynomial
	\begin{equation}
		P_+ (k_x) \equiv P(k_x) = c_0^+ k_x^2 + c_1^+ k_x + c_2^+ \,,
	\end{equation}
	with $ c_j^+ $, $ ( j = 0,1,2) $ as in \eqref{eq:CoefficientsOfPpmNN}, satisfies:
	
	I. (\textit{Curve traced}) The curve $ k_x \mapsto P (k_x) $ traces out
	\begin{itemize}
		\item[(a)] a point iff $ \mu = \lambda_1 = \lambda_2 = 0 $;
		
		\item[(b)] a straight line iff \textbullet\ $ | \mu |^2 = \lambda_1 \lambda_2 $ and \textbullet\ $ \lambda_2 \lambda_1' + \lambda_1 \lambda_2' \neq \overline{\mu} \mu' + \mu \overline{\mu}' $;
		
		\item[(c)] a half-line traced both ways iff \textbullet\ $ \lambda_2 = - \lambda_1 $ and \textbullet\ $ \mu \neq 0 $ or $ \lambda_1 \neq 0 $;
		
		\item[(d)] a parabola iff $ | \mu |^2 \neq \lambda_1 \lambda_2 $ and $ \lambda_1 + \lambda_2 \neq 0 $.
	\end{itemize}
	
	II. (\textit{Avoidance of origin}) Depending on the cases in I, the curve avoids the origin, cf.\,\eqref{eq:GeneralAvoidance}, according to
	\begin{itemize}
		\item[(a-d)] always.
	\end{itemize}
	
	III. (\textit{Winding about origin}) If the conditions I, II are met, the winding number is as in Lm.\,\ref{lem:Winding}, cf.\,(\ref{eq:WindingLine}, \ref{eq:WindingParabola}), with the restatement
	\begin{equation}
		N(P) = 
		\begin{cases}
			0 \,, & (|\mu|^2 > \lambda_1 \lambda_2) \\
			\op{sgn} (\lambda_1 + \lambda_2) \,, & (|\mu|^2 < \lambda_1 \lambda_2)
		\end{cases} \label{eq:NofPNN}
	\end{equation}
	in case (d).
\end{proposition}

\section{Proofs of the main results} \label{sec:Proofs}

In this section, we provide proofs of the main results. We start with Prop.\,\ref{prop:Map1}, reporting the value of the integer tuple $ \cal{V} $, cf.\,\eqref{eq:IntVector}, on the single orbit making up family DD. It is shown as a direct consequence of the preparatory results Props.\,\ref{prop:GEps}, \ref{prop:Varieties}, \ref{prop:EscapeHeight}, \ref{prop:BDD}. We then proceed to prove Props.\,\ref{prop:NewMap2} and \ref{prop:NewMap4}, concerning the values of $\cal{V}$ within families ND and NN. The first (second) one immediately implies the \enquote{index map} of Prop.\,\ref{prop:Map2} (Prop.\,\ref{prop:Map4}). Jointly, the two also imply Thm.\,\ref{thm:Typicality} on typicality of bulk-edge correspondence and lack thereof. They are moreover essential in showing Thm.\,\ref{thm:transitions} on the kinds of transitions of $ \cal{V} $, and their relation to \textit{spectral events}. Thm.\,\ref{thm:transitions} in turn implies Cor.\,\ref{cor:ViolationMechanism} on the mechanism for violation of bulk-edge correspondence.  \vspace{1\baselineskip}

We follow the program above, and start by discussing Prop.\,\ref{prop:Map1}.
\begin{proof}[Proof of Prop.\,\ref{prop:Map1}]
	The proposition consists of two statements. The first one is that
	\begin{equation}
		[\ul{A}] = [\ul{A}_D] \,, \qquad (\ul{A} \in \op{DD}) \,,
	\end{equation}
	where $ \ul{A}_D $ denotes Dirichlet boundary conditions. They are in DD because they involve no $ \partial_y $ derivatives, cf.\,Def.\,\ref{def:Families}, or equivalently because their $ 2 \times 4 $ matrix form is
	\begin{equation}
		\ul{A}_D = 
		\begin{pmatrix}
			0 & 1 & 0 & 0 \\
			0 & 0 & 1 & 0
		\end{pmatrix} \,,
	\end{equation}
	cf.\,\eqref{eq:Dirichlet} and \eqref{eq:AulA}, whence $ a_1 = a_2 = 0 $ and $ \ul{A}_D \in \op{DD} $ by Lm.\,\ref{lem:Dictionary}. Moreover, DD consists in a single orbit, namely
	\begin{equation}
		[\ul{A}_1] = [\ul{A}_2] \qquad (\ul{A}_1, \ul{A}_2 \in \op{DD})
	\end{equation}
	by Prop.\eqref{prop:UBCs}. Thus, $ [\ul{A}] = [\ul{A}_D] $ holds true for all $ \ul{A} \in \op{DD} $.
	
	The second statement is that 
	\begin{equation}
		\cal{V} \equiv \cal{V} (U) = (P, \, I, \, E, \, B) = (2,0,-1,0)
	\end{equation}
	on the single orbit labeled by $U$. We proceed entry by entry.
	
	\textit{Number of proper mergers $P$.} We argued in Sect.\,\ref{subsec:ComputingP} that $P$ is obtained as
	\begin{equation}
		P = 2 - W_\infty \,,
	\end{equation}
	where $ W_\infty $ is the winding of the scattering amplitude $ S (k_x, \kappa) $ in a neighborhood of $ (k_x, \kappa) = (\cos \varphi / \varepsilon, \, - \sin \varphi / \varepsilon) \to \infty $ (i.e., $ \varepsilon \to 0 $). However, by Eq.\,\eqref{eq:Expansions},
	\begin{equation}
		S (\varepsilon, \varphi) = - \frac{g(\varepsilon, - \varphi)}{g (\varepsilon, \varphi)} \simeq 1 \,, \qquad (\varepsilon \to 0)
	\end{equation}
	at leading order, whence $ W_\infty = 0 $ and
	\begin{equation}
		P=2 \,.
	\end{equation}
	
	\textit{Number of improper mergers $I$.} By Eq.\,\eqref{eq:GeometricI}, this integer is the signed number of intersections of the curve $ c_- (c_+) $, implicitly defined by
	\begin{equation}
		G (k_x, c_\pm) \equiv G (k_x, k_{y \pm} = \pm \im c_{\pm} k_x) = 0 \,,
	\end{equation}
	with $ \cal{A}_\uparrow $, $ \cal{A}_\downarrow $, cf.\,(\ref{eq:UpperArc}, \ref{eq:LowerArc}). However, by \eqref{eq:Varieties}, 
	\begin{equation}
		G (k_x, c_\pm) = 1 \neq 0 \,,
	\end{equation}
	and the curve $ c_- (c_+) $ is empty. Thus, 
	\begin{equation}
		I = 0 \,.
	\end{equation}
	
	\textit{Number of escapes $E$.} Two asymptotically flat states exist, with asymptotic height
	\begin{equation}
		\omega_{\op{a}, \pm} = \pm \frac{1}{2 \nu} \qquad (k_x \to \pm \infty) \,,
	\end{equation}
	cf.\,\eqref{eq:AsymptoteDD}. By formula \eqref{eq:PracticalE},
	\begin{equation}
		E = -1 \,.
	\end{equation}
	
	\textit{Boundary winding $B$.} The value of $B$ was found in Prop.\,\ref{prop:BDD}, and reads $ B=0 $.
\end{proof}
For the proofs of Props.\,(\ref{prop:NewMap2}, \ref{prop:NewMap4}), we recall a few notions. The first one is that of \textit{orbit parameters}. The parametrization proposed in Prop.\,\ref{prop:BCs} for the families DD-NN, cf.\,Def.\,\ref{def:Families}, is redundant if the object of interest are orbits $ [\ul{A}] $ under $ GL (2, \bbC) $ action, rather than single BCs $ \ul{A} $. We call \textit{orbit parameters} those that survive the lift $ \ul{A} \to [\ul{A}] $. In practice, they are the ones explicitly appearing in the von Neumman unitary $U$ labeling the orbit, cf.\,Prop.\,\ref{prop:UBCs}. Even \textit{orbit parameters} are partly redundant when concerned with global properties of the edge Hamiltonian $ H^\# $, due to the freedom in reparametrizing $ k_x $, cf.\,Rem.\,\ref{rem:TranslationsKx}. Parameters invariant under this freedom are called \textit{genuine}.  Concretely, those are $ (\op{Re} \alpha, \op{Im} \alpha, \lambda) \equiv (\alpha_R, \alpha_I, \lambda) \in \bbR^3 $ in family ND, and $ (\mu_R, \mu_I, \op{Im}(\mu' \overline{\mu}), \lambda_1, \lambda_2) \in \bbR^5 $ in family NN, cf.\,Rems.\,\ref{rem:ConsequenceTranslationKx}, \ref{rem:GenuineParsND}, \ref{rem:GenuineParsNN}. The discussion does not apply to DD, since a single orbit is parametrized by a single, constant parameter. Genuine parameters are unconstrained, and were chosen real. We thus call the \textit{parameter space} of ND or NN $ \bbR^3 \ni (\alpha_R, \alpha_I, \lambda) $ or $ \bbR^5 \ni (\mu_R, \mu_I, \op{Im}(\mu' \overline{\mu}), \lambda_1, \lambda_2) $, respectively.

Alongside these notions, we shall use the following argument several times. Let $ \bbR^n \ni p \mapsto Z(p) \in \bbZ $ be a $ \bb{Z} $-valued function. Let $Z$ be continuous everywhere except on a finite number of \textit{transition surfaces}, namely manifolds in $ \bbR^n $ of codimension one, implicitly defined by
\begin{equation}
	f(p) = 0 \label{eq:TransitionSurface}
\end{equation}
for some $ f: \bbR^n \to \bbR $. The function $Z$ is uniquely characterized by the (transition) surfaces where it changes, and its value in the regions that have those surfaces as boundaries. A transition surface like \eqref{eq:TransitionSurface} partitions parameter space in two regions
\begin{equation}
	f(p) > 0 \,, \qquad f(p) < 0 \,.
\end{equation}
In presence of a second surface $ g(p) = 0 $, $ \bbR^n \ni p $ would overall be partitioned in four regions, identified by the four ways of combining $ f(p) \gtrless 0 $ and $ g(p) \gtrless 0 $. The reasoning generalizes immediately to an arbitrary number of transition surfaces. The value $ Z^* $ of $ Z $ in a region $ \cal{R} $ where it is constant is simply found by evaluating $ Z(p^*) $ for some $ p^* \in \cal{R} $.

As observed in Rem.\,\ref{rem:IntVector}, the integers $ P, I, E $ are global properties of $ H^\# $, and they hence depend on \textit{genuine parameters} only. They are moreover piece-wise constant functions $ \bbR^n \to \bbZ $ ($ n=3,5 $ for ND, NN respectively), and the above characterization applies.
\begin{remark} \label{rem:gOfZero}
	Recall Eq.\,\eqref{eq:WInftyEquality}:
	\begin{equation}
		W_\infty = \lim_{\epsilon \rightarrow 0} \frac{1}{2 \pi \im} \int_{\mathcal{C}_{\epsilon,K}^{\infty}} S^{-1} \mathrm{d} S = \lim_{R \rightarrow 0} \frac{1}{2 \pi \im} \int_{\gamma_R} g^{-1} \mathrm{d} g \,,
	\end{equation}
	where $ g = g (\varepsilon, \varphi) $. If the winding of $ g(0,\varphi) $ is well-defined, then it is stable in $ \varepsilon > 0 $ and provides $ W_\infty $. That need though not be the case, and in fact when $ g (0, \varphi) = 0 $ occurs for some $ \varphi $. The latter will turn out to be the case in family ND, but not in NN.
\end{remark}
\begin{proof}[Proof of Prop.\,\ref{prop:NewMap2}]
	We discuss the entries of the integer tuple $ \cal{V} $ one by one. Throughout the proof, we shall assume $ 1 + \im \alpha \neq 0 $. By continuity, the assumption can be dropped. \vspace{1\baselineskip}
	
	\textit{Number of proper mergers $P$.} The winding of $ g (0, \varphi) $ is not well-defined, cf.\,Rem.\,\ref{rem:gOfZero}, and does not dispense from computing that of $ g (\varepsilon, \varphi) $. We will nonetheless consider $ g(0, \varphi) $ first. The winding of $S = - g(\varepsilon, -\varphi) / g (\varepsilon, \varphi)$ is unchanged when replacing $ g \rightsquigarrow \varepsilon^2 g \eqqcolon g' $, whence by \eqref{eq:Expansions}, and dropping primes,
	\begin{equation}
		\begin{gathered}
			g ( \varepsilon, \varphi) = \left( d (\varphi) + 2 \lambda' \varepsilon \right) a_1 \wedge \left( \cos \varphi b_1 + \varepsilon a_1' \right) \,, \\
			d (\varphi) = 2 \lambda \cos \varphi + \im \nu | 1 + \im \alpha |^2 \sin \varphi - \nu | 1 - \im \alpha |^2 \sqrt{1 + \cos^2 \varphi} \,.
		\end{gathered}
	\end{equation}
	This means
	\begin{equation}
			g (0, \varphi) = d(\varphi) (a_1 \wedge b_1) \cos \varphi \,.
	\end{equation}
	We note that
	\begin{equation}
		\op{Im} d (\varphi) > 0 \,, \qquad (0 < \varphi < \pi)
	\end{equation}
	and that 
	\begin{equation}
		\begin{gathered}
			d(0) = 2 \lambda - \sqrt{2} \nu | 1 - \im \alpha |^2 \,, \\
			d (\pi) = - 2 \lambda - \sqrt{2} \nu | 1 - \im \alpha |^2 \,, \\
			d ( \pi / 2 ) = z_0 \,, \qquad d ( \pi / 2 ) = \overline{z_0} \,, \label{eq:ValuesOfd}
		\end{gathered}
	\end{equation}
	with
	\begin{equation}
		z_0 \coloneqq \nu ( \im | 1 + \im \alpha |^2 - | 1 - \im \alpha |^2 ) \label{eq:DefZ0}
	\end{equation}
	having $ \op{Im} (z_0) > 0 $. In Fig.\,\ref{fig:WindingND}, left panel, we draw $ g (0, \varphi) $ assuming $ d(0) < 0, \, d(\pi) < 0 $.
	\begin{figure}[hbt]
		\centering
		\includegraphics[width=0.9\linewidth]{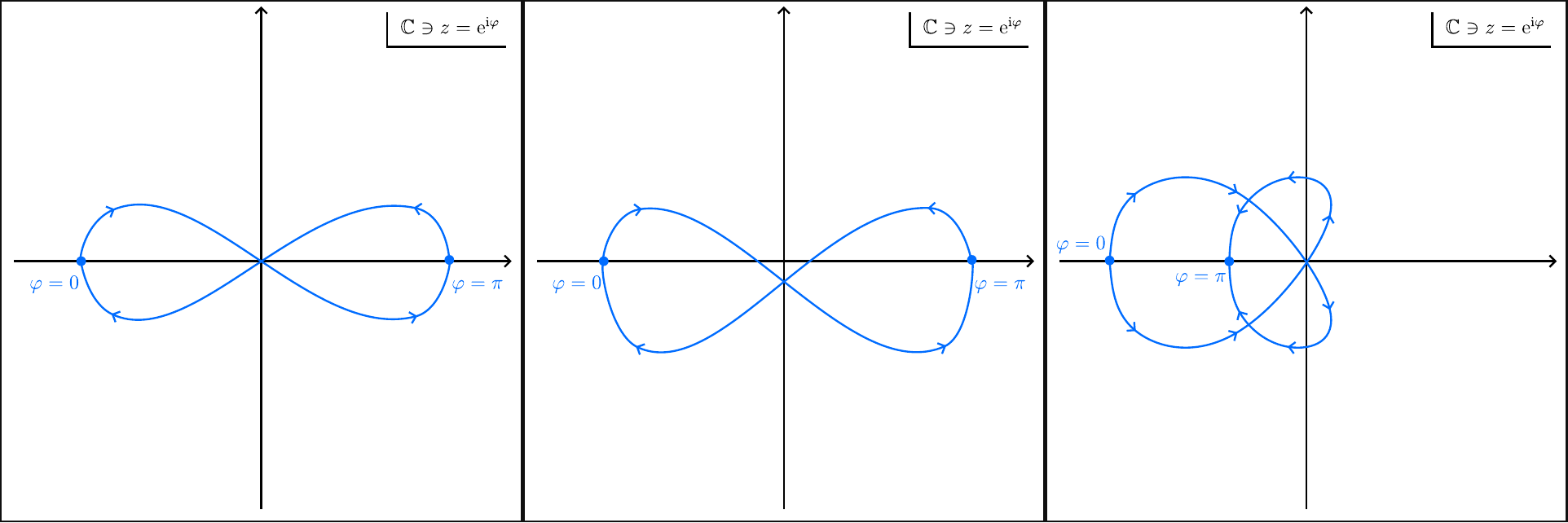}
		\caption{Left panel: $ g(0, \varphi) $ for $ d(0), \, d (\pi) < 0 $. Center panel: $ g(\varepsilon, \varphi) $ for $ d(0), \, d (\pi) < 0 $, avoided crossing with the origin. Right panel: $ g(0, \varphi) $ for $ d(0) <  0, \, d (\pi) > 0 $.}
		\label{fig:WindingND}
	\end{figure}
	
	Let 
	\begin{equation}
		\sigma \coloneqq \op{sgn} \op{Im} (a_1 \wedge a_1') \cdot \overline{(a_1 \wedge b_1)} \,.
	\end{equation}
	We note the invariance of $ \sigma $ under $ a_i' \to a_i' + \tau b_i $, cf.\,Rem.\,\ref{rem:TranslationsKx}, since it changes only the real part of the product. For now, let $ \sigma > 0 $.
	
	The curve in the figure visits the origin both at $ \varphi = \pm \pi/2 $. We claim it avoids it for small $ \varepsilon $, and in fact as shown in Fig.\,\ref{fig:WindingND}, center panel. The winding then vanishes. In order to establish this, we consider $ g (\varepsilon, \varphi) $ at first order near $ \varphi = \pm \pi /2 $, i.e., for small $ \gamma = \cos \varphi $:
	\begin{equation}
		g (\varepsilon, \varphi) = (\pm \im \nu |1 + \im \alpha|^2 - \nu |1 - \im \alpha|^2) a_1 \wedge (\gamma b_1 + \varepsilon a_1') + \cal{O} \big( (\varepsilon, \gamma)^2 \big) \,, \quad (\varphi \to \pm \pi/2) \,,
	\end{equation}
	namely
	\begin{equation}
		g (\varepsilon, \varphi) = \stackon[.1pt]{z}{\brabar}_0 a_1 \wedge (\gamma b_1 + \varepsilon a_1') + \cal{O} \big( (\varepsilon, \gamma)^2 \big) \,, \quad (\varphi \to \pm \pi/2) \,,
	\end{equation}
	cf.\,\eqref{eq:DefZ0}. Notice the invariance of the line $ \gamma \mapsto \gamma (a_1 \wedge b_1) + \varepsilon (a_1 \wedge a_1') $ (up to reparametrization) under the usual reparametrization of $ a_i' $.
	
	We keep $ \sigma > 0 $ for now. Then the line for $ \varepsilon > 0 $ runs to the \textit{left} of that for $ \varepsilon = 0 $, provided they are oriented according to increasing $ \gamma $. That is the case near $ \varphi = - \pi / 2 $, but not near $ \varphi = \pi / 2 $. Correspondingly, the origin will be found to the \textit{right} of $ g (\varepsilon, \varphi) $ near $ \varphi = - \pi/2 $ (respectively, left near $ \varphi = \pi/2 $). That confirms Fig.\,\ref{fig:WindingND}, center panel.
	
	Though $ \sigma < 0 $ has the curve move to the other side of the origin, the winding around it still vanishes. The same conclusion holds for $ d(0) > 0, \, d (\pi) > 0 $. It remains to be seen what happens when $ d(0) $, $ d(\pi) $ are of opposite signs, starting with $ d(0) < 0 , \, d(\pi) > 0 $ as in Fig.\,\ref{fig:WindingND}, right panel. The local situation near the origin is as above, and the same avoidance occurs at $ \varepsilon > 0 $. The winding number of $ g (\varepsilon, \varphi) $ equals $ -1 $. A similar picture for $ d(0) < 0 , \, d(\pi) > 0 $ yields the winding number $ +1 $. In summary,
	\begin{equation}
		2 - P = W_\infty =
		\begin{cases}
			1 \,, & d(\pi) < 0 < d(0) \\
			0 \,, & d(0), d(\pi) < 0 \\
			-1 \,, & d(0) < 0 < d(\pi) \,. \\
		\end{cases}
	\end{equation}
	Case $ d(0) > 0 , \ d(\pi) > 0 $ has been omitted because impossible: At least one in $ d(0) $, $ d (\pi) $ must be negative by $ d(0) + d(\pi) < 0 $, cf.\,\eqref{eq:ValuesOfd}.
	
	It remains to rephrase the cases as in \eqref{eq:PND}. The condition $ d(\pi) < 0 < d(0) $ is rephrased as
	\begin{equation}
		( q < - \sqrt{2} \ \op{and} \ q < \sqrt{2} ) \,,
	\end{equation}
	namely $ q < - \sqrt{2} $, using the definition \eqref{eq:MAndQ} of $ q $. Similarly, $d(0) < 0 < d(\pi)$ is equivalent to $ q > \sqrt{2} $. Finally, $ d(0), d(\pi) < 0 $ is rephrased as $ |q| < \sqrt{2} $. This concludes this part of the proof. 
	
	In closing, we highlight that $P$ has two \textit{transition surfaces}, namely
	\begin{equation}
		q (p) \mp \sqrt{2} = 0 \,, \qquad (p = (\alpha_I, \alpha_R, \lambda)) \,,
	\end{equation}
	where $ q(p) $ is given as a function of the genuine parameters in \eqref{eq:MAndQ}.
	\vspace{1\baselineskip}
	
	\textit{Number of improper mergers $I$.} We count intersections of
	\begin{equation}
		c_- = m c_+ + q \,, \label{eq:line}
	\end{equation}
	cf.\,Eq.\,\eqref{eq:CurveND}, with $ \cal{A}_\uparrow $ and $ \cal{A}_\downarrow $, cf.\,(\ref{eq:UpperArc}, \ref{eq:LowerArc}). Their signed number determines $I$ through \eqref{eq:GeometricI}, and $ \bbR^3 \ni (\alpha_R, \alpha_I, \lambda) = p \mapsto I(p) \in \bbZ $ is a piece-wise constant function on parameter space. We thus characterize it by: (a) its transition surfaces; (b) its value in regions that have those surfaces as boundaries.
	
	(a) Let
	\begin{equation}
		A_{\uparrow, -1} = (-1, +1) \,, \quad A_{\uparrow, 0} = (0, \sqrt{2}) \qquad \big( A_{\downarrow, 0} = (0, -\sqrt{2}) \,, \quad A_{\downarrow,+1} = (+1, -1) \big) \label{eq:EdgePoints}
	\end{equation}
	denote the \textit{edge points} of the target arc $ \cal{A}_\uparrow $ ($ \cal{A}_\downarrow $). The number of intersections of $ c_- (c_+) $ with $ \cal{A}_\uparrow $ only depends on the curve's height at $ c_+ = -1, 0 $, namely on the value of $ c_- (-1;p) $ and $ c_- (0;p) $. Since $ p \mapsto c_- (\xi;p) \ (\xi = -1, 0)$ are continuous functions, the number of intersections only changes when the \textit{edge points} $ A_{\uparrow, \xi} $ are met, i.e.
	\begin{equation}
		\begin{array}{llcl}
			A_{\uparrow, -1}: & c_- (-1;p) = +1 & \leftrightarrow & q = (m+1) \, \\
			A_{\uparrow, 0}: & c_- (0;p) = \sqrt{2} & \leftrightarrow & q = \sqrt{2} \,. \\
		\end{array} \label{eq:TransUpperArcND}
	\end{equation}
	Similarly, the number of intersections with $ \cal{A}_\downarrow $ changes at
	\begin{equation}
		\begin{array}{llcl}
			A_{\downarrow, 0}: & c_- (0;p) = -\sqrt{2} & \leftrightarrow & q = -\sqrt{2} \\
			A_{\downarrow, +1}: & c_- (+1;p) = -1 & \leftrightarrow & q = -(m+1) \,.  \\
		\end{array} \label{eq:TransLowerArcND}
	\end{equation}
	No transition surface other than (\ref{eq:TransUpperArcND}, \ref{eq:TransUpperArcND}) exists, and regions are defined by picking
	\begin{equation}
		q < - \sqrt{2} \,, \qquad - \sqrt{2} < q < \sqrt{2} \,, \qquad q > \sqrt{2} \,,
	\end{equation} 
	and one alternative for each of these two inequalities:
	\begin{equation}
		q \gtrless - (m+1) \,, \qquad q \gtrless (m+1) \,.
	\end{equation}
	Some of those regions are empty, and do not appear in Eq.\,\eqref{eq:IND}. For example, $ q < m+1 $ implies $ q < 1 $ by $ m \leq 0 $, and is thus incompatible with $ q > \sqrt{2} $.
	
	(b) Consider for example $ q < - \sqrt{2} $ and $ q < m+1 $. The third inequality is necessarily $ q < -(m+1) $, or equivalently $ q + m < -1 $, due to $ m \leq 0 $ and $ q < - \sqrt{2} $. Geometrically, cf.\,Fig.\,\ref{fig:MergingND}, the three inequalities mean
	\begin{equation}
		c_-(-1) < +1 \,, \qquad c_- (0) < - \sqrt{2} \,, \qquad c_- (+1) < -1 \,,
	\end{equation}
	whence the line $ c_- (c_+) $ avoids both $ \cal{A}_\uparrow $ and $ \cal{A}_\downarrow $, i.e., $ I=0 $. 
	
	By contrast, if $ q < - \sqrt{2} $ and $ q > m+1 $ ($ q < -(m+1) $ still), we have
	\begin{equation}
		c_-(-1) > +1 \,, \qquad c_- (0) < - \sqrt{2} \,, \qquad c_- (+1) < -1 \,.
	\end{equation}
	The arc $ \cal{A}_\uparrow $ is intersected while $ \cal{A}_\downarrow $ is still avoided, whence $ I = +1 $ by \eqref{eq:GeometricI}.
	\begin{figure}[hbt]
		\centering
		\includegraphics[width=0.4\linewidth]{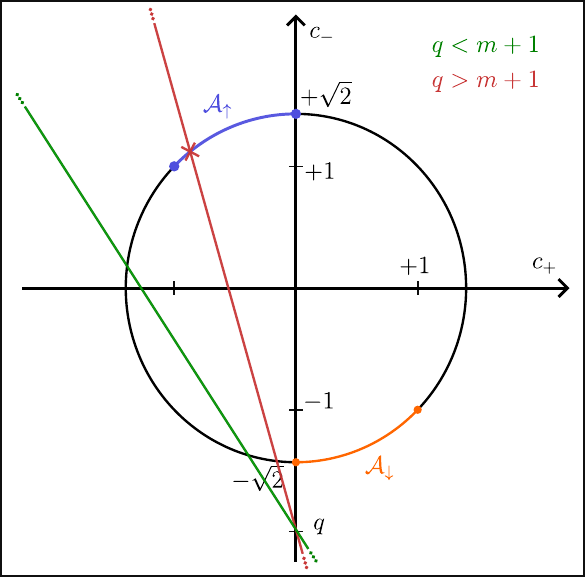}
		\caption{Lines $ c_- (c_+) $ with same intercept $ q < - \sqrt{2} $ and different slope $ m < 0 $. No intersection with $ \cal{A}_\downarrow $ is possible for $ q < - \sqrt{2} $, whereas $ \cal{A}_\uparrow $ is intersected for $ q > m+1 $.}
		\label{fig:MergingND}
	\end{figure}
	Repeating the analysis for all non-empty regions yields \eqref{eq:IND}. \vspace{1\baselineskip}
	
	\textit{Number of escapes $E$.} As per \eqref{eq:PracticalE}, we have to count the number of horizontal asymptotes of positive height, encoded by Eq.\,\eqref{eq:AsymptoteDN}:
	\begin{equation}
		\omega_{\op{a}, \pm} = \pm \frac{\lambda - \nu (1 + |\alpha|^2)}{2 \nu (\lambda \mp 2 \nu \alpha_I)} \equiv \pm \frac{N}{D_{\pm}} \,. \label{eq:HeightND}
	\end{equation}
	Zeros of $N$, $ D_+ $, $ D_- $ identify the transition surfaces
	\begin{equation}
		q - (m-1) = 0 \,, \qquad q + (m+1) = 0 \,, \qquad q - (m+1) = 0 \,,
	\end{equation}
	respectively. More precisely, the numerator $N$ is positive for
	\begin{equation}
		\lambda > \nu (1 + |\alpha|^2) \ \leftrightarrow \ q < m-1 \,.
	\end{equation}
	The denominators $ D_\pm $ are positive for
	\begin{equation}
		\lambda > \pm 2 \nu \alpha_I \ \leftrightarrow \ q < \mp (m+1) \,.
	\end{equation} 
	As a consequence,
	\begin{equation}
		\omega_{\op{a}, +}
		\begin{cases}
			>0 \,, & (q < m-1) \ \lor \ q > -(m+1) \\
			<0 \,, & m-1 < q < -(m+1)
		\end{cases} \label{eq:EPlusND}
	\end{equation}
	and
	\begin{equation}
		\omega_{\op{a}, -}
		\begin{cases}
			>0 \,, & m-1 < q < m+1 \\
			<0 \,, & ( q < m-1 ) \ \lor  ( q > m+1 ) \,.
		\end{cases} \label{eq:EMinusND}
	\end{equation}
	Eq.\,\eqref{eq:EValueND} follows. \vspace{1\baselineskip}
	
	\textit{Boundary winding $B$.} Notice that $ q=0 $ iff $ \lambda = 0 $ and $ \op{sgn} (q) = - \op{sgn} (\lambda) $. If $ q=0 $, point III (a) of Prop.\,\ref{prop:BND} implies $ B = 0 $, by
	\begin{equation}
		B = N (\tilde{P}_-) - N (\tilde{P}_+) \,,
	\end{equation}
	cf.\,\eqref{eq:BDifferenceTilde}, and the fact that $ \tilde{P}_\pm (k_x) $ both trace a single point. By contrast, if $ q \neq 0 $ both $ \tilde{P}_\pm (k_x) $ trace a line. Eq.\,\eqref{eq:NofPND} then states
	\begin{equation}
		N (\tilde{P}_+) = \frac{1}{2} \op{sgn} (\lambda) = - \frac{1}{2} \op{sgn} (q) \,.
	\end{equation}
	The $-$ variant is recovered by the replacement \eqref{eq:ReplacementND}:
	\begin{equation}
		N (\tilde{P}_-) = \frac{1}{2} \op{sgn} (-\lambda) = \frac{1}{2} \op{sgn} (q) \,.
	\end{equation}
	Thus,
	\begin{equation}
		B = N (\tilde{P}_-) - N (\tilde{P}_+) = \op{sgn} (q) \,,
	\end{equation}
	as claimed in Eq.\,\eqref{eq:BND}.
\end{proof}
\begin{proof}[Proof of Prop.\,\ref{prop:NewMap4}]
	\textit{Number of proper mergers $P$.} Like in the previous proof, we start by inspecting $ g (0, \varphi) $. In fact, since its winding is identical to $ \varepsilon^2 (a_1 \wedge a_2)^{-1} g (0, \varphi) $, and $ a_1 \wedge a_2 \neq 0 $ in family NN, we inspect the latter (while calling it $ g (0, \varphi) $ all the same). By Eq.\,\eqref{eq:Expansions}, it reads
	\begin{equation}
		g (0,\varphi) \equiv g (\varphi) = a (\varphi) \cos \varphi + \im b (\varphi) \sin \varphi \,,
	\end{equation}
	where
	\begin{equation}
		\begin{aligned}
			a(\varphi) &\coloneqq \nu (2 \mu_R + \lambda_1 + \lambda_2) \sqrt{1 + \cos^2 \varphi} + 2 (|\mu|^2 - \lambda_1 \lambda_2) \cos \varphi \,, \\
			b (\varphi) &\coloneqq 2 \nu^2 \sqrt{1 + \cos^2 \varphi} + \nu (2 \mu_R - \lambda_1 - \lambda_2) \cos \varphi \,.
		\end{aligned}
	\end{equation}
	The origin of the complex plane is met for no $ \varphi \in [0, 2\pi] $, if not for exceptional values of the parameters $ \mu \in \bbC $, $ \lambda_1, \lambda_2 \in \bbR $. There is thus no need to consider the full Jost function $ g (\varepsilon, \varphi) $, $ \varepsilon \neq 0 $, cf.\,Rem.\,\ref{rem:gOfZero}.
	
	The winding number of a loop $ \gamma (t) \in \bbC $, $ t \in [t_i, t_f] \subset \bbR $, about the origin is completely determined by: The number of intersections with the real axis, $ \op{Im} \gamma (t_j) = 0 $, $ j \in J $ ($J$ some index set); The position of $ \op{Re} \gamma (t_j) = \gamma (t_j) $ relative to each other and zero; The orientation of the curve. We exhibit the data above for $ g (0, \varphi) $, as a function of the parameters, upon noticing that 
	\begin{equation}
		g(\varphi) = \overline{g (2 \pi - \varphi)}
	\end{equation} 
	allows us to study $ \varphi \in [0, \pi] $ only.
	
	Zeros of $ \op{Im} g (\varphi) $ are: \textbullet\ $ \varphi = 0, \pi $ for all $ \mu, \lambda_1, \lambda_2 $; $ \varphi = \hat{\varphi} $ implicitly given by
	\begin{equation}
		\frac{\cos \hat{\varphi}}{\sqrt{1 + \cos^2 \hat{\varphi}}} = - \frac{2 \nu}{2 \mu_R - \lambda_1 - \lambda_2} \,,
	\end{equation}
	if solutions to this equation exist. Since the l.h.s.\,has range $[-1/\sqrt{2}, 1 / \sqrt{2}] $, that is the case for
	\begin{equation}
		| 2 \mu_R - \lambda_1 - \lambda_2 | \geq 2 \nu \sqrt{2} \ \leftrightarrow \ | \mu_R - \Sigma | \geq \nu \sqrt{2} \,,
	\end{equation}
	cf.\,\eqref{eq:SigmaAndDelta} for the definition of $ \Sigma $. We moreover notice that $ \hat{\varphi} \in (0, \pi/2) $ if $ (\mu_R - \Sigma) < - \nu \sqrt{2} $ and $ \hat{\varphi} \in (\pi/2, \pi) $ if $ (\mu_R - \Sigma) > \nu \sqrt{2} $.
	
	The orientation is fixed by specifying $ g (\tilde{\varphi}) $ for some $ \tilde{\varphi} \neq 0, \pi, \hat{\varphi} $, e.g.
	\begin{equation}
		g (\pi/2) = 2 \im \nu^2 \,.
	\end{equation}
	
	At the intersection points, the value of $ g (\varphi) $ is
	\begin{align}
		g (0) &= 2 (|\mu|^2 - \lambda_1 \lambda_2) + \nu \sqrt{2} (2 \mu_R + \lambda_1 + \lambda_2) \nonumber \\
		&= 2 (|\mu|^2 + \Delta^2 - \Sigma^2) + 2 \nu \sqrt{2} (\mu_R + \Sigma) = 2 (\Delta^2 - \cal{M}_-) \,, \\
		g (\pi) &= 2 (|\mu|^2 - \lambda_1 \lambda_2) - \nu \sqrt{2} (2 \mu_R + \lambda_1 + \lambda_2) \nonumber \\
		&= 2 (|\mu|^2 + \Delta^2 - \Sigma^2) - 2 \nu \sqrt{2} (\mu_R + \Sigma) = 2 (\Delta^2 - \cal{M}_+) \,, \\
		g (\hat{\varphi}) &= \frac{2 \nu^2 (4 \mu_I^2 + (\lambda_1 - \lambda_2)^2)}{(2 \mu_R - \lambda_1 - \lambda_2)^2 - 4 \nu^2} = \frac{2 \nu^2 (\mu_I^2 + \Delta^2)}{(\mu_R - \Sigma)^2 - \nu^2} \,,
	\end{align}
	cf.\,(\ref{eq:SigmaAndDelta}, \ref{eq:MPlusMinus}). Notice that $ g (\hat{\varphi}) > 0 $, whenever $ \hat{\varphi} $ exists as a solution of $ \op{Im} g (\hat{\varphi}) = 0 $, namely when $ |\mu_R - \Sigma| > \nu \sqrt{2} $.
	
	The winding number of $g$ is thus uniquely fixed if we specify a region with the following (candidate) transition surfaces as boundaries
	\begin{equation}
		\mu_R - \Sigma \pm \nu \sqrt{2} = 0 \,, \qquad g(0) = 0 \,, \qquad g(\pi) = 0 \,.
	\end{equation}
	Let us inspect one such region. The others are in all respects analogous.
	
	\textbullet\ $ g (\pi) < g(0) < 0 $ and $ (\mu_R - \Sigma) < - \nu \sqrt{2} $.
	
	To make contact with Eq.\,\eqref{eq:PNN}, we rephrase the constraints in the notation of that equation:
	\begin{equation}
		\begin{array}{lcl}
			g(\pi) < 0 & \leftrightarrow & \Delta^2 < \cal{M}_+ \\
			g(0) < 0 & \leftrightarrow & \Delta^2 < \cal{M}_- \\
			g(\pi) < g(0) & \leftrightarrow & \cal{M}_- < \cal{M}_+ \,,
		\end{array}
	\end{equation}
	whence we sit in case
	\begin{equation}
		\Delta^2 < \cal{M}_- < \cal{M}_+ \quad \land \quad \mu_R - \Sigma < - \nu \sqrt{2} \,.
	\end{equation}
	Drawing a curve according to this data, cf.\,left panel of Fig.\,\ref{fig:MergingNN}, reveals
	\begin{equation}
		W_\infty = + 2 \ \leftrightarrow \ P = 0 \,.
	\end{equation}
	Repeating the reasoning for all cases, some of which actually correspond to \textit{empty} regions in parameter space, leads to the following table, whose entries represent values of $P$ in a given region.
	\begin{center}
		\begin{tabular}{ c|c"c|c|c }
			\multicolumn{2}{c"}{\xrowht[()]{10pt}} & ${\mu_R - \Sigma < - \nu \sqrt{2}}$ & ${|\mu_R - \Sigma| < \nu \sqrt{2}}$ & ${\mu_R - \Sigma > \nu \sqrt{2}}$ \\
			\thickhline
			\multicolumn{2}{c"}{\xrowht{10pt} $ {\Delta^2 > \cal{M}_+, \, \cal{M}_-} $} & $ 2 $ & $ 2 $ & $ 2 $ \\
			\hline
			\multicolumn{2}{c"}{\xrowht{10pt} $ {\cal{M}_+ > \Delta^2 > \cal{M}_-} $} & $ 1 $ & $1$ & $\varnothing$ \\
			\hline
			\multicolumn{2}{c"}{\xrowht{10pt} $ {\cal{M}_+ > \Delta^2 > \cal{M}_-} $} & $\varnothing$ & $3$ & $3$ \\
			\hline
			\multicolumn{2}{c"}{\xrowht{10pt} $ {\Delta^2 < \cal{M}_- < \cal{M}_+} $} & $0$ & $\varnothing$ & $\varnothing$ \\
			\hline
			\multicolumn{2}{c"}{\xrowht{10pt} $ {\Delta^2 < \cal{M}_+ < \cal{M}_-} $} & $\varnothing$ & $\varnothing$ & $4$ \\
		\end{tabular}
		\label{tab:WindNN}
	\end{center}
	By inspection of this table,
	\begin{equation}
		\mu_R - \Sigma \pm \nu \sqrt{2} = 0
	\end{equation}
	were, in hindsight, not transition surfaces. Eq.\,\ref{eq:PNN} follows. \vspace{1\baselineskip}
	
	\textit{Number of improper mergers $I$.} The curve $ c_- (c_+) $ whose intersections with $ \cal{A}_\uparrow, \, \cal{A}_\downarrow $ yield $I$ is now a hyperbola. More precisely,
	\begin{align}
		c_- (c_+) &= \frac{2 (|\mu|^2 - \lambda_1 \lambda_2) + \nu (2 \mu_R - \lambda_1 - \lambda_2)c_+}{\nu (2 \mu_R + \lambda_1 + \lambda_2) + 2 \nu^2 c_+} \nonumber \\
		&= \frac{(|\mu|^2 + \Delta^2 - \Sigma^2) +  \nu (\mu_R - \Sigma) c_+}{\nu (\mu_R + \Sigma) + \nu^2 c_+} \equiv c_- (c_+;p) \,,
	\end{align}
	cf.\,\eqref{eq:CurveNN} and \eqref{eq:SigmaAndDelta}, where $ p = (\mu_R, \mu_I, \Sigma, \Delta) \in \bbR^4 $ represents a point in parameter space (the \textit{genuine parameter} $ \op{Im} (\mu' \overline{\mu}) $ is not reported, since it does not appear). Two facts will be important in the following. First, the center $C$ of the hyperbola lies at
	\begin{equation}
		C = (x_c, y_c) = \Big( - \frac{(\mu_R + \Sigma)}{\nu}, \, \frac{\mu_R - \Sigma}{\nu} \Big) \,.
	\end{equation}
	Second, the hyperbola always sits in the first and third quadrant w.r.t.\,its (horizontal and vertical) asymptotes, namely 
	\begin{equation}
		c_-'(c_+;p) \coloneqq \frac{\diff}{\diff c_+} c_- (c_+;p) \leq 0 \qquad (c_+ \in \bbR, \ p \in \bbR^4) \,.
	\end{equation}
	
	Like in the previous proof, we determine $ I(p) $ by: (a) identifying its the \textit{transition surfaces}; (b) evaluating it in each region with the transition surfaces as boundaries.
	
	(a) Once again, the number of intersections of the hyperbola with $ \cal{A}_\uparrow $ only depends on its height at $ c_+ = -1, 0 $, namely on the value of $ c_- (-1;p) $ and $ c_- (0;p) $. If $ p \mapsto c_- (\xi;p) \ (\xi = -1, 0)$ were continuous functions, like in family ND, the number of intersections would only change when the \textit{edge points} $ A_{\uparrow, \xi} $ are met, i.e.
	\begin{equation}
		\begin{array}{llcl}
			A_{\uparrow, -1}: & c_- (-1;p) = +1 & \leftrightarrow & \Delta^2 = \cal{I}_+ \\
			A_{\uparrow, 0}: & c_- (0;p) = \sqrt{2} & \leftrightarrow & \Delta^2 = \cal{M}_+ \,. \\
		\end{array} \label{eq:TransUpperArc1}
	\end{equation}
	However, such functions are discontinuous for $ x_c \equiv x_c (p) = \xi $ $ (\xi = -1,0) $, namely 
	\begin{equation}
		\begin{array}{ll}
			x_c = -1: & \mu_R + \Sigma = \nu \\
			x_c = 0: & \mu_R + \Sigma = 0 \,.
		\end{array} \label{eq:TransUpperArc2}
	\end{equation}
	Jumps in the number of intersections with $ \cal{A}_\uparrow $ must occur across the transition surfaces \eqref{eq:TransUpperArc1}, and could occur across \eqref{eq:TransUpperArc2}.
	
	Similarly, the number of intersections with $ \cal{A}_\downarrow $ may jump at
	\begin{equation}
		\begin{array}{llcl}
			A_{\downarrow, 0}: & c_- (0;p) = -\sqrt{2} & \leftrightarrow & \Delta^2 = \cal{M}_- \\
			A_{\downarrow, +1}: & c_- (+1;p) = -1 & \leftrightarrow & \Delta^2 = \cal{I}_-  \\
			x_c = 0: & \mu_R + \Sigma = 0 & & \\
			x_c = 1: & \mu_R + \Sigma = - \nu  & &
		\end{array}
	\end{equation}
	(the $ x_c = 0 $ transition is shared with $ \cal{A}_\uparrow $).
	
	A region whose boundaries are transition surfaces is thus specified by picking
	\begin{equation}
		\mu_R + \Sigma > \nu \,, \quad 0 < \mu_R + \Sigma < \nu \,, \quad - \nu < \mu_R + \Sigma < 0 \,, \quad \mu_R + \Sigma < - \nu
	\end{equation}
	and one alternative
	\begin{equation}
		\Delta^2 \gtrless \chi  
	\end{equation}
	for each $ \chi \in \{ \cal{M}_+, \cal{M}_-, \cal{I}_+, \cal{I}_- \} $. This entails a grand total of $ 4 \times 2^4 = 64 $ regions, an example of which could be
	\begin{equation}
		\cal{R}_1 \coloneqq \{ (\mu_R, \mu_I, \Sigma, \Delta) \in \bbR^4 \ | \ \mu_R + \Sigma < - \nu \,, \ \cal{M}_+ < \Delta^2 < \cal{M}_- \,, \ \Delta^2 < \cal{I}_\pm \} \,. \label{eq:Region1}
	\end{equation}
	The tally can be lowered to the forty-eight cases of Tab.\,\ref{tab:MergeNN1AndNN2} by observing that
	\begin{equation}
		\cal{M}_+ > \cal{M}_- \ \leftrightarrow \ \mu_R + \Sigma > 0 \,.
	\end{equation}
	As a consequence, picking $ \cal{M}_+ > \Delta^2 > \cal{M}_- $ and $ - \nu < \mu_R + \Sigma < 0 $ or $ \mu_R + \Sigma < - \nu $ at once would specify an empty region. The choice of $ \Delta^2 \gtrless \cal{I}_\pm $ is still free, so this rules out $ 4 \times 2 = 8 $ cases. Out of the remaining $ 56 $, another $ 8 $ are similarly eliminated by
	\begin{equation}
		\cal{M}_+ < \cal{M}_- \ \leftrightarrow \ \mu_R + \Sigma < 0 \,. 
	\end{equation}
	
	(b) Given a region $ \cal{R} $ as specified above, we evaluate $ I(p) $ by simply picking $ p \in \cal{R} $, drawing $ c_- (c_+;p) $, and counting its signed number of intersections with the target arcs. Consider for example $ \cal{R}_1 $, cf.\,\ref{eq:Region1}. The constraints geometrically read
	\begin{equation}
		x_c (p) > 1 \,, \quad - \sqrt{2} < c_- (0;p) < \sqrt{2} \,, \quad c_- (-1;p) > +1 \,, \quad c_- (+1;p) > -1 \,, \label{eq:HyperbolaConstraints}
	\end{equation}
	for all $ p \in \cal{R}_1 $, upon noticing that $ c_- (-1;p) > +1 $ is equivalent to
	\begin{equation}
		( \mu_R + \Sigma > \nu \ \land \ \Delta^2 > \cal{I}_+ ) \ \lor \ ( \mu_R + \Sigma < \nu \ \land \ \Delta^2 < \cal{I}_+ ) \,,
	\end{equation}
	$ c_- (+1;p) > -1 $ to
	\begin{equation}
		( \mu_R + \Sigma > -\nu \ \land \ \Delta^2 > \cal{I}_- ) \ \lor \ ( \mu_R + \Sigma < -\nu \ \land \ \Delta^2 < \cal{I}_- ) \,,
	\end{equation}
	and $ c_- (0;p) < \pm \sqrt{2} $ to
	\begin{equation}
		( \mu_R + \Sigma < 0 \ \land \ \Delta^2 > \cal{M}_\pm ) \ \lor \ ( \mu_R + \Sigma > 0 \ \land \ \Delta^2 < \cal{M}_\pm ) \,.
	\end{equation}
	A representative curve $ c_- (c_+;p) $ is drawn in Fig.\,\ref{fig:MergingNN}.
	\begin{figure}[hbt]
		\centering
		\includegraphics[width=0.65\linewidth]{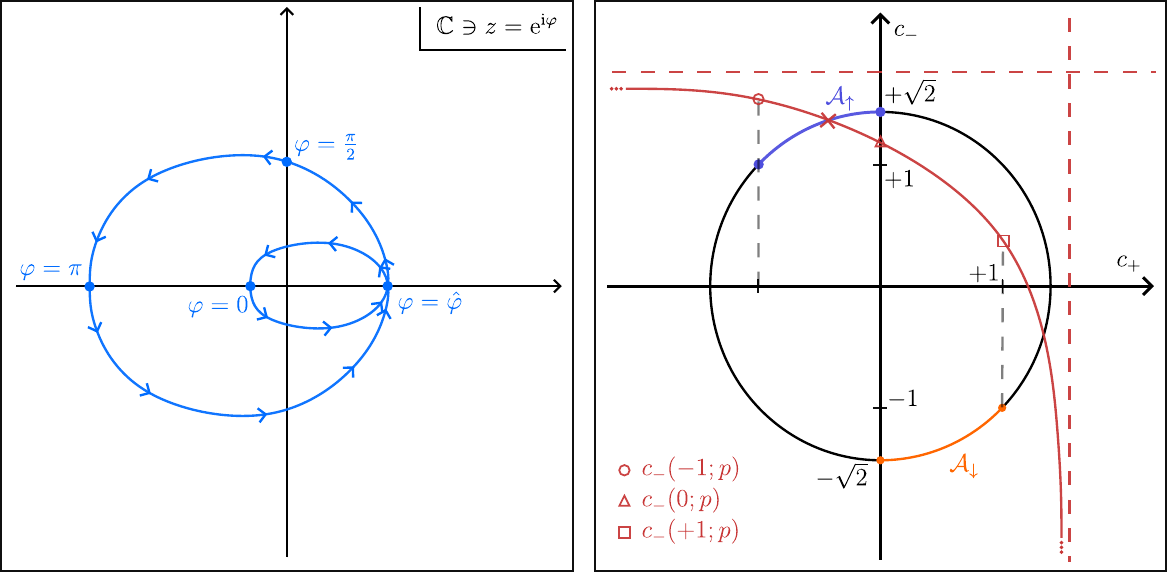}
		\caption{Left panel: Curve $ g(\varphi) $ for $\Delta^2 < \cal{M}_- < \cal{M}_+ $, $ \mu_R - \Sigma < - \nu \sqrt{2}$. Notice that $ g(\hat{\varphi}) > 0 $, and that the orientation is fixed by $ \op{Im} g (\pi/2) > 0 $ and $ 0 < \hat{\varphi} < \pi/2 $. Right panel: Hyperbola $ c_- (c_+;p) $ for $ p \in \cal{R}_1 $. Highlighted are the points $ c_- (\xi; p) $, $ ( \xi = -1,0,1 ) $, upon which constraints are imposed.}
		\label{fig:MergingNN}
	\end{figure}
	By $ c_-' (c_+;p) \leq 0 $ and the constraints \eqref{eq:HyperbolaConstraints}, the lower branch of the hyperbola intersects $ \cal{A}_\uparrow $ and avoids $ \cal{A}_\downarrow $. By \eqref{eq:GeometricI}, the only intersection with $ \cal{A}_\uparrow $ contributes $ +1 $, whence
	\begin{equation}
		I = +1 \,.
	\end{equation}
	Similar counting for the other 47 regions provides the entries of Tab.\,\ref{tab:MergeNN3AndNN4}. \vspace{1\baselineskip}
	
	\textit{Number of escapes $E$}. Just like in the proof of Prop.\,\ref{prop:NewMap2}, we need to count the signed number of positive horizontal asymptotes. This time,
	\begin{equation}
		\omega_{\op{a}, \pm} = \frac{|\mu|^2 - \nu^2 + \nu (\lambda_1 + \lambda_2) - \lambda_1 \lambda_2}{4 \nu^2 \mu_R \pm 2 \nu ( |\mu|^2 + \nu^2 - \lambda_1 \lambda_2 )} = \frac{| \mu |^2 + \Delta^2 - (\Sigma - \nu)^2}{4 \nu^2 \mu_R \pm 2 \nu (|\mu|^2 + \nu^2 + \Delta^2 - \Sigma^2)} \equiv \pm \frac{N}{D_\pm} \,, \label{eq:AsymptoteHeightNewVars}
	\end{equation}
	cf.\,Eq.\,\eqref{eq:AsymptoteNN}. $E$ changes across the transitions surfaces
	\begin{equation}
		N = 0 \ \leftrightarrow \ \Delta^2 = (\Sigma - \nu)^2 - | \mu |^2 \equiv \cal{E} \,,
	\end{equation}
	cf.\,\eqref{eq:EPlusMinus}, and
	\begin{equation}
		D_\pm = 0 \ \leftrightarrow \ \Delta^2 = \Sigma^2 - | \mu |^2 \mp 2 \nu \mu_R - \nu^2 \equiv \cal{I}_\pm \,.
	\end{equation}
	The last two were already encountered in the discussion of $I$.
	
	Tab.\,\ref{tab:EscapesNN} now stems from observing that $ \omega_{\op{a}, \pm} > 0 $ for
	\begin{equation}
		(\Delta^2 > \cal{E} \ \land \ \Delta^2 \gtrless \cal{I}_\mp) \quad \lor \quad (\Delta^2 < \cal{E} \ \land \ \Delta^2 \lessgtr \cal{I}_\mp) \,,
	\end{equation}
	and recalling that $ \omega_{\op{a}, +} > 0 $ ($ \omega_{\op{a}, -} > 0 $) contributes $ -1 $ ($+1$) to $E$, or zero otherwise.
	
	\textit{Boundary winding $B$.} As far as $B$ is concerned, Prop.\,\ref{prop:NewMap4} is a mere restatement of Prop.\,\ref{prop:BNN}, if not for the fact that 
	\begin{equation}
		B = N(P_-) - N(P_+)
	\end{equation}
	is reported explicitly. Recall \eqref{eq:NofPNN}, which says (re-instating subscripts)
	\begin{equation}
		\begin{aligned}
			N(P_+) &= 
			\begin{cases}
				0 \,, & (|\mu|^2 > \lambda_1 \lambda_2) \\
				\op{sgn} (\lambda_1 + \lambda_2) \,, & (|\mu|^2 < \lambda_1 \lambda_2)
			\end{cases} \\
			&\equiv 
			\begin{cases}
				0 \,, & \Delta^2 > \cal{B} \\
				\op{sgn} (\Sigma) \,, & \Delta^2 < \cal{B} \,,
			\end{cases}
		\end{aligned}
	\end{equation}
	cf.\,(\ref{eq:SigmaAndDelta}, \ref{eq:DefB}). The $-$ variant is obtained via the replacement \eqref{eq:SecondReplacement}. It entails $ \Sigma \rightsquigarrow - \Sigma $, $ \Delta^2 \rightsquigarrow \Delta^2 $, $ \cal{B} \rightsquigarrow \cal{B} $, whence
	\begin{equation}
		N(P_-) = 
		\begin{cases}
			0 \,, & \Delta^2 > \cal{B} \\
			-\op{sgn} (\Sigma) \,, & \Delta^2 < \cal{B} \,,
		\end{cases}
	\end{equation}
	and
	\begin{equation}
		B = N(P_-) - N(P_+) = -2 \op{sgn} (\Sigma) \,.
	\end{equation}
	This completes the proof.
\end{proof}
\begin{proof}[Proof of Prop.\,\ref{prop:Map2}]
	Trivial consequence of Prop.\,\ref{prop:NewMap2}.
\end{proof}
\begin{proof}[Proof of Prop.\,\ref{prop:Map4}]
	Follows by specializing Prop.\,\ref{prop:NewMap4} to the particle-hole symmetric subcase of family NN, via the restrictions detailed in Prop.\,\ref{prop:PHSBCs}. In particular, \eqref{eq:NNPHS} demands $ \mu_R \equiv \op{Re} \mu = 0 $, and the new transition manifolds are found by
	\begin{equation}
		\chi^\Xi = \chi \rvert_{\mu_R = 0} \qquad (\chi = \cal{M}_{\pm}, \cal{I}_\pm, \cal{E}_\pm, \cal{B}) \,.
	\end{equation}
	Notice that $ \cal{I}_+^\Xi = \cal{I}_-^\Xi \equiv \cal{I}^\Xi $, cf.\,also \eqref{eq:PHSTransNN}. Two transition manifolds of $I$ come to coincide when restricted to the particle-hole submanifold within family NN. As a consequence, $I$ undergoes a non-elementary transition while crossing $\cal{I}^\Xi$, cf.\,Thm.\,\ref{thm:transitions} and Rem.\,\ref{rem:NonElementary}. This is displayed in Fig.\,\ref{fig:IndicesIV}, where crossing $\cal{I}^\Xi$ yields $ \Delta I = \pm 2 $.
\end{proof}
\begin{proof}[Proof of Thm.\,\ref{thm:Typicality}]
	In family DD, there is no notion of typicality, since it consists in a single orbit. In family ND, typicality of violations (or not) of bulk-edge correspondence is already evident from Fig.\,\ref{fig:IndicesII}. In family NN, inspection of \eqref{eq:PNN} and Tab.\,\ref{tab:MergeNN3AndNN4} reveals that
	\begin{equation}
		C_+ = 2 = M = P + I \,,
	\end{equation}
	namely bulk-edge correspondence holds, for $ \Delta^2 > \cal{I}_\pm $. Conversely, it is violated in the complementary case. The proof is concluded by noticing that both sets are positive-measure.
\end{proof}
\begin{proof}[Proof of Thm.\,\ref{thm:transitions}]
	We divide the proof in two parts, (i) and (ii). In (i), we argue, based on Props.\,\ref{prop:NewMap2} and \ref{prop:NewMap4}, that the entries of $ \cal{V} $, cf.\,\ref{eq:IntVector}, only change in the specified combinations (a)-(d). In (ii), we connect such changes to spectral events.
	\begin{itemize}
		\item[(i)] \textit{Relations between transitions of integers.}
		\begin{itemize}
			\item[(a)] The claim is $ \Delta P = - \Delta I = \pm 1 $, and that $P$ does not change otherwise. Both facts are trivially true in family DD, where parameter space is a point and no transitions exist. We analyze ND and NN separately, while DN is recovered from ND by the substitution detailed in Lm.\,\ref{lem:Dictionary}. 
			
			ND: Transitions of $P$ only happen at
			\begin{equation}
				q (p) \mp \sqrt{2} = 0 \,, \qquad (p = (\alpha_R, \alpha_I, \lambda) \in \bbR^3) \,,
			\end{equation} 
			cf.\,Eq.\,\eqref{eq:PND}. As a consequence of Eq.\,\eqref{eq:IND}, they are shared by $I$. When $ q = \sqrt{2} $ ($ q = - \sqrt{2} $) is crossed in direction of increasing $ q $,
			\begin{equation}
				\Delta P = +1 \qquad (\Delta P = +1)
			\end{equation}
			by Eq.\,\eqref{eq:PND} or Fig.\,\ref{fig:IndicesII}. At the same time,
			\begin{equation}
				\Delta I = -1 \qquad (\Delta I = -1) \,,
			\end{equation}
			proving this part of the claim. 
			
			NN: By Eq.\,\eqref{eq:PNN}, transition surfaces of $P$ are
			\begin{equation}
				\Delta^2 - \cal{M}_\pm = 0 \,.
			\end{equation}
			They are shared by $I$, cf.\,Tab.\,\ref{tab:MergeNN3AndNN4}. Inspection of Eq.\,\eqref{eq:PNN} and Tab.\,\ref{tab:MergeNN3AndNN4} allows checking $ \Delta P = - \Delta I = \pm 1 $ within each of the regions defined by the boundaries
			\begin{equation}
				\mu_R + \Sigma = - \nu \,, \quad \mu_R + \Sigma = 0 \,, \quad \mu_R + \Sigma = \nu 
			\end{equation}
			and
			\begin{equation}
				\Delta^2 = \cal{I}_\pm \,.
			\end{equation}
			The statement $ \Delta P = - \Delta I = \pm 1 $ is seen to hold true while crossing $ \Delta^2 = \cal{M}_\pm $, for each such region separately, and thus generally for family NN.
			
			\item[(b)] The claim is that there exist transitions where $ I, E $ change jointly, and in that case $ \Delta I = \Delta E = \pm 1 $. Family DD is again trivial, and will be omitted from now on. In families ND and NN, the transition surfaces of $I$ that are not shared with $P$ are however shared with $E$. They read
			\begin{equation}
				q \mp (m+1) = 0 \,,
			\end{equation}
			in the former family, cf.\,\eqref{eq:IND}, and
			\begin{equation}
				\Delta^2 - \cal{I}_\pm = 0
			\end{equation}
			in the latter. That $ \Delta I = - \Delta E = \pm 1 $ while crossing the aforementioned surfaces follows from inspection of Fig.\,\ref{fig:IndicesII} in family ND, and of Tabs.\,\ref{tab:MergeNN3AndNN4}, \ref{tab:EscapesNN} in family NN. In the latter case, comparing the two tables is actually not immediate, and requires careful ordering of the surfaces in direction of increasing $ \Delta^2 $.
			
			\item[(c)] In family ND (NN), $E$ has a single transition surfaces
			\begin{equation}
				q = m-1  \qquad (\Delta^2 = \cal{E}) \,,
			\end{equation}
			not shared with other indices, cf.\,Eqs.\,(\ref{eq:EValueND}, \ref{eq:EPlusMinus}). That $ \Delta E = \pm 2 $ while crossing them is read by Eq.\,\eqref{eq:EValueND} (family ND), or by subtracting the two columns of Tab.\,\ref{tab:EscapesNN} (family NN) along the rows where $ \Delta E \neq 0 $. In the rows with $ \Delta E = 0 $, passing from one column to the other does not constitute, by the definition of Thm.\,\ref{thm:transitions}, a transition of $ \cal{V} $. Incidentally, we notice that this move across columns with $ \Delta E = 0 $ is associated to the spectral event discussed in Rem.\,\ref{rem:WeirdETrans}.
			
			\item[(d)] Neither in ND nor in NN are the transition surfaces of $B$, $ q = 0 $ and $ \Delta^2 = \cal{B} $ respectively (cf.\,(\ref{eq:BND}, \ref{eq:DefB})), shared by other indices. The integer $B$ therefore changes independently from the rest.
		\end{itemize}
		
		\item[(ii)] \textit{Correspondence with spectral events.}
		\begin{itemize}
			\item[(a)] Levinson's relative theorem, cf.\,\ref{thm:Levinson}, establishes a correspondence between $P$ and the signed number of mergers with the bulk spectrum. If $P$ changes, one such merger must appear of vanish. By continuity of edge eigenvalues $ \omega (k_x) $ in $ k_x $, such mergers cannot appear or disappear at will, but must rather be drawn from, or repelled to, infinity. Hence the claim.
			
			\item[(b)] In families DD (trivially), ND and NN, the only typical asymptotic behaviors of edge eigenvalues are parabolic, $ \omega \gtrsim k_x^2 $ (cf.\,\eqref{eq:RegionI}), or flat, $ \omega \to \omega_{\op{a}, \pm} $, $ (k_x \to \pm \infty) $, cf.\,Eqs.\,(\ref{eq:AsymptoteDD}, \ref{eq:AsymptoteDN}, \ref{eq:AsymptoteNN}). This fact is seen as follows. Parabolic behavior is typical because edge states with $ \omega \gtrsim k_x^2 $ exist at $ |k_x| \to \infty $ in positive-measure subsets of ND, NN, cf.\,Props.\,\ref{prop:NewMap2} and \ref{prop:NewMap4}. By now considering $ |\omega|^2 \ll k_x^2 $ at $ |k_x| \to \infty $, we cover the entire $ (k_x, \omega) $-plane. In this region, Prop.\,\ref{prop:EscapeHeight} proves that the only typical behavior is $\omega$ asymptotically flat. Different asymptotic behaviors $ \omega = \cal{O} (k_x^\gamma) $, $ \gamma \in (0,1) \cup (1,2) $, are admitted only exceptionally, and more specifically only on the transition surfaces of $ E $, cf.\,(\ref{eq:SubparInftyND}, \ref{eq:ZeroInftyND}, \ref{eq:SubparInftyNN}, \ref{eq:ZeroInftyNN}).
			
			If $ \Delta I = - \Delta E = \pm 1 $, a parabolic state is lost and a flat one gained, or vice-versa. Due to continuity and the typicality statement above, the loss and gain shall be ascribed to the same edge state turning from parabolic to flat, or vice-versa. 
			
			\item[(c)] If $E$ alone changes, the transition involves asymptotically flat states only. At $ k_x \to \pm \infty $, there can be at most one flat state at a time, cf.\,\eqref{eq:AsymptoteHeightNewVars}. These two facts, combined with $ \Delta E = \pm 2 $, cf.\,point (i), lead to the only logical conclusion that two asymptotes, one at $ k_x \to - \infty $ and the other one at $ k_x \to + \infty $, cross zero in opposite direction.
			
			\item[(d)] This case is different from the other one, in that the transition is not characterized by a change in $ \cal{V} $, but rather by loss of self-adjointness. This is to say, at transition there exists at least one $ k_x^* \in \bbR $ such that the boundary conditions $ k_x \mapsto \ul{A} (k_x) $, cf.\,\eqref{eq:ExplicitUlA}, have
			\begin{equation}
				\op{rk} \ul{A} (k_x^*) < 2 \,. \label{eq:LossSA}
			\end{equation}
			We shall assume that exactly one such point $ k_x^* $ exists, because the extension to finitely many (or even a countable infinity of) points is immediate.
			
			With this simplifying assumption, the claim has become that, if \eqref{eq:LossSA} holds, then
			\begin{equation}
				\bbR \subset \sigma (H^\# (k_x^*)) \,,
			\end{equation}
			and $B$ is ill-defined, yet it does not change while crossing the BC $ \ul{A} $.
			
			We prove the claim on the spectrum first. Since
			\begin{equation}
				\sigma_{\op{ess}} (H^\# (k_x^*)) = (-\infty, - \omega^+ (k_x^*)) \cup \{ 0 \} \cup (\omega^+ (k_x^*), +\infty) \,,
			\end{equation}
			cf.\,Eq.\,\eqref{eq:BandRim}, we need only prove
			\begin{equation}
				\bbR \setminus \sigma_{\op{ess}} (H^\# (k_x^*)) \subset \sigma (H^\# (k_x^*)) \,,
			\end{equation}
			Let $ \omega \in \bbR \setminus \sigma_{\op{ess}} (H^\# (k_x^*)) $ (namely, in the gap) and let $ \kappa $ be such that $ \omega = \omega_+ (k_x^*, \kappa) $. Then, a candidate eigenstate of energy $ \omega $ is the scattering state $ \psi_s = \tilde{\psi}_s (y; k_x^*, \kappa) \eul^{\im (k_x^* x - \omega t)} $, with
			\begin{equation}
				\tilde{\psi}_s = \Big( \alpha \hat{\psi}_{ \op{in} } \eul^{- \im \kappa y} + \beta \hat{\psi}_{ \op{out} } \eul^{\im \kappa y} + \gamma \hat{\psi}_{ \op{in} } \eul^{- \im \kappa_{\op{ev}} y} \Big) \,,
			\end{equation}
			cf.\,\eqref{eq:ScattState} with some choice of \textit{normalized} bulk sections $ \hat{\psi}_j $, $ (j = \op{in}, \ \op{out}, \ \op{ev}) $. W.l.o.g., we can set $ \gamma = 1 $, just like we set $ \alpha = 1 $ in Eq.\,\eqref{eq:ExplicitScattState}, while defining the scattering amplitude $ S $. The scattering state satisfies the given boundary conditions iff
			\begin{equation}
				\ul{A} M \Psi_s = 0 \ \longleftrightarrow \ \alpha \ul{A} M \Psi_{\op{in}} + \beta \ul{A} M \Psi_{\op{out}} + \ul{A} M \Psi_{\op{ev}} = 0 \,,
			\end{equation}
			cf.\,\eqref{eq:BCScattState}. We rearrange this equation as
			\begin{equation}
				\begin{pmatrix}
					\ul{A} M \Psi_{\op{in}} & \ul{A} M \Psi_{\op{out}}
				\end{pmatrix}
				\begin{pmatrix}
					\alpha \\
					\beta
				\end{pmatrix}
				= - \ul{A} M \Psi_{\op{ev}} \,.             
			\end{equation}
			By Cramer's rule,
			\begin{equation}
				\begin{gathered}
					\alpha = - \op{det} (\ul{A} M \Psi_{\op{ev}}, \, \ul{A} M \Psi_{\op{out}}) \,, \\
					\beta = - \op{det} (\ul{A} M \Psi_{\op{in}}, \, \ul{A} M \Psi_{\op{ev}}) \,.
				\end{gathered}
			\end{equation}
			By hypothesis, $ \op{rk} \ul{A} (k_x^*) < 2 $, whence
			\begin{equation}
				\alpha (k_x^*, \kappa) = \beta (k_x^*, \kappa) = 0 \,.
			\end{equation}
			The resulting state
			\begin{equation}
				\tilde{\psi}_s = \hat{\psi}_{\op{ev}} (k_x^*, \kappa) \eul^{\im \kappa_{\op{ev}} y}
			\end{equation}
			thus satisfies the boundary conditions, at $ k_x^* $, for all $ \omega $ in the gap, or equivalently all $ \kappa \in \bbC $ s.t. $ \omega = \omega (k_x^*, \kappa) $. It is moreover in $ L^2 (\bbR_+ \ni y) $ by $ \kappa_{\op{ev}} \in \im \bbR_+ $. This concludes this part of the proof.
			
			To prove the claim on $B$, we recall \eqref{eq:PPlusPMinus}:
			\begin{equation}
				B(U) = N(P_-) - N(P_+) \,,
			\end{equation}
			where $U$ is the orbit of $ \ul{A} $, $ N (\cdot) $ denotes the winding number of a curve $ \bbR \to \bbC $, and
			\begin{equation}
				P_\pm = \op{det} A_\pm = \op{det} (A_1 \pm A_2) \,,
			\end{equation}
			cf.\,\eqref{eq:APlusAMinus} and \eqref{eq:JuxtaposeA1A2}. By Claim \ref{claim:RkA1PlusA2}
			\begin{equation}
				\op{rk} \ul A < 2 \ \leftrightarrow \ \op{rk} A_\pm < 2 \,.
			\end{equation}
			Thus, $ \op{rk} \ul A (k_x^*) < 2 $ iff
			\begin{equation}
				P_\pm (k_x^*) = \op{det} (A_\pm (k_x^*)) = 0 \,,
			\end{equation}
			namely loss of self-adjointness is equivalent to both curves $ P_\pm : \bbR \to \bbC $ crossing the origin, in which case $B$ is ill-defined. This is the first claim. The second one is that $B$ does nonetheless not change in value, when crossing a locally non self-adjoint BC. We prove the statement family-wise. It is trivial in DD, so we focus on ND and NN. In ND, loss of self-adjointness occurs at the single (if existent) solution $ k_x^* $ of
			\begin{equation}
				(a_1 \wedge a_1') + k_x (a_1 \wedge b_1) = 0 \,,
			\end{equation}
			cf.\,\eqref{eq:SAAEND}. It however produces no change in $B$, due to \eqref{eq:PTilde}. In NN, all boundary conditions are everywhere self-adjoint, cf.\,Lm.\,\ref{lem:MaxRank} and Rem.\,\ref{rem:d21NeqZero}, and no transitions in the sense of this point thus occur.
		\end{itemize}
	\end{itemize}
\end{proof}
\begin{proof}[Proof of Cor.\,\ref{cor:ViolationMechanism}]
	Since bulk-edge correspondence was defined as $ C_+ = M $, cf.\,Def.\,\ref{def:BEC}, only a change in $ C_+ $ or $M$ can alter its status. However, $ C_+ $ is a constant, independent of boundary conditions. Thus, $ \Delta M \neq 0 $ is the only possible violation mechanism. Out of the allowed transitions (a)-(d) of Thm.\,\ref{thm:transitions}, only those of type (b) produce
	\begin{equation}
		\Delta M = \Delta P + \Delta I = 0 \pm 1 \neq 0 \,,
	\end{equation} 
	whence the claim.
\end{proof}

\section{Acknowledgements} \label{sec:Acknowledgements}

The authors heartily thank J. Kellendonk for numerous illuminating discussions, especially on the role of von Neumann unitaries, and C. Tauber for his careful reading of the manuscript and valuable remarks. They moreover thank G. C. Thiang and again C. Tauber for discussions at an early stage of the work.

\appendix

\section{Further facts on boundary conditions} \label{app:SelfAdj}

This appendix recalls the results of \cite{GJT21} on self-adjointness of $ 2 \times 6 $ boundary conditions $ A $, cf.\,Def.\,\ref{def:OperativeBC}, previously displayed in Sect.\,\ref{subsec:DetailsBCs}. Next, a proof is provided for Lm.\,\ref{lem:Dictionary} and Props.\,\ref{prop:BCs}, \ref{prop:PHSBCs}, \ref{prop:UBCs}. We later turn our attention to failure of self-adjointness at a single fiber $ k_x $, signaled by $ A(k_x) $ failing to be maximal rank. Whether a boundary condition $ A $ is maximal rank a.e.\,or everywhere in $ k_x $ is the question tackled in Prop.\,\ref{lem:MaxRank} below.
\vspace{1\baselineskip}

As anticipated, we start by reviewing the formalism and results of \cite{GJT21}. Recall Eq.\,\eqref{eq:BoundaryCondition}: A scattering state $ \psi_s = \tilde{\psi}_s (y) \eul^{\im (k_x x - \omega t)} $, cf.\,\eqref{eq:ScattState} in Def.\,\ref{def:ScattState}, satisfies the boundary condition $ A (k_x) $ if
\begin{equation} \label{eq:OrigBC}
	A (k_x) \Psi = 0 \,,
\end{equation} 
where $ \Psi = ( \tilde{\psi}_s (0), \tilde{\psi}_s' (0))^T $ and 
\begin{equation} \label{eq:A}
	A (k_x) = A_c + k_x A_k \,,
\end{equation}
as per Eq.\,\eqref{eq:NewA}. The matrix $ A $ has precisely two rows because two equations suffice to uniquely identify a solution $ \varphi \in L^2 (\bbR_+)^{\otimes 3} $ of $ H(k_x) \varphi = \omega \varphi $, cf.\,Eq.\,\eqref{eq:EdgeHamPrecursor}. On the other hand, one equation is not enough, so that
\begin{equation}
	\op{rk} A(k_x) \overset{!}{=} 2 \label{eq:RankCondition}
\end{equation}
a.e.\,in $ k_x $. Differently put, all self-adjoint boundary conditions are represented by matrices $ A $ of maximal rank.

It is proven in \cite{GJT21} that a given BC $ A: k_x \mapsto A(k_x) $, cf.\,again Def.\,\ref{def:OperativeBC}, encodes for a self-adjoint realization $ H^{\#} $ of the formal differential operator $ H(k_x) $ if and only if
\begin{subequations}
	\begin{gather} 
		A (k_x) N = 0 \,, \qquad A (k_x) \hat{\Omega} A^* (k_x) = 0 \,, \label{eq:CondSelfAdjKx} \\
		\ds N = 
		\begin{pmatrix}
			\nu & 0 \\
			0 & 0 \\
			0 & 0 \\
			0 & 1 \\
			1 & 0 \\
			0 & 0 
		\end{pmatrix}
		\,, \qquad \hat{\Omega} =
		\begin{pmatrix}
			0 & 0 & -\lambda & 0 & 0 & 0 \\
			0 & 0 & 0 & 0 & 0 & - \nu^{-1} \\
			- \lambda & 0 & 0 & 0 & \lambda \nu & 0 \\
			0 & 0 & 0 & 0 & 0 & 0 \\
			0 & 0 & \lambda \nu & 0 & 0 & 0 \\
			0 & - \nu^{-1} & 0 & 0 & 0 & 0 \\
		\end{pmatrix}
		\,, \qquad \lambda = \frac{1}{1 + \nu^2} \label{eq:NOmega}
	\end{gather}
\end{subequations}
almost everywhere in $ k_x $. Equivalently, cf.\,\eqref{eq:A},
\begin{subequations}
	\begin{gather} 
		A_c N = 0 \ , \hspace{12pt} A_k N = 0 \,, \label{eq:CondSelfAdj1} \\
		A_c \hat{\Omega} A^*_c = 0, \hspace{10pt} A_c \hat{\Omega} A_k^* - A_k \hat{\Omega} A_c^* = 0, \hspace{10pt} A_k \hat{\Omega} A_k^* = 0 \,. \label{eq:CondSelfAdj2}
	\end{gather}
\end{subequations}

Eq.\,\eqref{eq:CondSelfAdj1} implies that two out of the six columns of $ A $ are redundant, so that one may generically write
\begin{equation} \label{eq:AcAk}
	A_c = (a_1, \ a_1', \ a_2', \ 0, \ - \nu a_1, \ \nu a_2) \,, \hspace{13pt}
	A_k = (0, \ b_1, \ b_2, \ 0, \ 0, \ 0) \,,
\end{equation}
where $ a_1, a_2 , a_1', a_2', b_1, b_2  \in \bbC^2 $ and $ (\cdot, \cdot) $ denotes the horizontal juxtaposition of two (or more) matrices. This is to say, if $ X $ is an $ l \times m $ and $ Y $ is an $ l \times n $, $ Z = (X,Y) $ is an $ l \times (m+n) $ built up by \enquote{gluing} $ Y $ to the right of $ X $. Eq.\,\eqref{eq:AcAk} is equivalent to \eqref{eq:GeneralTwoBySix}, whence our claim \eqref{eq:ExplicitA0AxAy} on the general form of $ A_0, A_x, A_y $ follows.

The structure \eqref{eq:A} of $ A $, with only four independent columns, prompts the simpler rewriting as a $ 2 \times 4 $ matrix $ \ul{A} $
\begin{gather} 
	\ul{A} (k_x) = \ul{A}_{c} + k_x \ul{A}_k \,, \nonumber\\
	\ul{A}_{c} = (a_1, \ a_1', \ a_2', \ a_2), \hspace{10pt} \ul{A}_k = (0, \ b_1, \ b_2, \ 0) \,, \label{eq:AUnderlined}
\end{gather}
and the matrix $ M $ of \eqref{eq:MatM} is read off by direct comparison of $ A_c, \, A_k $ with $ \ul{A}_{c}, \ \ul{A}_{k} $, cf.\,(\ref{eq:AcAk}, \ref{eq:AUnderlined}).

Any BC $ A = \ul{A} M $, with $ \ul{A} $ as in Eq.\,\eqref{eq:AUnderlined}, satisfies \eqref{eq:CondSelfAdj1} by construction. On the other hand, \eqref{eq:CondSelfAdj2} is yet to be imposed on $ \ul{A} $. Again by $ A = \ul{A} M $, it translates to
\begin{equation}
	A (k_x) \hat{\Omega} A^* (k_x) = \ul{A} (k_x) M \hat{\Omega} M^* \ul{A}^* (k_x) \equiv - \ul{A} (k_x) \Omega \ul{A}^* (k_x) \overset{!}{=} 0 \label{eq:RecoveringOmega}
\end{equation}
almost everywhere in $ k_x $, where
\begin{equation}
	\Omega \coloneqq - M \hat{\Omega} M^* =
	\begin{pmatrix}
		0 & \id \\
		\id & 0
	\end{pmatrix} \,,
\end{equation}
the second identity following from Eqs.\,(\ref{eq:MatM}, \ref{eq:NOmega}) and matrix multiplication. By \eqref{eq:RecoveringOmega}, we have thus recovered \eqref{eq:CondSelfAdjKxNew}.

We can now move on to proving Lm.\,\ref{lem:Dictionary} and Props.\,\ref{prop:BCs}, \ref{prop:PHSBCs}, starting with the former.
\begin{proof}[Proof of Lm.\,\ref{lem:Dictionary}]
	Recall Eq.\,\eqref{eq:ExplicitA0AxAy}, namely
	\begin{equation}
		A_y = (0, \ - \nu a_1, \ \nu a_2) \,.
	\end{equation}
	Then
	\begin{itemize}[itemindent=1em]
		\item[DD:] $ \op{rk} (A_y) = 0 $ iff $ \op{rk} (- \nu a_1, \ \nu a_2) = 0 $, i.e.\,, $ a_1 = a_2 = 0 $;
		
		\item[NN:] $ \op{rk} (A_y) = 2 $ iff $ \op{rk} (- \nu a_1, \ \nu a_2) = 2 $, i.e.\,, $ a_1, a_2 $ linearly independent.
	\end{itemize}
	If $ \op{rk} (A_y) = 1 $, there exists an invertible matrix
	\begin{equation}
		G = 
		\begin{pmatrix}
			\bar{g}_1^T \\
			\bar{g}_2^T
		\end{pmatrix} \,, 
	\end{equation}
	with $ g_1, g_2 \in \bbC^2 $ and $ g_1, g_2 $ l.i.\,, such that
	\begin{equation}
		G (- \nu a_1, \ \nu a_2) = 
		\begin{pmatrix}
			- \nu g_1 \cdot a_1 & \nu g_1 \cdot a_2 \\
			- \nu g_2 \cdot a_1 & \nu g_2 \cdot a_2 
		\end{pmatrix}
	\end{equation}
	has one row equal to zero (w.l.o.g.\,, the second one), where $ \cdot $ denotes inner product in $ \bbC^2 $. If the second row is zero, $ a_1, a_2 \perp g_2 $ in $ \bbC^2 $, and thus $ a_1 \parallel a_2 $, i.e.\,, linearly dependent.
	
	The further distinction between ND and DN arises as follows. Since $ g_1 $ is l.i. from $ g_2 $ and non-zero, the inner products $ g_1 \cdot a_1, g_1 \cdot a_2 $ appearing in the first row are only zero for $ a_1 = 0 $ or $ a_2 = 0 $, respectively. If a term in $\partial_y u$ persists in the BC after elimination of the second row, it means that the $ [G A_y]_{12} \neq 0 $, whence $ g_1 \cdot a_1 \neq 0 $ and thus $ a_1 \neq 0 $. This characterizes family ND. Family DN is obtained in the same way, requiring $ \partial_y v $ to persist. 
\end{proof}
Prop.\,\ref{prop:BCs} reduces to two simpler propositions by means of the following reasoning. The condition to be met is
\begin{equation}
	A_1 (k_x) A_2^* (k_x) + A_2 (k_x) A_1^* (k_x) = 0 \label{eq:CondOnAAgain}
\end{equation}
a.e.\,in $ k_x $, cf.\,\eqref{eq:CondOnA}. Recall the decomposition $ A_i (k_x) = A^0_i + k_x B_i $, cf.\,\eqref{eq:A1A2}. If one lets $ |1\rangle = (1,0)^\text{T}, \ |2\rangle = (0,1)^\text{T} $, then
\begin{equation}
	A_1^0 = | a_1 \rangle \langle 1 | + | a_1' \rangle \langle 2  | \,, \qquad A_2^0 = | a_2 \rangle \langle 2 | + | a_2' \rangle \langle 1 | \,, \qquad B_1 = | b_1 \rangle \langle 2 | \,, \qquad B_2 = | b_2 \rangle \langle 1 | \,. \label{eq:A01B1Dirac}
\end{equation}
The l.h.s.\,of \eqref{eq:CondOnAAgain} can be rewritten as
\begin{equation} 
	A_1 (k_x) A_2^* (k_x) + A_2 (k_x) A_1^* (k_x) \equiv C_0 + C_1 k_x + C_2 k_x^2 \,, \label{eq:C0C1C2}
\end{equation}
and it equals zero a.e.\,in $ k_x $ only if the (matrix) coefficients 
\begin{equation}
	C_0 = A^0_1 A^{0*}_2 + \ \text{h.c.} \,, \qquad C_1 = B_1 A^{0*}_2 + A^0_1 B_2^* + \ \text{h.c.} \,, \qquad C_2 = B_1 B_2^* + \ \text{h.c.} \label{eq:NewC0C1C2}
\end{equation}
are identically zero (h.c.\,standing for \enquote{hermitian conjugate}). That $ C_2 \equiv 0, \ \forall b_1, b_2 $ is immediate from \eqref{eq:A01B1Dirac}. Whether or not $ C_0 = 0 $ ($ C_1 = 0 $) is the content of Prop.\,\ref{prop:C0isZero} (Prop. \ref{prop:C1isZero}). Their proof rests on the following preliminary claim, to be proven later.
\begin{claim}
	\label{cl:Dirac}
	\begin{enumerate}
		\item[(i)] If $ \varphi_i \in \bbC^2, \ (i = 1,2) $ are linearly independent, so are $ | \varphi_i \rangle \langle \varphi_j |, \ (i,j=1,2) $.
		
		\item[(ii)] Let $ \bbC^2 \ni \psi \neq 0 $. Then,
		\begin{equation}
			| \varphi \rangle \langle \psi | + | \psi \rangle \langle \varphi | = 0
		\end{equation}
		iff
		\begin{equation}
			| \varphi \rangle = \im \lambda | \psi \rangle \,,
		\end{equation}
		for some $ \lambda \in \bbR $.
	\end{enumerate}
\end{claim}

\begin{proposition}
	\label{prop:C0isZero}
	Depending on the cases DD through NN, cf.\,Lm.\,\ref{lem:Dictionary}, the equation
	\begin{equation}
		C_0 = 0 \label{eq:C0isZero}
	\end{equation}
	amounts to
	\begin{itemize}
		\item[DD:] $ a_1', a_2' \in \bbC^2 $ (arbitrary);
		
		\item[ND:] For some $ \lambda' \in \bbR $,
		\begin{equation}
			\bar{\alpha} a_1' + a_2' = \im \lambda' a_1 \,;
		\end{equation} 
		
		\item[DN:] Same as ND up to $ a_1 \leftrightarrow a_2 $;
		
		\item[NN:] For some $ \mu' \in \bbC $ and $ \lambda_1', \lambda_2' \in \bbR $,
		\begin{equation}
			\begin{array}{lcl}
				a_1' & = & \mu' a_1 + \im \lambda_1' a_2 \\
				a_2' & = & - \bar{\mu}' a_2 + \im \lambda_2' a_1 \,.
			\end{array} 
		\end{equation} 
	\end{itemize}
\end{proposition} 
\begin{proof}
	Using the previously introduced Dirac notation, $ C_0 = 0 $ reads
	\begin{equation}
		| a_1' \rangle \langle a_2 | + | a_2 \rangle \langle a_1' | + | a_2' \rangle \langle a_1 | + | a_1 \rangle \langle a_2' | = 0 \,. \label{eq:C0ZeroDirac}
	\end{equation}
	Consider now cases DD--NN.
	\begin{itemize}[itemindent=1em]
		\item[DD:] Is evident;
		
		\item[ND:] Eq.\,\eqref{eq:C0ZeroDirac} becomes
		\begin{equation}
			( \bar{\alpha} | a_1' \rangle + | a_2' \rangle ) \langle a_1 | + | a_1 \rangle ( \alpha \langle a_1' | + \langle a_2' | ) = 0\,.
		\end{equation}
		By (ii) of Claim \ref{cl:Dirac}:
		\begin{equation}
			\bar{\alpha} | a_1' \rangle + | a_2' \rangle = \im \lambda' | a_1 \rangle \,,
		\end{equation}
		as wished for;
		
		\item[DN:] Same as ND up to $ a_1 \leftrightarrow a_2 $;
		
		\item[NN:] Here, $ a_1, a_2 $ are a basis of $ \bbC^2 $, whence
		\begin{align}
			| a_1' \rangle &= \mu_1 | a_1 \rangle + \gamma_1 | a_2 \rangle \nonumber \\
			| a_2' \rangle &= \mu_2 | a_2 \rangle + \gamma_2 | a_1 \rangle\,, \label{eq:a1pa2p}
		\end{align}
		$ \mu_i,\gamma_i \in \bbC $.
		
		Then
		\begin{equation}
			| a_1' \rangle \langle a_2 | + | a_2 \rangle \langle a_1' | = \mu_1 | a_1 \rangle \langle a_2 | + \bar{\mu}_1 | a_2 \rangle \langle a_1 | + (\gamma_1 + \bar{\gamma}_1) | a_2 \rangle \langle a_2 | \,,
		\end{equation}
		and likewise for the other two terms of Eq.\,\eqref{eq:C0ZeroDirac}.
		
		In view of linear independence of $ a_1, a_2 $, the equations
		\begin{equation}
			\mu_1 = - \bar{\mu}_2 \,, \qquad \gamma_1 + \bar{\gamma}_1 = 0 \,, \qquad \gamma_2 + \bar{\gamma}_2 = 0
		\end{equation}
		follow by (i) of Claim\ref{cl:Dirac}. Therefore, setting
		\begin{equation}
			\mu_1 = - \bar{\mu}_2 \equiv \mu' \,, \qquad \gamma_1 = \im \lambda_1' \,, \qquad \gamma_2 = \im \lambda_2' \,,
		\end{equation}
		the desired conclusion is read from \eqref{eq:a1pa2p}.
	\end{itemize}
\end{proof}
\begin{proposition}
	\label{prop:C1isZero}
	Depending on the cases DD through NN, cf.\,Lm.\,\ref{lem:Dictionary}, the equation
	\begin{equation}
		C_1 = 0 \,, \label{eq:C1isZero}
	\end{equation}
	with $ C_1 $ as in \eqref{eq:NewC0C1C2}, amounts to
	\begin{itemize}[itemindent=1em]
		\item[DD:] $ b_1, b_2 \in \bbC^2 $ (arbitrary);
		
		\item[ND:] For some $ \lambda' \in \bbR $,
		\begin{equation}
			\bar{\alpha} b_1 + b_2 = \im \lambda a_1 \,;
		\end{equation} 
		
		\item[DN:] Same as ND up to $ a_1 \leftrightarrow a_2 $;
		
		\item[NN:] For some $ \mu \in \bbC $ and $ \lambda_1, \lambda_2 \in \bbR $,
		\begin{equation}
			\begin{array}{lcl}
				b_1 & = & \mu a_1 + \im \lambda_1 a_2 \\
				b_2 & = & - \bar{\mu} a_2 + \im \lambda_2 a_1 \,.
			\end{array} 
		\end{equation} 
	\end{itemize}
\end{proposition}
\begin{proof}
	Equation \eqref{eq:C1isZero} consists of two halves with interchanged indices $ 1 \leftrightarrow 2 $. It therefore suffices to study
	\begin{equation}
		B_1 A^{0*}_2 + A^0_1 B_2^* = 0
	\end{equation}
	in the various cases. Switching to Dirac notation, the previous identity reads
	\begin{equation}
		| b_1 \rangle \langle a_2 | + | a_2 \rangle \langle b_1 | + | b_2 \rangle \langle a_1 | + | a_1 \rangle \langle b_2 | = 0\,. \label{eq:C1ZeroDirac}
	\end{equation}
	The latter is formally equivalent to \eqref{eq:C0ZeroDirac} upon identification $ b_1 \leftrightarrow a_1' $, $ b_2 \leftrightarrow a_2' $. The claims follow by the same reasoning of Prop. \ref{prop:C0isZero}.
\end{proof}
Once the proof of Claim \ref{cl:Dirac} is given, considering Props.\,\ref{prop:C0isZero} and \ref{prop:C1isZero} together yields Prop.\,\ref{prop:BCs} as a corollary.
\begin{proof}[Proof of Claim \ref{cl:Dirac}]
	\begin{itemize}
		\item[(i)] Let $ \langle \varphi_i | $ be the dual basis of $ | \varphi_i \rangle $, i.e.
		\begin{equation}
			\langle \varphi_i | \varphi_j \rangle = \delta_{i j}.
		\end{equation}
		Then
		\begin{equation}
			\sum_{i,j} \lambda_{i j} | \varphi_i \rangle \langle \varphi_j | = 0
		\end{equation}
		implies $ \lambda_{i_0 j_0} = 0 $ by taking the matrix element between $ \langle \varphi_{i_0} | $ and $ | \varphi_{j_0} \rangle $.
		
		\item[(ii)] Let $ \psi \neq 0 $ and assume $ | \psi \rangle = \gamma | \varphi \rangle, \ \gamma \in \bbC $. Then
		\begin{equation}
			| \varphi \rangle \langle \psi | + | \psi \rangle \langle \varphi | = \bar{\gamma} | \varphi \rangle \langle \varphi | + \gamma | \varphi \rangle \langle \varphi | = 0 \ \Leftrightarrow \ \gamma = \im \lambda, \ \lambda \in \bbR.
		\end{equation}
		Thus $ | \varphi \rangle = \im \lambda | \psi \rangle $, implying
		\begin{equation}
			| \varphi \rangle \langle \psi | + | \psi \rangle \langle \varphi | = 0 \,.
		\end{equation}
		The other direction follows from point (i).
	\end{itemize}
\end{proof}

Prop.\,\ref{prop:PHSBCs} is a statement about the particle-hole symmetric subsets of families DD--NN. However, by Def.\,\ref{def:PHS} and Eq.\,\eqref{eq:USymmetry}, particle-hole symmetry is a property of the orbit $ [\ul{A}] $, cf.\,\eqref{eq:Cosets}, under point-wise (in $k_x$) multiplication by $ \op{GL} (2,\bbC) $, rather than of the single map $ k_x \mapsto \ul{A} (k_x) $. Equivalently, it is a property of the von Neumann unitary $ k_x \mapsto U(k_x) $, cf.\,Def.\,\ref{def:OperativeBC}, once the claimed bijective correspondence between $ U $ and $ [\ul{A}] $ is proven.

We thus adopt the following strategy. First, the correspondence $ U \leftrightarrow [\ul{A}] $ is shown. Then, the image of the four families DD--NN under $ \upsilon: [\ul{A}] \mapsto U $, cf.\,\eqref{eq:UFormula}, is found by explicit calculation. Particle-hole symmetry is imposed on the resulting $ k_x \mapsto U(k_x) $ by
\begin{equation}
	U (k_x) = \overline{U (-k_x)} \,, \label{eq:USymmetryAgain}
\end{equation}
cf.\,\eqref{eq:USymmetry}, completing the proof of Prop.\,\ref{prop:UBCs}. The constraints stemming from \eqref{eq:USymmetryAgain} are finally pulled back to $ k_x \mapsto \ul{A} (k_x) $, yielding Prop.\,\ref{prop:PHSBCs}.
\begin{lemma} \label{lem:UpsilonBijection}
	The map $ \upsilon: \ul{A} \mapsto U $, defined by Eq.\,\eqref{eq:UFormula}, between boundary conditions $ \ul{A} $ and von Neumann unitaries $ U $, cf.\,Def.\,\ref{def:OperativeBC}, is:
	\begin{itemize}
		\item[(i)] Well-defined as a map on orbits $ [\ul{A}] $, cf.\,Eq.\,\eqref{eq:Cosets};
		
		\item[(ii)] A bijection $ [\ul{A}] \leftrightarrow U $.
	\end{itemize}
\end{lemma}
This Lemma rests on the following claim.
\begin{claim} \label{claim:RkA1PlusA2}
	Let 
	\begin{equation}
		\ul{A} = (A_1, A_2) \in \op{Mat}_{2 \times 4} (\bbC) \,, \label{eq:FormOfA}
	\end{equation}
	with $ A_i \in \op{Mat}_{2} (\bbC) $, $ (i = 1,2) $ and
	\begin{equation}
		A_1 A_2^* + A_2 A_1^* = 0 \,. \label{eq:PropOfA}
	\end{equation} 
	Then any among $ \op{rk} \ul{A} = 2 $, $ \op{rk} (A_1 + A_2) = 2 $, $ \op{rk} (A_1 - A_2) = 2 $ implies the others.
\end{claim}
\begin{proof}
	For any matrix $M$, we have
	\begin{equation}
		\op{rk} M = \op{rk} (M M^*) \,, \label{eq:RankProperty}
	\end{equation}
	both being equal to $ \op{rk} M^* $. Moreover, for any matrix $M$ having two rows, we have
	\begin{equation}
		\op{rk} (M M^*) = 2 \ \iff \ \det (M M^*) \neq 0 \,,
	\end{equation}
	since $ M M^* $ is a square matrix of order $2$. We note that
	\begin{equation}
		\begin{gathered}
			\ul{A} \, \ul{A}^* = A_1 A_1^* + A_2 A_2^* \,, \\
			(A_1 \pm A_2) (A_1 \pm A_2)^* = A_1 A_1^* + A_2 A_2^* \,,
		\end{gathered}
	\end{equation}
	where the first equation comes from \eqref{eq:FormOfA} and the second one from \eqref{eq:PropOfA}.
	
	Applying now the above to $ M = \ul{A} $ and $ M = A_1 + A_2 $, or $ M = \ul{A} $ and $ M = A_1 - A_2 $, the claim follows.
\end{proof}
\begin{proof}[Proof of Lemma \ref{lem:UpsilonBijection}]
	\begin{itemize}
		\item[(i)] The map $ \upsilon $ is well-defined on orbits because
		\begin{align}
			\upsilon (G \ul{A}) &= (G A_1 + G A_2)^{-1} (G A_1 - G A_2) = (A_1 + A_2)^{-1} G^{-1} G (A_1 - A_2) \nonumber \\
			&= (A_1 + A_2)^{-1} (A_1 - A_2) = \upsilon (\ul{A}) \,,
		\end{align}
		namely $ \ul{A} $ and $ G \ul{A} $ share the same image through $ \upsilon $, for all $ G: k_x \mapsto G(k_x) $ with $ G(k_x) $ invertible a.e.\,in $ k_x $.
		
		\item[(ii)] By picking a target space coincident with its range, we make $ \upsilon: [\ul{A}] \mapsto U $ surjective by construction. All that is left to check is injectivity, which holds iff
		\begin{equation}
			\upsilon (\ul{A}) = U = \tilde{U} = \upsilon (\ul{\tilde{A}}) \ \Rightarrow \ \ul{A}, \ul{\tilde{A}} \in [\ul{A}] \,.
		\end{equation}
		The equality $ U = \tilde{U} $ means
		\begin{equation}
			(A_1 + A_2)^{-1} (A_1 - A_2) = (\tilde{A}_1 + \tilde{A}_2)^{-1} (\tilde{A}_1 - \tilde{A}_2) \,. \label{eq:UeqTildeU}
		\end{equation}
		By assumption of self-adjointness, $ \ul{A}, \ul{\tilde{A}} $ are maximal rank a.e.\,in $k_x$. Then, by Claim \ref{claim:RkA1PlusA2} so are $ A_1 + A_2 $ and $ \tilde{A}_1 + \tilde{A}_2 $. Therefore, there exists $G: k_x \mapsto G(k_x)$, $ G(k_x) $ invertible a.e.\,in $ k_x $, such that
		\begin{equation}
			(A_1 + A_2) = G (\tilde{A}_1 + \tilde{A}_2) \,. \label{eq:G}
		\end{equation}
		Indeed, $G$ is explicitly given by
		\begin{equation}
			G = (A_1 + A_2) (\tilde{A}_1 + \tilde{A}_2)^{-1} \,.
		\end{equation}
		Writing $ A_1+A_2 $ as in \eqref{eq:G} inside of \eqref{eq:UeqTildeU} results in
		\begin{gather}
			(\tilde{A}_1 + \tilde{A}_2)^{-1} G^{-1} (A_1 - A_2) = (\tilde{A}_1 + \tilde{A}_2)^{-1} (\tilde{A}_1 - \tilde{A}_2) \nonumber \\
			\longleftrightarrow \quad (A_1 - A_2) = G (\tilde{A}_1 - \tilde{A}_2) \,. \label{eq:G2}
		\end{gather}
		Summing and subtracting Eqs.\,\eqref{eq:G}, \eqref{eq:G2} yields
		\begin{equation}
			A_1 = G \tilde{A}_1 \quad \land \quad A_2 = G \tilde{A}_2 \,,
		\end{equation}
		namely $ \ul{A} = G \ul{\tilde{A}} $ or $ \ul{\tilde{A}} = G^{-1} \ul{A} $ with $ G $ invertible, as desired.
	\end{itemize}	
\end{proof}
The calculations leading to the proof of Prop.\,\ref{prop:UBCs} are simplified by recalling the following elementary facts in linear algebra, gathered in a claim. We recall the adjugate of a matrix $ M \in \op{Mat}_n (\bbC) $ for $ n=2 $:
\begin{equation}
	M =
	\begin{pmatrix}
		m_{11} & m_{12} \\
		m_{21} & m_{22}
	\end{pmatrix} \,, \qquad
	\op{adj} (M) =
	\begin{pmatrix}
		m_{22} & -m_{12} \\
		-m_{21} & m_{11}
	\end{pmatrix} \,. \label{eq:DefAlmostInverse}
\end{equation} 
\begin{claim} \label{claim:Inverse}
	The operation $ \op{adj} $ enjoys the following properties:
	\begin{enumerate}
		\item $ M^{-1} = (\det M)^{-1} \op{adj} (M) $, for all $M \in \op{GL} (n, \bbC) $;
		
		\item $ \op{adj} (M+N) = \op{adj} (M) + \op{adj} (N) $, for all $ M,N \in \op{Mat}_n (\bbC) $;
		
		\item Let $ M, N \in \op{Mat}_2 (\bbC) $ be given in terms of column vectors $ m_i, n_i \in \bbC^2 \ (i=1,2) $, i.e., $ M = (m_1,m_2) $ and $ N = (n_1, n_2) $. Then
		\begin{equation}
			\op{adj} (M) \cdot N - \op{adj} (N) \cdot M = 
			\begin{pmatrix}
				- m_1 \wedge n_2 - m_2 \wedge n_1 & - 2 m_2 \wedge n_2 \\
				2 m_1 \wedge n_1 & m_1 \wedge n_2 + m_2 \wedge n_1 
			\end{pmatrix} \,, \label{eq:AdjProp3}
		\end{equation}
		where 
		\begin{equation}
			z \wedge w = 
			\begin{pmatrix}
				z_1 \\
				z_2
			\end{pmatrix} \wedge
			\begin{pmatrix}
				w_1 \\
				w_2
			\end{pmatrix}
			\coloneqq \det (z,w) = z_1 w_2 - w_1 z_2 \,, \label{eq:DefWedge}
		\end{equation}
		is the wedge product $ \wedge: \bbC^2 \times \bbC^2 \to \bbC^2 $.
	\end{enumerate}
\end{claim}
\begin{proof}
	1.\,is known and 2.\,is obvious.
	
	3.\,Upon writing $ m_i = (m_{i1}, m_{i2})^T $ in terms of components, the entries of $M$ end up transposed as compared to \eqref{eq:DefAlmostInverse}, and likewise for $N$. The property is then shown by direct calculation:
	\begin{align}
		\op{adj} (M) \cdot N &= 
		\begin{pmatrix}
			m_{22} & -m_{21} \\
			-m_{12} & m_{11} 
		\end{pmatrix}
		\begin{pmatrix}
			n_{11} & n_{21} \\
			n_{12} & n_{22} 
		\end{pmatrix} \nonumber \\
		&= 
		\begin{pmatrix}
			-m_2 \wedge n_1 & -m_2 \wedge n_2 \\
			m_1 \wedge n_1 & m_1 \wedge n_2 
		\end{pmatrix} \,,
	\end{align}	
	cf.\,Eqs.\,(\ref{eq:DefAlmostInverse}, \ref{eq:DefWedge}).
	
	By interchanging $ M,N $ and taking differences, we obtain \eqref{eq:AdjProp3}
\end{proof}
\begin{proof}[Proof of Prop.\,\ref{prop:UBCs}]
	(i) Using the adjugation map $ \op{adj} (\cdot) $, cf.\,\eqref{eq:DefAlmostInverse}, and its properties $1$ and $2$ in Claim \ref{claim:Inverse}, we notice that
	\begin{align}
		(A_1 + A_2)^{-1} (A_1 - A_2) &= \frac{1}{\det (A_1 + A_2)} ( \op{adj} (A_1) + \op{adj} (A_2) ) (A_1 - A_2) \label{eq:USmartFormula} \\
		&= \frac{1}{\det (A_1 + A_2)} \big( \op{adj} (A_1) A_1 - \op{adj} (A_1) A_2 + \op{adj} (A_2) A_1 - \op{adj} (A_2) A_2 \big) \nonumber \\
		&= \frac{1}{\det (A_1 + A_2)} \big( (\det A_1 - \det A_2) \id_2 + \op{adj} (A_2) A_1 - \op{adj} (A_1) A_2 \big) \,. \nonumber 
	\end{align}
	The calculation of
	\begin{equation}
		U = (A_1 + A_2)^{-1} (A_1 - A_2)
	\end{equation}
	thus proceeds by explicitly writing out $ A_1, A_2, A_1+A_2 $ and their determinants in cases DD--NN, and applying the formula \eqref{eq:USmartFormula}. The following expressions
	\begin{equation}
		\ul{A} (k_x) = (A_1 (k_x), \, A_2 (k_x)) \,, \quad A_1 (k_x) = (a_1, \, a_1' + k_x b_1) \,, \quad A_2 (k_x) = (a_2' + k_x b_2, \, a_2) \,,
	\end{equation}
	cf.\,(\ref{eq:JuxtaposeA1A2}, \ref{eq:ExplicitUlA}) will come in handy.
	\begin{itemize}[itemindent=1em]
		\item[DD:] The starting point is Eq.\,\eqref{eq:AUnderlined}. By Lm.\,\ref{lem:Dictionary}, family DD has $ a_1 = a_2 = 0 $. Self-adjointness is then achieved for all $ a_1',a_2',b_1,b_2 \in \bbC^2 $, cf.\,Prop.\,\ref{prop:BCs}. The general element $ \ul{A}: k_x \mapsto \ul{A} (k_x) $ thus has 
		\begin{equation} \label{eq:ACaseDD}
			\ul{A} (k_x) = 
			(0, \, a_1' + k_x b_1, \, a_2' + k_x b_2, \, 0 )
			\equiv (A_1 (k_x), \, A_2 (k_x)) \,.
		\end{equation}
		The associated von Neumann unitary $ U = \upsilon(\ul{A}): k_x \mapsto U(k_x) $ is given by
		\begin{equation}
			U (\ul{A} (k_x)) \coloneqq (A_1(k_x) +  A_2 (k_x))^{-1} (A_1 (k_x) - A_2 (k_x)) \,,
		\end{equation}
		cf.\,\eqref{eq:UFormula}, where explicitly
		\begin{align}
			A_1(k_x) +  A_2 (k_x) &= 
			(a_2' + k_x b_2, \, a_1' + k_x b_1) \,, \nonumber \\ 
			A_1(k_x) -  A_2 (k_x) &= 
			(-a_2' - k_x b_2, \, a_1' + k_x b_1) \,,
		\end{align}
		by Eq.\,\eqref{eq:ACaseDD}.
		
		In this simple case, there is no need to resort to the general formula \eqref{eq:USmartFormula}. Indeed, observing that
		\begin{equation}
			A_1(k_x) -  A_2 (k_x) = (A_1(k_x) +  A_2 (k_x)) \cdot 
			\begin{pmatrix}
				-1 & 0 \\
				0 & 1 \\
			\end{pmatrix} \,,
		\end{equation}
		immediately yields
		\begin{equation}
			U (k_x) = 
			\begin{pmatrix}
				-1 & 0 \\
				0 & 1
			\end{pmatrix} \equiv J \label{eq:UCaseDD}
		\end{equation}
		by Eq.\,\eqref{eq:UFormula}. Eq.\,\eqref{eq:UCaseDD} holds for all $ k_x \in \bbR $ and $ a_1',a_2',b_1,b_2 \in \bbC^2 $.
		
		\item[ND:] Lm.\,\ref{lem:Dictionary} and self-adjointness, cf.\,Prop.\,\ref{prop:BCs}, produce the general form
		\begin{equation}
			\ul{A} (k_x) = (a_1, \, a_1' + k_x b_1, \, \im (k_x \lambda + \lambda') a_1 - \bar{\alpha} (a_1' + k_x b_1), \, \alpha a_1 ) \equiv (A_1 (k_x), \, A_2 (k_x)) \label{eq:ACaseND}
		\end{equation}
		for $ \ul{A} (k_x) $, cf.\,\eqref{eq:AUnderlined}. It follows that
		\begin{equation}
			\begin{array}{lcl}
				A_1 (k_x) & = & (a_1 \,, \ a_1' + k_x b_1) \\
				A_2 (k_x) & = & (\im (k_x \lambda + \lambda') a_1 - \bar{\alpha} (a_1' + k_x b_1) \,, \ \alpha a_1) \\
				(A_1 + A_2) (k_x) & = & ((1 +\im (k_x \lambda + \lambda')) a_1 - \bar{\alpha} (a_1' + k_x b_1) \,, \ \alpha a_1 + (a_1' + k_x b_1) ) \,,
			\end{array}
		\end{equation}
		whence
		\begin{equation}
			\begin{array}{lcl}
				\det A_1 (k_x) & = & \chi \\
				\det A_2 (k_x) & = & + | \alpha |^2 \chi \\
				\det (A_1 + A_2) (k_x) & = & ( 1 + \im (k_x \lambda + \lambda') + | \alpha |^2 ) \chi \,,
			\end{array}
		\end{equation}
		where
		\begin{equation}
			\chi \coloneqq a_1 \wedge (a_1' + k_x b_1) \,,
		\end{equation}
		for the scope of this proof. Moreover,
		\begin{equation}
			\op{adj} (A_2) A_1 - \op{adj} (A_1) A_2 = \chi
			\begin{pmatrix}
				- \im (k_x \lambda + \lambda') & - 2 \alpha \\
				2 \bar{\alpha} & \im (k_x \lambda + \lambda')
			\end{pmatrix} \,.
		\end{equation}
		Eq.\,\eqref{eq:UND} now follows by \eqref{eq:USmartFormula} and singling out $J$, cf.\,\eqref{eq:JSummand}, as a separate summand.
		
		\item[NN:] Within this family, Lm.\,\ref{lem:Dictionary} and self-adjointness, cf.\,Prop.\,\ref{prop:BCs}, produce the general form
		\begin{align}
			\ul{A} (k_x) &= (a_1, \, (\mu' + k_x \mu) a_1 + \im (\lambda_1' + k_x \lambda_1) a_2, \, \im (\lambda_2' + k_x \lambda_2) a_1 - (\bar{\mu}' + k_x \bar{\mu}) a_2, \, a_2 ) \nonumber \\
			&\equiv (A_1 (k_x), \, A_2 (k_x)) \label{eq:ACaseNN} \,.
		\end{align}
		It follows that
		\begin{equation}
			\begin{array}{lcl}
				A_1 (k_x) & = & (a_1 \,, \ (\mu' + k_x \mu) a_1 + \im (\lambda_1' + k_x \lambda_1) a_2) \\
				A_2 (k_x) & = & (\im (\lambda_2' + k_x \lambda_2) a_1 - (\bar{\mu}' + k_x \bar{\mu}) a_2 \,, \ a_2) \\
				(A_1 + A_2) (k_x) & = & 
				\begin{pmatrix}
					(1 +\im (\lambda_2' + k_x \lambda_2)) a_1 - (\bar{\mu}' + k_x \bar{\mu}) a_2 \\
					(\mu' + k_x \mu) a_1 + ( 1 + \im (\lambda_1' + k_x \lambda_1) ) a_2
				\end{pmatrix}^T \,,
			\end{array}
		\end{equation}
		whence
		\begin{equation}
			\begin{array}{lcl}
				\det A_1 (k_x) & = & \im (\lambda_1' + k_x \lambda_1) (a_1 \wedge a_2) \\
				\det A_2 (k_x) & = & \im (\lambda_2' + k_x \lambda_2) (a_1 \wedge a_2) \\
				\det (A_1 + A_2) (k_x) & = & \big[ ( 1 + \im (\lambda_1' + k_x \lambda_1) ) (1 +\im (\lambda_2' + k_x \lambda_2)) \\
				& & + (\mu' + k_x \mu) (\bar{\mu}' + k_x \bar{\mu}) \big] (a_1 \wedge a_2) \,.
			\end{array}
		\end{equation}
		Moreover,
		\begin{align}
			&\op{adj} (A_2) A_1 - \op{adj} (A_1) A_2 = (a_1 \wedge a_2) \times \nonumber \\
			& 
			\begin{pmatrix}
				2 (1 + \im (\lambda_1' + k_x \lambda_1) ) - \det (A_1 + A_2) & 2 (\mu' + k_x \mu) \\
				2 (\bar{\mu}' + k_x \bar{\mu}) & -2 (1 + \im (\lambda_2' + k_x \lambda_2) ) + \det (A_1 + A_2)
			\end{pmatrix} \,.
		\end{align}
		It then follows by \eqref{eq:USmartFormula} that
		\begin{equation}
			U (k_x) = J + \hat{U} (k_x) \,,
		\end{equation}
		cf.\,\eqref{eq:JSummand}, where
		\begin{align}
			\hat{U} (k_x) =& \frac{2}{( 1 + \im (\lambda_1' + k_x \lambda_1) ) (1 +\im (\lambda_2' + k_x \lambda_2)) + (\mu' + k_x \mu) (\bar{\mu}' + k_x \bar{\mu})} \nonumber \\
			&
			\begin{pmatrix}
				1 + \im (\lambda_1' + k_x \lambda_1) & \mu' + k_x \mu \\
				\bar{\mu}' + k_x \bar{\mu} & - 1 - \im (\lambda_2' + k_x \lambda_2)
			\end{pmatrix} \,.
		\end{align}
		Noticing that 
		\begin{equation}
			\hat{U} (k_x) = - \frac{2}{\det \tilde{U} (k_x)} \tilde{U} (k_x) \,, \label{eq:HatUNN}
		\end{equation}
		where
		\begin{equation}
			\tilde{U} (k_x) = 
			\begin{pmatrix}
				1 + \im (\lambda_1' + k_x \lambda_1) & \mu' + k_x \mu \\
				\bar{\mu}' + k_x \bar{\mu} & - 1 - \im (\lambda_2' + k_x \lambda_2)
			\end{pmatrix}
		\end{equation}
		finally yields \eqref{eq:UNN}.
	\end{itemize}
	(ii) The second statement of the proposition concerns particle-hole symmetry of the von Neumann unitaries obtained above. It is proven by imposing
	\begin{equation}
		U (k_x) = \overline{U (-k_x)} \label{eq:USymmetryThird}
	\end{equation}
	to those unitaries.
	\begin{itemize}[itemindent=1em]
		\item[DD:] By \eqref{eq:UCaseDD}, $ U(k_x) $ is constant and real-valued for all $ k_x $. Eq.\,\eqref{eq:USymmetryThird} is thus always satisfied, irrespective of the $ a_1',a_2',b_1,b_2 $ that appeared in the $ \ul{A} (k_x) $ we started with, cf.\,\eqref{eq:ACaseDD}.
		
		\item[ND:] It is immediate from \eqref{eq:UND} that \eqref{eq:USymmetryThird} is met for all $ k_x $ only if
		\begin{equation}
			\lambda' = 0 \,, \qquad \alpha \in \bbR \,.
		\end{equation}
		
		\item[NN:] Given that $ J = \overline{J} $ is itself PH-symmetric, condition \eqref{eq:USymmetryThird} translates to
		\begin{equation}
			\hat{U} (k_x) = \overline{\hat{U} (-k_x)} \,, \label{eq:ChiSymmetry}
		\end{equation}
		a.e.\,in $ k_x $. Given moreover the structure \eqref{eq:HatUNN} of $ \hat{U} (k_x) $, the condition for particle-hole symmetry further reduces to
		\begin{equation}
			\tilde{U} (k_x) = \overline{\tilde{U} (-k_x)} \,.
		\end{equation}
		This is explicitly to say
		\begin{align}
			\tilde{U} (k_x) &= 
			\begin{pmatrix}
				1 + \im (\lambda_1' + k_x \lambda_1) & \mu' + k_x \mu \\
				\bar{\mu}' + k_x \bar{\mu} & -1 -\im (\lambda_2' + k_x \lambda_2)
			\end{pmatrix} \nonumber \\
			&=
			\begin{pmatrix}
				1 - \im (\lambda_1' - k_x \lambda_1) & \bar{\mu}' - k_x \bar{\mu} \\
				\mu' - k_x \mu & -1 + \im (\lambda_2' - k_x \lambda_2)
			\end{pmatrix}
			= \overline{\tilde{U} (-k_x)} \,,
		\end{align}
		recalling $ \lambda_i, \lambda_i' \in \bbR \ (i = 1,2) $ by hypothesis. It is thus readily seen that
		\begin{equation}
			\lambda_1' = 0 \,, \qquad \lambda_2' = 0 \,, \qquad \mu \in \im \bbR \,, \qquad \ \mu' \in \bbR \,,
		\end{equation}
		as claimed.
	\end{itemize}
\end{proof}
\begin{proof}[Proof of Prop.\,\ref{prop:PHSBCs}]
	The claims of this proposition are simply obtained by imposing the conditions detailed at the end of Prop.\,\ref{prop:UBCs} on the parametrizations DD--NN of Prop.\,\ref{prop:BCs}.
\end{proof}

We close this appendix by studying when $ \ul{A} (k_x) $ is maximal rank. Let us introduce the property
\begin{equation}
	P(k_x): \qquad \text{rk} (A_1 (k_x), A_2 (k_x)) = 2
\end{equation}
and consider the statements
\begin{itemize}
	\item[(a)] $ P(k_x) $ holds true a.e. in $ k_x $;
	
	\item[(b)] $ P(k_x) $ holds true for all $ k_x $.
\end{itemize}
\begin{lemma}
	\label{lem:MaxRank}
	With reference to the four cases of Lemma \ref{lem:Dictionary}:
	\begin{itemize}[itemindent=1em]
		\item[DD:] (a) fails iff all of the following are met:
		\begin{equation}
			a_1' \wedge a_2' = 0 \,, \qquad a_1' \wedge b_2 + b_1 \wedge a_2' = 0 \,, \qquad b_1 \wedge b_2 = 0 \,,
		\end{equation}
		where the $\wedge$ product in $ \bbC^2 $ is defined as $ z \wedge w \coloneqq z_1 w_2 - w_1 z_2 $, for $ z = (z_1,z_2) \in \bbC^2 $ (same for $ w $).
		
		(b) holds true iff 
		\begin{equation}
			D(k_x) = (a_1' \wedge a_2') + (a_1' \wedge b_2 + b_1 \wedge a_2') k_x +  (b_1 \wedge b_2) k_x^2
		\end{equation}
		has no real zeros $ k_x $.
		
		\item[ND:] (a) fails iff
		\begin{equation}
			a_1 \wedge a_1' = 0 \,, \qquad a_1 \wedge b_1 = 0\,.
		\end{equation}
		
		(b) holds true iff
		\begin{equation}
			E(k_x) = a_1 \wedge (a_1' + k_x b_1) = (a_1 \wedge a_1') + (a_1 \wedge b_1) k_x \label{eq:SAAEND}
		\end{equation}
		has no real zeros $k_x$.
		
		\item[DN:] likewise ($1 \leftrightarrow 2$).
		
		\item[NN:] (a,b) both hold true.
	\end{itemize}
\end{lemma}
\begin{proof}
	By lower semicontinuity of the function $ k_x \mapsto \text{rk} (A_1 (k_x), A_2 (k_x)) $, the negation of (a) is \enquote{$ P(k_x) $ fails for all $ k_x $}. With this observation in mind, we compute the rank of
	\begin{equation}
		\ul{A}(k_x) = (a_1 \,, \ a_1' + k_x b_1 \,, \ a_2' + k_x b_2 \,, \ a_2)\,. \label{eq:ulAExplicit}
	\end{equation}
	\begin{itemize}[itemindent=1em]
		\item[DD:] In this case the rank is the same as that of $ (a_1' + k_x b_1, a_2' + k_x b_2) $. Thus $ P(k_x) $ holds true iff $ D(k_x) \neq 0 $, where
		\begin{align}
			D(k_x) &= (a_1' + k_x b_1) \wedge (a_2' + k_x b_2) \\
			&= (a_1' \wedge a_2') + (a_1' \wedge b_2 + b_1 \wedge a_2') k_x +  (b_1 \wedge b_2) k_x^2 \,.
		\end{align}
		By our first observation
		\begin{itemize}
			\label{obs:OnD}
			\item (a) fails iff $ D(k_x) \equiv 0 $;
			\item (b) holds iff $ D $ has no zeros.
		\end{itemize}
		The conclusions follow.
		
		\item[ND:] By Lm.\,\ref{lem:Dictionary} the vector
		\begin{align}
			a_2' + k_x b_2 &= (\im \lambda' a_1 - \bar{\alpha} a_1') + k_x (\im \lambda a_1 - \bar{\alpha} b_1) \\
			&= \im ( \lambda' + k_x \lambda) a_1 - \bar{\alpha} (a_1' + k_x b_1)
		\end{align}
		is a linear combination of two other vectors seen in \eqref{eq:ulAExplicit}. In view of $ a_2 = 0 $, the rank is that of $ (a_1, a_1' + k_x b_1) $. Thus $ P(k_x) $ holds true iff $ E(k_x) \neq 0 $, this time with
		\begin{equation}
			E(k_x) = a_1 \wedge (a_1' + k_x b_1) = (a_1 \wedge a_1') + (a_1 \wedge b_1) k_x \,.
		\end{equation}
		The conclusion then follows as in case (i).
		
		\item[DN:] By $ 1 \leftrightarrow 2 $.
		
		\item[NN:] By $ \text{rk} \ (a_1, a_2) = 2 $ the property $ P(k_x) $ applies to all $ k_x $.
	\end{itemize}
\end{proof}

\section{Infinite-momentum expansions of bulk eigensections and Jost functions} \label{app:ExpLambda}

Through the change of variables \eqref{eq:ChangeOfVariables}, we derive the infinite-momentum expansion of the outgoing and evanescent amplitudes $ \hat{\psi}^{0} (k_x, \kappa), \hat{\psi}^{\infty} (k_x, \kappa_{\op{ev}}) $, cf.\,\eqref{eq:ChoiceOfSections}, employed in the construction of the scattering state \eqref{eq:ExplicitScattState}. Such expansions are then used to prove Prop.\,\ref{prop:GEps}.
\vspace{1\baselineskip}

We start with the Taylor expansion of $ \hat{\psi}^{0}, \hat{\psi}^{\infty} $ around $ (k_x, k_y) \to \infty $. This limit is achieved as $ \varepsilon \to 0 $, given the coordinate change
\begin{equation}
	(k_x, k_y) = \left( \frac{\cos \varphi}{\varepsilon}, - \frac{\sin \varphi}{\varepsilon} \right) \,.
\end{equation}
As a preliminary step towards the desired expansions, the Taylor series of $ \omega_+ (k_x, \kappa) = \omega_+ (k_x, \kappa_{ev}) $, $ q_+ \equiv q (k_x, \pm \kappa) $ and $ q_- \equiv q(k_x, \kappa_{ev}) $ are found:
\begin{align} \label{eq:PrelExps}
	&\omega_+ (\varepsilon, \varphi) = \frac{\nu}{\varepsilon^2} + \frac{C \nu}{2} + \frac{1}{8 \nu^2} (4 \nu f - 1) \varepsilon^2 + o(\varepsilon^3) \,, \nonumber \\
	& q_+ (\varepsilon, \varphi) = -1 + \frac{\varepsilon^2}{2 \nu^2} + \frac{8 \nu f - 3}{8 \nu^4} \varepsilon^4 + o (\varepsilon^5) \,, \nonumber \\
	& q_- (\varepsilon, \varphi) = 1 + \frac{\varepsilon^2}{2 \nu^2} - \frac{\varepsilon^4}{8 \nu^4} + o(\varepsilon^5) \,,
\end{align}
where the last one has been determined with the help of 
\begin{align}
	\kappa_{\op{ev}} (\varepsilon, \varphi) \Big\vert_{\varepsilon = 0} = \im \left( \frac{\xi}{\varepsilon} + \frac{C}{2 \xi} \varepsilon - \frac{C^2}{8 \xi^3} \varepsilon^3 + \frac{C^3}{16 \xi^5} \varepsilon^5 - \frac{5 C^4}{128 \xi^7} \varepsilon^7 + o(\lambda^8) \right),
\end{align}
with $ C = (1-2 \nu f)/\nu^2 $ and $ \xi \equiv \xi (\varphi) = (1 + \sin^2 \varphi)^{1/2} $.

With this small tool kit, we find
\begin{align} 
	&\eta_{\op{o}} (\varepsilon, \varphi) \coloneqq \eta^0 (k_x (\varepsilon, \varphi), \kappa (\varepsilon, \varphi) ) = e^{\im \varphi} \left( \frac{\varepsilon}{\nu} - \frac{1 - 2 \nu f}{2 \nu^3} \varepsilon^3 + \frac{8 \nu^2 f^2 - 8 \nu f + 3}{8 \nu^5} \varepsilon^5 \right) + o (\varepsilon^6)  \nonumber \\
	&u_{\op{o}} (\varepsilon, \varphi) \coloneqq u^0 (k_x (\varepsilon, \varphi), \kappa (\varepsilon, \varphi) ) = 1 + \im \frac{\sin \varphi}{2 \nu^2} \varepsilon^2 + \im \sin \varphi \frac{8 \nu f - 3}{8 \nu^4} \varepsilon^4 + o (\varepsilon^5) \nonumber \\
	&v_{\op{o}} (\varepsilon, \varphi) \coloneqq v^0 (k_x (\varepsilon, \varphi), \kappa (\varepsilon, \varphi) ) = -\im \left( 1 + \frac{\sin \varphi}{2 \nu^2} \varepsilon^2 + \sin \varphi \frac{8 \nu f - 3}{8 \nu^4} \varepsilon^4 + o (\varepsilon^5) \right) \,, \label{eq:OutExp}
\end{align}
for the \textit{outgoing} amplitude, cf.\,\eqref{eq:DefEtaOEtaE}, as well as
\begin{align}
	\eta_{\op{e}} (\varepsilon, \varphi) &\coloneqq \eta^{\infty} ( k_x (\varepsilon, \varphi), \kappa_{\op{ev}} (\varepsilon, \varphi) ) \nonumber \\ 
	&= \frac{\chi}{\nu} \varepsilon - \frac{C}{2 \nu} \left( \frac{1}{\xi} +  \chi \right) \varepsilon^3 + \frac{1}{4 \nu} \left[ C^2 \chi - \frac{(4 \nu f - 1) \chi}{2 \nu^4} + \frac{C^2}{\xi} + \frac{C^2}{2 \xi^3} \right] \varepsilon^5 + o(\varepsilon^6) \nonumber \\
	u_{\op{e}} (\varepsilon, \varphi) &\coloneqq u^{\infty} ( k_x (\varepsilon, \varphi), \kappa_{\op{ev}} (\varepsilon, \varphi) ) = 1 + \frac{\xi}{2 \Gamma \nu^2} \varepsilon^2 + o( \varepsilon^3 ) \nonumber \\
	v_{\op{e}} (\varepsilon, \varphi) &\coloneqq v^{\infty} ( k_x (\varepsilon, \varphi), \kappa_{\op{ev}} (\varepsilon, \varphi) ) = \im \left( 1 + \frac{\cos \varphi}{2 \Gamma \nu^2} \varepsilon^2 + o(\varepsilon^3) \right) \,, \label{eq:EvExp}
\end{align}
with $ \chi \equiv \chi (\varphi) \coloneqq \sin \varphi - \xi (\varphi) $ and $ \Gamma \equiv \Gamma (\varphi) \coloneqq \sin \varphi + \xi (\varphi) $, for the \textit{evanescent} amplitude.

If needed, the \textit{incoming} amplitude $ l_{\op{i}} (\varepsilon, \varphi) \coloneqq l^0 (k_x (\varepsilon, \varphi), -\kappa (\varepsilon, \varphi) ) \ (l = \eta, u, v)$ can be obtained from the outgoing one by $ \varphi \leftrightarrow - \varphi $. Notice moreover that the $ u,v $ terms are constant in $ \varepsilon $, i.e., they are $ \mathcal{O} (1) $, whereas the $ \eta $ terms go to zero as $ \mathcal{O} (\varepsilon) $. This observation will be useful when computing the $ \varepsilon \to 0 $ expansion of the scattering amplitude $ S(\varepsilon,\varphi) $.

The expansions above are now employed to provide a proof of Prop.\,\ref{prop:GEps}.
\begin{proof}[Proof of Prop.\,\ref{prop:GEps}]
	Before specializing to the three families DD, ND, NN of interest, we bring $ \ul{A} V $, cf.\,\eqref{eq:gFromV}, in a form more suitable to expanding in $ \varepsilon \to 0 $.
	
	By Eq.\,\eqref{eq:A1A2}, we may rewrite
	\begin{equation}
		\ul{A} V = (A^0_1 + k_x B_1, A^0_2 + k_x B_2) 
		\begin{pmatrix}
			V_1 \\
			V_2
		\end{pmatrix}
		\eqqcolon \tilde{X} + k_x Y \,, \label{eq:APsiTilde}
	\end{equation}
	where
	\begin{equation}
		\begin{gathered}
			A^0_1 = | a_1 \rangle \langle 1 | + |a_1' \rangle \langle 2 | \,, \qquad A^0_2 = | a_2' \rangle \langle 1 | + |a_2 \rangle \langle 2 | \,, \qquad B_1 = | b_1 \rangle \langle 2 | \,, \qquad B_2 = | b_2 \rangle \langle 1 | \,, \\
			V_1 = | 1 \rangle \langle (\eta_{\op{o}} - \im \nu \kappa u_{\op{o}}, \, \eta_{\op{e}} - \im \nu \kappa_{\op{ev}} u_{\op{e}} ) | + | 2 \rangle \langle (u_{\op{o}}, \, u_{\op{e}}) | \,, \\
			V_2 = | 1 \rangle \langle (v_{\op{o}}, \, v_{\op{e}}) | + | 2 \rangle \langle (\im \nu \kappa v_{\op{o}}, \, \im \nu \kappa_{\op{ev}} v_{\op{e}}) | \,,
		\end{gathered}
	\end{equation}
	in Dirac notation,  with $ | 1 \rangle = (1,0)^{\op{T}}, \ | 2 \rangle = (0,1)^{\op{T}} $. As a consequence,
	\begin{equation}
		\begin{aligned}
			\tilde{X} \coloneqq A^0_1 V_1 + A^0_2 V_2 = &+ | a_1 \rangle \langle ( \eta_{\op{o}} - \im \nu \kappa u_{\op{o}}, \eta_{\op{e}} - \im \nu \kappa_{\op{ev}} u_{\op{e}}) | + | a_2 \rangle \langle (\im \nu \kappa v_{\op{o}}, \im \nu \kappa_{\op{ev}} v_{\op{e}}) | \\
			&+ | a_1' \rangle \langle (u_{\op{o}},u_{\op{e}}) | + | a_2' \rangle \langle (v_{\op{o}},v_{\op{e}}) | \,, \\
			Y \coloneqq B_1 V_1 + B_2 V_2 = &+ | b_1 \rangle \langle (u_{\op{o}},u_{\op{e}}) | + |b_2 \rangle \langle (v_{\op{o}},v_{\op{e}}) | \,. \label{eq:XTildeY}
		\end{aligned}
	\end{equation}
	As observed below \eqref{eq:EvExp}, $ u_j, v_j \ (j = \op{o}, \op{e})$ are $ \mathcal{O} (1) $ as $ \varepsilon \to 0 $, while $ \eta_j $ is $ \mathcal{O} (\varepsilon) $. Noticing moreover that $ k_x, \kappa, \kappa_{\op{ev}} = \mathcal{O} (\varepsilon) $, we conclude that $ \tilde{X} = \mathcal{O} (\varepsilon^{-1}) $, while $ Y = \mathcal{O} (1) $. We thus write their Taylor expansion as
	\begin{equation}
		\begin{aligned}
			\tilde{X} &= \tilde{X}_{-1} \varepsilon^{-1} + \tilde{X}_0 + \tilde{X}_1 \varepsilon + ... \\
			Y &= Y_0 + Y_1 \varepsilon + Y_2 \varepsilon^2 + ... \,, \label{eq:XTildeYExpansion}
		\end{aligned}
	\end{equation}
	whence
	\begin{equation}
		(\ul{A} V) (\varepsilon, \varphi) = \tilde{X} + \varepsilon^{-1} \cos \varphi Y \equiv M_{-1} (\varphi) \varepsilon^{-1} + M_0 (\varphi) + M_1 (\varphi) \varepsilon + ... \,,
	\end{equation}
	with
	\begin{equation}
		M_{-1} = \tilde{X}_{-1} + \cos \varphi Y_0 \,, \qquad M_0 = \tilde{X}_0 + \cos \varphi Y_1 \,, \label{eq:MMinusOneMZero}
	\end{equation}
	and so on. By inserting the expansions (\ref{eq:OutExp}, \ref{eq:EvExp}) of $ l_{\op{o}} \,, \ l_{\op{e}} \ (l = \eta, u, v) $ into (\ref{eq:XTildeYExpansion}, \ref{eq:MMinusOneMZero}), we obtain the leading and subleading orders of $ \ul{A} V $:
	\begin{align}
		M_{-1} = & \im \nu | a_1 + \im a_2 \rangle \langle 1| \sin \varphi + \nu | a_1 - \im a_2 \rangle \langle 2 | \sqrt{1 + \cos^2 \varphi} \nonumber \\
		&+ (|b_1 - \im b_2 \rangle \langle 1 | + |b_1 + \im b_2 \rangle \langle 2 | ) \cos \varphi \nonumber \\
		\equiv& - X \sin \varphi + \im X_{ev} \sqrt{1 + \cos^2 \varphi} +  Y \cos \varphi \,, \\
		M_0 = &+ | a_1' - \im a_2' \rangle \langle 1| + |a_1' + \im a_2' \rangle \langle 2 | \equiv X_0 \,.
	\end{align}
	The \textit{Jost function} $g$, cf.\,\ref{eq:gFromV}, hence reads
	\begin{align}
		g (\varepsilon, \varphi) = \op{det} (\ul{A} V) &= \varepsilon^{-2} \op{det} (M_{-1} + \varepsilon M_0 ) + \cal{O} (1) \nonumber \\
		&\equiv \varepsilon^{-2} \op{det} \chi (\varepsilon, \varphi) + \cal{O} (1) \,, \qquad \varepsilon \to 0 \,. \label{eq:ApproximateJost}
	\end{align}
	It is convenient to rewrite
	\begin{equation}
		\chi(\varepsilon, \varphi) =  X_0 \varepsilon - X \sin \varphi + \im X_{ev} \sqrt{1 + \cos^2 \varphi} + Y \cos \varphi \,, \label{eq:Chi}
	\end{equation}
	where
	\begin{equation}
		\label{eq:X0XXevY}
		\begin{array}{lcl}
			X_0 & \coloneqq & | a_1' - \im a_2' \rangle \langle 1 | + | a_1' + \im a_2' \rangle \langle 2 | \\
			X & \coloneqq & - \im \nu | a_1 + \im a_2 \rangle \langle 1 | \\
			X_{\op{ev}} & \coloneqq & - \im \nu |a_1 - \im a_2 \rangle \langle 2 | \\
			Y & \coloneqq & | b_1 - \im b_2 \rangle \langle 1 | + | b_1 + \im b_2 \rangle \langle 2 | \,. \\
		\end{array}
	\end{equation}
	In families DD--ND--NN, the first two orders in the expansion of $g$ are thus found by expressing $ X_0, X, X_{\op{ev}}, Y $ in the parametrization of Prop.\,\ref{prop:BCs}, and noticing that
	\begin{equation}
		\op{det} (|v_1 \rangle \langle w_1 | + |v_2 \rangle \langle w_2 | ) = \op{det} V \cdot \op{det} W \,, \qquad V = (v_1, v_2) \,, \hspace{6pt}
		W =
		\begin{pmatrix}
			w_1^T \\
			w_2^T
		\end{pmatrix}\,,
	\end{equation}
	where $ v_j, w_j $ $ (j=1,2) $ are column vectors in $ \bbC^2 $, so that in particular 
	\begin{equation}
		\op{det} (|v_1 \rangle \langle 1 | + |v_2 \rangle \langle 2 |) = \op{det} V = v_1 \wedge v_2 \,. \label{eq:DetTrick}
	\end{equation}
	The expressions for $ X_0, X, X_{\op{ev}}, Y $ in families DD, ND, NN are reported in the tables below.
	\begin{center}
		\begin{tabular}{ c"c|c}
			\xrowht[()]{10pt}
			 & $ X_0 $ & $X$ \\
			\thickhline\xrowht{10pt}
			DD & $ | a_1' - \im a_2' \rangle \langle 1 | + | a_1' + \im a_2' \rangle \langle 2 | $ & $ 0 $ \\
			\hline\xrowht{10pt}
			ND & $| (1 + \im \bar{\alpha}) a_1' + \lambda' a_1 \rangle \langle 1 | + | (1 - \im \bar{\alpha}) a_1' - \lambda' a_1 \rangle \langle 2 |$ & $ - \im \nu | (1 + \im \alpha) a_1 \rangle \langle 1 | $ \\
			\hline\xrowht{10pt}
			NN & $| (\mu' + \lambda_2') a_1 + \im (\lambda_1' + \bar{\mu}') a_2 \rangle \langle 1 | + | (\mu' - \lambda_2') a_1 + \im ( \lambda_1' - \bar{\mu}') a_2 \rangle \langle 2 |$ & $- \im \nu | a_1 + \im a_2 \rangle \langle 1 |$ \\
			\multicolumn{3}{c}{\vspace{1\baselineskip}} \\
			\xrowht[()]{10pt}
			 & $ Y $ & $ X_{\op{ev}} $ \\
			\thickhline\xrowht{10pt}
			DD & $ | b_1 - \im b_2 \rangle \langle 1 | + | b_1 + \im b_2 \rangle \langle 2 | $ & $ 0 $ \\
			\hline\xrowht{10pt}
			ND & $ | (1 + \im \bar{\alpha}) b_1 + \lambda a_1 \rangle \langle 1 | + | ( 1 - \im \bar{\alpha} ) b_1 - \lambda a_1 \rangle \langle 2 | $ & $ - \im \nu | (1 - \im \alpha) a_1 \rangle \langle 2 |$ \\
			\hline\xrowht{10pt}
			NN & $ | (\mu + \lambda_2) a_1 + \im ( \lambda_1 + \bar{\mu} ) a_2 \rangle \langle 1 | + | (\mu - \lambda_2) a_1 + \im ( \lambda_1 - \bar{\mu} ) a_2 \rangle \langle 2 | $ & $- \im \nu | a_1 - \im a_2 \rangle \langle 2 |$ \\
		\end{tabular}
		\label{tab:XZeroAndFriends}
	\end{center}
	The value of $ \chi (\varepsilon, \varphi) $ is now computed for each case DD, ND, NN. Its determinant yields the leading orders in the expansion of $g$, cf.\,\eqref{eq:ApproximateJost}.
	\begin{itemize}
		\item[DD:] By Tab.\,\ref{tab:XZeroAndFriends} and Eq.\,\eqref{eq:Chi}, we obtain
		\begin{align}
			\chi (\varepsilon, \varphi) = &| (a_1' \varepsilon + b_1 \cos \varphi) - \im (a_2' \varepsilon + b_2 \cos \varphi) \rangle \langle 1 | \nonumber \\
			&+ | (a_1' \varepsilon + b_1 \cos \varphi) + \im (a_2' \varepsilon + b_2 \cos \varphi) \rangle \langle 2 | \,.
		\end{align}
		The expression is already arranged so that its determinant immediately corresponds to the first row of \eqref{eq:Expansions}.
		
		\item[ND:] We similarly compute $ \chi $ and arrange it as
		\begin{align}
			\chi (\varepsilon, \varphi) =& | ( \im \nu (1 + \im \alpha) \sin \varphi + \lambda \cos \varphi + \lambda' \varepsilon) a_1 + (1 + i \bar{\alpha}) (\varepsilon a_1' + \cos \varphi b_1) \rangle \langle 1 | \nonumber \\
			&+ | ( \nu (1 - \im \alpha) \sin \varphi - \lambda \cos \varphi - \lambda' \varepsilon) a_1 + (1 - i \bar{\alpha}) (\varepsilon a_1' + \cos \varphi b_1) \rangle \langle 2 | \nonumber \\
			\equiv& | v_1 \rangle \langle 1 | + | v_2 \rangle \langle 2 | \,.
		\end{align} 
		By (\ref{eq:DetTrick}, \ref{eq:ApproximateJost}), then
		\begin{align}
			g (\varepsilon, \varphi) =& \varepsilon^{-2} \det \chi (\varepsilon, \varphi) + \cal{O} (1) \nonumber \\
			=& \varepsilon^{-2} a_1 \wedge (\varepsilon a_1' + \cos \varphi b_1) \big[ (1-\im \bar{\alpha}) ( \im \nu (1 + \im \alpha) \sin \varphi + \lambda \cos \varphi + \lambda' \varepsilon ) \nonumber \\
			&- (1 + \im \bar{\alpha}) ( \nu (1 - \im \alpha) \sqrt{1 + \cos^2 \varphi} - \lambda \cos \varphi - \lambda' \varepsilon ) ] + \cal{O} (1) \,,  
		\end{align}
		and the second row of \eqref{eq:Expansions} follows by simple algebra.
		
		\item[NN:] Using the notation of Eq.\,\eqref{eq:A1A2B1B2} and Tab.\,\ref{tab:XZeroAndFriends}, we obtain
		\begin{equation}
			\chi (\varepsilon, \varphi) = | A_1 (\varepsilon, \varphi) a_1 + A_2 (\varepsilon, \varphi) a_2 \rangle \langle 1 | + | B_1 (\varepsilon, \varphi) a_1 + B_2 (\varepsilon, \varphi) a_2 \rangle \langle 2 | \,.
		\end{equation}
		Then
		\begin{equation}
			\det \chi (\varepsilon, \varphi) = a_1 \wedge a_2 ( A_1 (\varepsilon, \varphi) B_2 (\varepsilon, \varphi) - A_2 (\varepsilon, \varphi) B_1 (\varepsilon, \varphi) ) \,,
		\end{equation}
		and the fourth row of \eqref{eq:Expansions} follows immediately from \eqref{eq:ApproximateJost}.
	\end{itemize}
\end{proof}

\section{Proof of Prop.\,\ref{prop:EscapeHeight}} \label{app:ProofEscapes}

\begin{proof}[Proof of Prop.\,\ref{prop:EscapeHeight}] \label{proof:EscapeHeight}
	Under the assumption $ |\omega| \ll k_x^2 $ at $ |k_x| \to \infty $, cf.\,\eqref{eq:RegionE}, we look for poles of the scattering matrix $S$ through zeros of the Jost function $g$, cf.\,\eqref{eq:SWithG}. Its computation, cf.\,\eqref{eq:gFromV}, relies on the two evanescent sections $ \hat{\psi}^0 (k_x, k_{y}) = \hat{\psi}_+ $, $ \hat{\psi}^\infty (k_x, k_{y}) =  \hat{\psi}_- $, cf.\,\eqref{eq:PsiInfPsiZero}, with $ k_y = k_{y+} $, $ k_y = k_{y-} $, respectively. For the scope of this proof, we rescale those sections by
	\begin{equation}
		\hat{\psi}_+ \mapsto (k_x + \im k_{y+}) \omega \hat{\psi}_+ \,, \qquad \hat{\psi}_- \mapsto (k_x - \im k_{y-}) \omega \hat{\psi}_- \,,
	\end{equation}
	which does not affect the zeros of $g$, cf.\,Rem.\,\ref{rem:ProportionalSections}. Thus, $ \hat{\psi}_\pm = (\eta_\pm, u_\pm, v_\pm) $ with
	\begin{equation}
		\label{eq:PMSections}
		\begin{aligned}
			\eta_\pm &= X_\pm \,, \\
			u_\pm &= k_x \omega - \im k_{y\pm} (f - \nu X_\pm) \,, \\
			v_\pm &= k_{y\pm} \omega + \im k_x (f - \nu X_\pm) \,.
		\end{aligned}
	\end{equation}
	The three components are functions of $ k_x, \omega $ by way of (\ref{eq:DefKYPlusMinus}, \ref{eq:BranchFreeOmega}), with the understanding that the evanescent branches of $ k_{y\pm} $ are chosen. Notice moreover that the sections vanish at $ \omega = \pm f $ if $ \mp k_x > 0 $, a fact that produces spurious zeros of $g$ that need to be ignored.
	
	In terms of our modified bulk sections, $ \ul{A} V $, cf.\,(\ref{eq:APsiTilde}, \ref{eq:XTildeY}), is rewritten as
	\begin{equation}
		\begin{aligned}
			\ul{A} V =& \tilde{X} +  k_x Y \\
			=& | a_1 \rangle \langle (\eta_+ - \im \nu k_{y+} u_+, \, \eta_- - \im \nu k_{y-} u_-) | + | a_2 \rangle \langle (\im \nu k_{y+} v_+, \, \im \nu k_{y-} v_-) | \\
			&+ | a_1' \rangle \langle (u_+,u_-) | + | a_2' \rangle \langle (v_+, v_-) | \\
			&+ k_x \big( | b_1 \rangle \langle (u_+, u_-) | + |b_2 \rangle \langle (v_+, v_-) | \big) \,. \label{eq:AVStupid}
		\end{aligned}
	\end{equation}
	Unlike in the proof of Prop.\,\ref{prop:GEps}, we refrain from making approximations at this stage, and rather calculate $ g = \det \ul{A} V $ in full. At first, we rearrange \eqref{eq:AVStupid} as
	\begin{equation}
		\begin{aligned}
			\ul{A} V =& | a_1 ( \eta_+ - \im \nu k_{y+} u_+ ) + a_2 (\im \nu k_{y+} v_+) + (a_1' + k_x b_1) u_+ + (a_2' + k_x b_2) v_+ \rangle \langle 1 | \\
			&+ | (+ \longleftrightarrow -) \rangle \langle 2 | \,, 
		\end{aligned}
	\end{equation}
	where the second summand is obtained from the first one by exchanging $+$ and $-$ subscripts. Then, we use \eqref{eq:DetTrick} to obtain
	\begin{align}
		\op{det} (\ul{A} V) =& (a_1 \wedge a_2) \Big( \im \nu ( k_{y-} \eta_+ v_- - k_{y+} \eta_- v_+ ) +  \nu^2 k_{y+} k_{y-} (u_+ v_- - u_- v_+) \Big) \nonumber \\
		&+ a_1 \wedge (a_1' + k_x b_1) \Big( (\eta_+ u_- - \eta_- u_+) - \im \nu u_+ u_- (k_{y+} - k_{y-}) \Big) \nonumber \\
		&+ a_1 \wedge (a_2' + k_x b_2) \Big( (\eta_+ v_- - \eta_- v_+) - \im \nu (k_{y+} u_+ v_- - k_{y-} u_- v_+) \Big) \nonumber \\
		&+ a_2 \wedge (a_1' + k_x b_1) \Big( \im \nu (k_{y+} u_- v_+ - k_{y-} u_+ v_-) \Big) \nonumber \\
		&+ a_2 \wedge (a_2' + k_x b_2) \Big( \im \nu v_+ v_- (k_{y+} - k_{y-}) \Big) \nonumber \\
		&+ (a_1' + k_x b_1) \wedge (a_2' + k_x b_2) \Big( u_+ v_- - u_- v_+ \Big) \label{eq:DetAVFull} \\
		\equiv& (a_1 \wedge a_2) \Big( \im \nu F_1 +  \nu^2 F_2 \Big) \nonumber + a_1 \wedge (a_1' + k_x b_1) \Big( F_3 - \im \nu F_4 \Big) \nonumber \\
		&+ a_1 \wedge (a_2' + k_x b_2) \Big( F_5 - \im \nu F_6 \Big) + a_2 \wedge (a_1' + k_x b_1) \Big( \im \nu F_7 \Big) \nonumber \\
		&+ a_2 \wedge (a_2' + k_x b_2) \Big( \im \nu F_8 \Big) + (a_1' + k_x b_1) \wedge (a_2' + k_x b_2) F_9 \,, \label{eq:DetAVShort}
	\end{align}
	where
	\begin{equation}
		\begin{array}{lcl}
			F_1 & = & k_{y-} \eta_+ v_- - k_{y+} \eta_- v_+ \,, \\
			F_2 & = & k_{y+} k_{y-} (u_+ v_- - u_- v_+) \,,  \\
			F_3 & = & \eta_+ u_- - \eta_- u_+  \,, \\
			F_4 & = & u_+ u_- (k_{y+} - k_{y-}) \,, \\
			F_5 & = & \eta_+ v_- - \eta_- v_+ \,, \\
			F_6 & = & k_{y+} u_+ v_- - k_{y-} u_- v_+ \,, \\
			F_7 & = & k_{y+} u_- v_+ - k_{y-} u_+ v_- \,, \\
			F_8 & = & v_+ v_- (k_{y+} - k_{y-}) \,, \\
			F_9 & = & u_+ v_- - u_- v_+ \,.
		\end{array}
	\end{equation}
	We expand the factors $ F_j \ (j=1,...,9) $ (one by one) around $ \delta = X/k_x^2 \to 0 $, retaining the leading order. Since the $ F_j $ are odd upon interchanging labels $ \pm $, we can extract a common factor $ (X_+ - X_-) $. As for the remaining factors, Vieta's formulae are used on $ X_+ + X_- $ and $ X_+ X_- $, reinstating explicit dependence on $\omega$. The resulting leading orders read
	\begin{equation}
		F_j = (X_+ - X_-) (f \pm \omega) \tilde{F}_j \,, \qquad \ (j = 1,...,9) \,, \label{eq:FPrefactor}
	\end{equation}
	where
	\begin{equation}
		\label{eq:FExpansion}
		\begin{array}{lcl}
			\tilde{F}_1 & = & - k_x |k_x| \big( 1 + \cal{O} (\delta) \big)  \,, \\
			\tilde{F}_2 & = & - \im  (2 \nu)^{-1} k_x |k_x| \big( \mp 1 + 2 \nu \omega + \cal{O} (\delta) \big) \,, \\
			\tilde{F}_3 & = & |k_x| \big( 1 + \cal{O} (\delta) \big) \,, \\
			\tilde{F}_4 & = & - \im (2 \nu)^{-1} |k_x| \big( 1 + \cal{O} (\delta) \big)  \,, \\
			\tilde{F}_5 & = & \pm \im |k_x|  \big( 1 + \cal{O} (\delta) \big) \,, \\
			\tilde{F}_6 & = & - \nu^{-1} |k_x| \big( \mp 1 + \nu \omega +  \cal{O} (\delta) \big) \,, \\
			\tilde{F}_7 & = & |k_x| \big( \omega + \cal{O} (\delta) \big) \,, \\
			\tilde{F}_8 & = & \im (2 \nu)^{-1} |k_x| \big( 1 + \cal{O} (\delta) \big) \,, \\
			\tilde{F}_9 & = & \pm \im (2 \nu)^{-1} \big( \mp 1 + 2 \nu \omega + \cal{O} (\delta) \big) \,. \\
		\end{array}
	\end{equation} 
	Above, the first (second) alternative refers to $ k_x \to + \infty $ ($ k_x \to - \infty $). Two comments are in order. First, the shared multiplicative factors $ (X_+ - X_-) $ and $ (f \pm \omega) $ do not contribute meaningful zeros to $g$. This is doubly true: On the one hand, they are shared with the numerator $ g(k_x, -k_{y+}) $ of $S$, and could thus be simplified in the spirit of \eqref{eq:Simplification}; on the other hand, $ X_+ \neq X_- $ for all $ \omega \neq 0 $, and $ \omega = \mp f $ corresponds to the identically zero edge state. Second, including the $ k_x $ prefactors, the leading order of \eqref{eq:DetAVFull} in $ k_x $ is $ \cal{O} (k_x^2) $. That is what we shall retain next, while specializing to families DD, ND, NN and searching zeros at $ |k_x| \to \infty $. \vspace{1\baselineskip}
	
	\textit{Specializing to family DD.} Here $ a_1 = a_2 = 0 $, cf.\,Lm.\,\ref{lem:Dictionary}. Eq.\,\eqref{eq:DetAVFull} reduces to
	\begin{equation}
		g = (a_1' + k_x b_1) \wedge (a_2' + k_x b_2) F_9 \,.
	\end{equation}
	The only meaningful solution of $ g=0 $ is given by $ \tilde{F}_9 = 0 $, namely
	\begin{equation}
		\omega_{\op{a}, \pm} = \pm \frac{1}{2 \nu} \,, \label{eq:AsymptoteDDProof}
	\end{equation}
	cf.\,\eqref{eq:FExpansion}. No solution with $ \omega \to 0, \infty $ is possible. Eq.\,\eqref{eq:AsymptoteDDProof} had already appeared in \cite{TDV20}. \vspace{1\baselineskip}
	
	\textit{Specializing to family ND.} Here $ a_2 = \alpha a_1 $ for some $ \alpha \in \bbC $, cf.\,Lm.\,\ref{lem:Dictionary}, and by Eq.\,\eqref{eq:ConditionsND}
	\begin{equation}
		(a_2' + k_x b_2) = - \bar{\alpha} (a_1' + k_x b_1) + \im (\lambda' + k_x \lambda) a_1 \,.
	\end{equation}
	Eq.\,\eqref{eq:DetAVFull} now becomes
	\begin{equation}
		g = a_1 \wedge (a_1' + k_x b_1) \Big( F_3 - \im \nu F_4 - \bar{\alpha} F_5 + \im \nu \bar{\alpha} F_6 + \im \nu \alpha F_7 - \im \nu |\alpha|^2 F_8 - \im k_x \lambda F_9 \Big) \,,
	\end{equation}
	having discarded subleading terms in $ k_x $. Inserting the leading orders of $ F_j \ (j=1,...,9) $, cf.\,(\ref{eq:FPrefactor}, \ref{eq:FExpansion}), we get
	\begin{equation}
		g = a_1 \wedge (a_1' + k_x b_1) (X_+ - X_-) (f \pm \omega) (2 \nu)^{-1} |k_x| \Big( \nu (1 + |\alpha|^2) - \lambda - 4 \nu^2 \alpha_I \omega \pm 2 \nu \lambda \omega \Big) 
	\end{equation}
	for $k_x \to \pm \infty$. If $ |\omega| \to \infty $, the constant terms in the big bracket are irrelevant, and the only zero of $ g $ is obtained by setting the coefficient of $ \omega $ to zero, namely
	\begin{equation}
		\lambda \mp 2 \nu \alpha_I = 0 \,.
	\end{equation}
	This equation defines a surface of codimension $1$ in the parameter space of boundary condition, which is of measure zero w.r.t.\,measures of Lebesgue class, as claimed.
	
	Similarly, if $ \omega \to 0 $, we find zeros of $g$ for
	\begin{equation}
		\lambda - \nu (1 + |\alpha|^2) = 0 \,.
	\end{equation}
	By contrast, generic zeros are found for $ \omega = \cal{O} (1) \to \omega_{\op{a}, \pm} $ at $ k_x \to \pm \infty $, and read
	\begin{equation}
		\omega_{\op{a}, \pm} = \pm \frac{\lambda - \nu (1 + | \alpha |^2)}{2 \nu ( \lambda \mp 2 \nu \alpha_I )} \,,
	\end{equation}
	as claimed. The particle-hole symmetric subfamily of ND is characterized by $ \alpha_I = 0 $, cf.\,Prop.\,\ref{prop:PHSBCs}. Notice how, in that case, $ \omega_{\op{a}, +} = - \omega_{\op{a},-} $, in agreement with the spectral symmetry $ (k_x, \omega) \leftrightarrow (-k_x, - \omega) $ induced by PHS. \vspace{1\baselineskip}
	
	\textit{Specializing to family NN.} This time 
	\begin{equation}
		\begin{aligned}
			(a_1' + k_x b_1) &= (\mu' + k_x \mu) a_1 + \im (\lambda_1' + k_x \lambda_1) a_2 \,, \\
			(a_2' + k_x b_2) &= -(\bar{\mu}' + k_x \bar{\mu}) a_2 + \im (\lambda_2' + k_x \lambda_2) a_1 \,,
		\end{aligned}
	\end{equation}
	and, upon discarding subleading terms, Eq.\,\eqref{eq:DetAVFull} goes to
	\begin{equation}
		\begin{aligned}
			g = (a_1 \wedge a_2) \Big( &\im \nu F_1 + \nu^2 F_2 \\
			+ &k_x \big( \im \lambda_1 (F_3 - \im \nu F_4) - \bar{\mu} (F_5 - \im \nu F_6) - \im \nu \mu F_7 + \nu \lambda_2 F_8 \big) \\
			+ &k_x^2 (\lambda_1 \lambda_2 - |\mu|^2) F_9 \Big) \,.
		\end{aligned}
	\end{equation}
	We now employ the expansions of $ F_j $ at $ k_x \to + \infty $ only. The result for $ k_x \to - \infty $ is deduced by symmetry. Those expansions and some algebra lead to
	\begin{equation}
		\begin{aligned}
			g = &(a_1 \wedge a_2) (X_+ - X_-) (f + \omega) (2 \nu)^{-1} k_x^2 \times \\
			&\Big( - \nu^2 + \nu (\lambda_1 + \lambda_2) + |\mu|^2 - \lambda_1 \lambda_2 + 2 \nu \omega (- \nu^2 - 2 \nu \mu_R + \lambda_1 \lambda_2 - |\mu|^2) \Big) \,,
		\end{aligned}
	\end{equation}
	whose only generic zeros are for asymptotically flat states $ \omega = \cal{O} (1) $ in $ k_x $, with asymptote
	\begin{equation}
		\omega_{\op{a},+} = \frac{|\mu|^2 - \nu^2 + \nu (\lambda_1 + \lambda_2) - \lambda_1 \lambda_2}{4 \nu^2 \mu_R + 2 \nu (|\mu|^2 + \nu^2 - \lambda_1 \lambda_2)}
	\end{equation}
	at $ k_x \to \pm \infty $. The $ | \omega| \to \infty, 0 $ cases are again exceptional, and realized when
	\begin{equation}
		- \nu^2 - 2 \nu \mu_R + \lambda_1 \lambda_2 - |\mu|^2 = 0 \ \longleftrightarrow \ \Delta^2 = \cal{I}_-
	\end{equation}
	or
	\begin{equation}
		- \nu^2 + \nu (\lambda_1 + \lambda_2) + |\mu|^2 - \lambda_1 \lambda_2 = 0 \ \longleftrightarrow \ \Delta^2 = \cal{E} \,,
	\end{equation}
	respectively.
	
	If $ \mu_R = 0 $, we sit in the particle-hole symmetric subfamily, and thus $ \omega_{0,-} \rvert_{\mu_R = 0} = - \omega_{0,+} \rvert_{\mu_R = 0} $. Since the symmetry is broken by the term in $ \mu_R $, the latter must appear with the same global sign in the two asymptotic regimes. All in all, we deduce
	\begin{equation}
		\omega_{\op{a},-} = - \frac{|\mu|^2 - \nu^2 + \nu (\lambda_1 + \lambda_2) - \lambda_1 \lambda_2}{-4 \nu^2 \mu_R + 2 \nu (|\mu|^2 + \nu^2 - \lambda_1 \lambda_2)} \,, \label{eq:AsymptoteNNMinusInfty}
	\end{equation}
	and thus
	\begin{equation}
		\omega_{\op{a}, \pm} = \frac{|\mu|^2 - \nu^2 + \nu (\lambda_1 + \lambda_2) - \lambda_1 \lambda_2}{4 \nu^2 \mu_R \pm 2 \nu (|\mu|^2 + \nu^2 - \lambda_1 \lambda_2)} \,,
	\end{equation}
	which is one of the claims. The zero for $ |\omega| \to \infty $ at $ k_x \to - \infty $ is moreover read by setting the denominator of \eqref{eq:AsymptoteNNMinusInfty} to zero,
	\begin{equation}
		- \nu^2 + 2 \nu \mu_R + \lambda_1 \lambda_2 - |\mu|^2 = 0 \ \longleftrightarrow \ \Delta^2 = \cal{I}_+ \,,
	\end{equation}
	whereas the case $ |\omega| \to 0 $ is left unchanged by $ k_x \leftrightarrow - k_x $ exchange. This completes the proof.
\end{proof}

\section{Proofs concerning boundary winding $B$} \label{app:ProofB}

This appendix contains the proofs of Lms.\,\ref{lem:CurveTraced}, \ref{lem:Avoidance}, \ref{lem:Winding} (proven together) and Prop.\,\ref{prop:BNN}. The latter will require a small preliminary lemma.

We start with the first three lemmas.
\begin{proof}[Proof of Lms.\,\ref{lem:CurveTraced}, \ref{lem:Avoidance}, \ref{lem:Winding}]
	We prove all three lemmas at once, in the order of cases (a-d).
	\begin{itemize}
		\item[(a)] Is clear on all counts I-III.
		
		\item[(b)] I is clear. II: In view of $ c_1 \neq 0 $, the equation $ P(k) = 0 $ is equivalent to $ |c_1|^2 k + c_2 \overline{c_1} = 0 $. The equation has no real solution $k$ iff $ d_{21} \equiv \op{Im} c_2 \overline{c_1} \neq 0 $. III: $ z = c_2 $ is a point of the curve ($k=0$), and $ c_1 $ is the tangent vector there. The orientation of the line is positive w.r.t.\,$ z=0 $ iff, when viewed as vectors in $ \bbR^2 $, i.e.\, $ c_i \equiv \vec{c}_i $, the pair $ (\vec{c}_2, \, \vec{c}_1) $ has positive wedge product. By
		\begin{equation}
			\vec{a} \wedge \vec{b} = a_1 b_2 - a_2 b_1 = - \op{Im} \big( (a_1 + \im a_2) (b_1 - \im b_2) \big) = - \op{Im} (a \overline{b}) \,,
		\end{equation}
		this is the case if $ \vec{c}_2 \wedge \vec{c}_1 = - \op{Im} c_2 \overline{c_1} = - d_{21} > 0 $.
		
		\item[(c)] Since $ c_0 \neq 0 $, the replacement of $ c_j $ by $ c_j' = c_j / c_0 = c_j \overline{c_0} | c_0 |^{-2}, \ (j = 0,1,2) $, can be considered, because it does not change any of the properties I-III, except for \eqref{eq:AvoidanceC} saying
		\begin{equation}
			4 c_2' > c_1'^2 \,.
		\end{equation}
		So, dropping primes, we have $ c_0 = 1 $ w.l.o.g., and $ d_{i0} = \op{Im} c_i \ (i=1,2) $. I: We have
		\begin{equation}
			\op{Im} P(k) = d_{20} \,, \label{eq:ImPCaseC}
		\end{equation}
		which shows that the trace of $P(k)$ is contained in a straight line. It is a half-line, because $ \op{Re} P(k) \to + \infty $, $ k \to \pm \infty $. II: By \eqref{eq:ImPCaseC}, the first alternative is clear; as for the second one, $ c_2 $ in \eqref{eq:PofK} is now also real, besides of $ c_0 = 1 $, $ c_1 $. The minimal value of $ P(k) $ is $ (4 c_2 - c_1^2)/4 $. III: is clear.
		
		\item[(d)] Again, $ c_0 = 1 $ w.l.o.g. I: We are in the generic case, i.e.\,the parabola, because the tangent vectors $ \dot{P} (k) = 2k + c_1 , \ (k \in \bbR) $ cannot be all parallel by $ \op{Im} c_1 \neq 0 $. II: The intersection of the parabola with the real axis occurs at $k$ such that $ \op{Im} P(k) = ( \op{Im} c_1 ) k + (\op{Im} c_2) = 0 $. There $ \op{Re} P (k) = k^2 + ( \op{Re} c_1 ) k + (\op{Re} c_2) $, and hence
		\begin{equation}
			\begin{aligned}
				(\op{Im} c_1)^2 (\op{Re} P(k)) &= (\op{Im} c_2)^2 - \big( (\op{Re} c_1) (\op{Im} c_2) - (\op{Im} c_1) (\op{Re} c_2) \big) (\op{Im} c_1) \\
				&= (\op{Im} c_2)^2 - (\op{Im} \overline{c_1} c_2) (\op{Im} c_1) = c (P) \,.
			\end{aligned}
		\end{equation} 
		So $ \op{Re} P(k) \gtrless 0 $ at the intersection point iff $ c(P) \gtrless 0 $.
		
		III: The parabola is horizontal, extends to the right, and $ \op{Im} \dot{P} (k) = \op{Im} c_1 $. Hence $ N(P) = 0 $ if $ c(P) > 0 $ and $ N(P) = - \op{sgn} (\op{Im} c_1) $ in the opposite case.
	\end{itemize}
\end{proof}
To specialize the general results of Lms.\,\ref{lem:CurveTraced}, \ref{lem:Avoidance}, \ref{lem:Winding} to family NN, we start from Eq.\,\eqref{eq:CoefficientsOfPpmNN} and recall that $ a_1, a_2 \in \bbC^2 $ are linearly independent by Lm.\,\ref{lem:Dictionary}, whence $ a_1 \wedge a_2 \neq 0 $. Moreover $ \lambda, \lambda' \in \bbR $ and $ \mu, \mu' \in \bbC $. We consider the variant $+$ only, and omit the subscript. The variant $-$ can be recovered by the replacement
\begin{gather}
	(a_1, a_2) \rightsquigarrow (-a_1, a_2) \,, \label{eq:FirstReplacement} \\
	(\lambda_1, \lambda_2) \rightsquigarrow (-\lambda_1, -\lambda_2) \,, \qquad (\lambda_1', \lambda_2') \rightsquigarrow (- \lambda_1', - \lambda_2') \,, \qquad (\mu, \mu') \rightsquigarrow (\mu, \mu') \,. \label{eq:SecondReplacement}
\end{gather}
By $ a_1 \wedge a_2 \neq 0 $, $ a_1 \wedge a_2 $ can moreover be replaced by $1$ for the purposes of Lms.\,\ref{lem:CurveTraced}, \ref{lem:Avoidance}, \ref{lem:Winding}, since all conditions are invariant under a common rescaling of the coefficients $ c_j, \ (j=0,1,2) $. In particular, the replacement \eqref{eq:FirstReplacement} can be omitted when passing to the $-$ variant. In summary, we have
\begin{align}
	c_0 &= | \mu|^2 - \lambda_1 \lambda_2 \,, \label{eq:c0NN} \\
	c_1 &= \im (\lambda_1 + \lambda_2) + ( \overline{\mu} \mu' + \mu \overline{\mu}' ) - (\lambda_2 \lambda_1' + \lambda_1 \lambda_2') \,, \label{eq:c1NN} \\
	c_2 &= \im (\lambda_1' + \lambda_2') + (1 - \lambda_1' \lambda_2' + |\mu'|^2) \,. \label{eq:c2NN}
\end{align}
Eqs.\,\eqref{eq:Dij} are moreover evaluated as
\begin{align}
	d_{10} &= (\lambda_1 + \lambda_2) (|\mu|^2 - \lambda_1 \lambda_2) \,, \label{eq:d10NN} \\
	d_{20} &= (\lambda_1' + \lambda_2') (|\mu|^2 - \lambda_1 \lambda_2) \,, \label{eq:d20NN} \\
	d_{21} &= -(\lambda_1 + \lambda_2) (1 - \lambda_1' \lambda_2' + |\mu'|^2) + (\lambda_1' + \lambda_2') (\overline{\mu} \mu' + \mu \overline{\mu}' - \lambda_2 \lambda_1' - \lambda_1 \lambda_2') \,. \label{eq:d21NN}
\end{align}
The following preliminary result is now needed for the proof of Prop.\,\ref{prop:BNN}.
\begin{lemma}
	Irrespective of case, we have (\ref{eq:d10NN}-\ref{eq:d21NN}) and
	\begin{gather}
		c(P) = (|\mu|^2 - \lambda_1 \lambda_2) \tilde{c} (P) \,, \label{eq:cofPFactorization} \\
		\tilde{c} (P) = (\lambda_1 + \lambda_2)^2 + (\lambda_1 \lambda_2' - \lambda_2 \lambda_1')^2 + | (\lambda_1' + \lambda_2') \mu - (\lambda_1 + \lambda_2) \mu' |^2 \,. \label{eq:cTildeofPNN}
	\end{gather}
\end{lemma}
\begin{proof}
	The expressions (\ref{eq:d10NN}, \ref{eq:d20NN}) for $ d_{10}, \, d_{20} $, cf.\,\eqref{eq:Dij}, are evident in view of $ c_1 $ real; that for $ d_{21} = - ( \op{Re} c_2 ) (\op{Im} c_1) + (\op{Im} c_2) (\op{Re} c_1) $ follows as well. We now turn to \eqref{eq:cofPFactorization}. A factor $ | \mu |^2 - \lambda_1 \lambda_2 $ can be extracted from both $ d_{10}, \, d_{20} $, and so from $ c(P) = d_{20}^2 - d_{21} d_{10} $, cf.\,\eqref{eq:cofP}. That yields \eqref{eq:cofPFactorization} with
	\begin{equation}
		\begin{aligned}
			\tilde{c} (P) =& d_{20} (\lambda_1' + \lambda_2') - d_{21} (\lambda_1 + \lambda_2) \\
			=& (\lambda_1' + \lambda_2')^2 (|\mu|^2 - \lambda_1 \lambda_2) \\
			&- (\lambda_1 + \lambda_2) (\lambda_1' + \lambda_2') ( \overline{\mu} \mu' + \mu \overline{\mu}' - \lambda_2 \lambda_1' - \lambda_1 \lambda_2' ) \\
			&+ (\lambda_1 + \lambda_2)^2 (1 - \lambda_1' \lambda_2' + | \mu'|^2) \,.
		\end{aligned}
	\end{equation}
	The term $1$ in the last bracket contributes $ ( \lambda_1 + \lambda_2 )^2 $ to $ \tilde{c} (P) $; all other terms not containing $ \mu, \mu' $ contribute (likewise)
	\begin{equation}
		\begin{aligned}
			&- (\lambda_1'^2 + 2 \lambda_1' \lambda_2' + \lambda_2'^2) \lambda_1 \lambda_2 \\
			&+ (\lambda_1 \lambda_1' + \lambda_1 \lambda_2' + \lambda_2 \lambda_1' + \lambda_2 \lambda_2') (\lambda_2 \lambda_1' + \lambda_1 \lambda_2') \\
			&- (\lambda_1^2 + 2 \lambda_1 \lambda_2 + \lambda_2^2) \lambda_1' \lambda_2' \\
			=& \lambda_1^2 \lambda_2'^2 + \lambda_2^2 \lambda_1'^2 - 2 \lambda_1 \lambda_2 \lambda_1' \lambda_2' \\
			=& (\lambda_1 \lambda_2' - \lambda_2 \lambda_1')^2 \,.
		\end{aligned}
	\end{equation}
	The terms containing $ \mu $ or $ \mu' $ are
	\begin{equation}
		\begin{aligned}
			& (\lambda_1' + \lambda_2')^2 |\mu|^2 - (\lambda_1 + \lambda_2) (\lambda_1' + \lambda_2') (\overline{\mu} \mu' + \mu \overline{\mu}') + (\lambda_1 + \lambda_2)^2 | \mu'|^2 \\
			=& | (\lambda_1' + \lambda_2') \mu - (\lambda_1 + \lambda_2) \mu' |^2 \,.
		\end{aligned}
	\end{equation}
	That proves \eqref{eq:cTildeofPNN}.
\end{proof}
\begin{proof}[Proof of Prop.\,\ref{prop:BNN}]
	Quite generally, $ c_0 = 0 $ and $ d_{10} = 0 $ translate to
	\begin{gather}
		|\mu|^2 = \lambda_1 \lambda_2 \,, \label{eq:MuSquare} \\
		(\lambda_1 + \lambda_2) ( | \mu |^2 - \lambda_1 \lambda_2 ) = 0 \,, \label{eq:d10EqualZeroNN}
	\end{gather}
	respectively. We shall prove the claims in order (a-d).
	\begin{itemize}
		\item[(a)] I: The stated conditions on $ \mu, \lambda_1, \lambda_2 $ (i.e., $=0$), clearly imply $ c_0 = 0 $, $ c_1 = 0 $. Conversely, we have \eqref{eq:MuSquare} and $ \op{Im} c_1 = 0 $, i.e.\,$ \lambda_1 + \lambda_2 = 0 $. Thus, $ | \mu |^2 = - \lambda_1^2 $, implying $ \mu = \lambda_1 = \lambda_2 = 0 $, as claimed. II: We have to show that $ c_2 = 0 $ is impossible. Indeed, that would imply $ \op{Im} c_2 = 0 $, i.e.\,$ \lambda_1' + \lambda_2' = 0 $, and then $ \op{Re} c_2 = 1 + \lambda_1'^2 + | \mu' |^2 > 0 $, which is the sought contradiction. III: Nothing to prove.
		
		\item[(b)] I: Lm.\,\ref{lem:CurveTraced} and (\ref{eq:c0NN}, \ref{eq:c1NN}) imply \textbullet\ $ | \mu |^2 = \lambda_1 \lambda_2 $ and \textbullet\ $ \lambda_1 + \lambda_2 \neq 0 $ or $ \lambda_2 \lambda_1' + \lambda_1 \lambda_2' \neq \overline{\mu} \mu' + \mu \overline{\mu}' $. However, $ \lambda_1 + \lambda_2 = 0 $ is not actually possible, because it would imply $ \lambda_2 = - \lambda_1 $ and in turn $ | \mu |^2 = - (\lambda_1)^2 $, which is only true when $ \mu = \lambda_1 = \lambda_2 $ and the curve traced is a point, cf.\,(a) I. Dropping the condition $ \lambda_1 + \lambda_2 \neq 0 $ yields the stated result. II: Lm.\,\ref{lem:Avoidance} states that the origin is avoided iff $ d_{21} \neq 0 $. The latter in turn reads
		\begin{equation}
			d_{21} = -(\lambda_1 + \lambda_2) (1 - \lambda_1' \lambda_2' + |\mu'|^2) + (\lambda_1' + \lambda_2') (\overline{\mu} \mu' + \mu \overline{\mu}' - \lambda_2 \lambda_1' - \lambda_1 \lambda_2') \,,
		\end{equation}
		cf.\,\eqref{eq:d21NN}. The shifts $ A_0 \mapsto A_0 + \im \tau A_x $, $ (\tau \in \bbR) $ of Rem.\,\ref{rem:TranslationsKx}, corresponding to reparametrizations of $ k_x \in \bbR $ that leave $ H^\# $ invariant, read
		\begin{equation}
			a_j' \mapsto a_j' + \tau b_j \,, \quad (j=1,2)
		\end{equation}
		in the parametrization \eqref{eq:ExplicitUlA}. In family NN, they in turn induce the following ones
		\begin{equation}
			\begin{aligned}
				\mu' &\mapsto \mu' + \tau \mu \,, \\
				\lambda_1' &\mapsto \lambda_1' + \tau \lambda_1 \,, \\
				\lambda_2' &\mapsto \lambda_2' + \tau \lambda_2 \,.
			\end{aligned}
		\end{equation}
		As it should, $ d_{21} $ is left invariant by those shifts (proof by direct substitution, making use of the hypothesis $ |\mu|^2 = \lambda_1 \lambda_2 $ from point I). We can then exploit the further hypothesis $ (\lambda_1 + \lambda_2) \neq 0 $, cf.\,point I, and w.l.o.g.\,pick a shift $ \tau $ such that
		\begin{equation}
			\lambda_1' + \lambda_2' = 0 \ \leftrightarrow \ \lambda_2' = - \lambda_1' \,.
		\end{equation}
		Then,
		\begin{equation}
			d_{21} = -(\lambda_1 + \lambda_2) (1 + (\lambda_1')^2 + |\mu'|^2) \neq 0 \,,
		\end{equation}
		i.e., the origin is always avoided, as claimed. III: Nothing to prove.

		\item[(c)] I: $ c_0 \neq 0 $ is equivalent to $ |\mu|^2 \neq \lambda_1 \lambda_2 $, at which point $ d_{10} = 0 $ states $ \lambda_1 + \lambda_2 = 0 $ by \eqref{eq:d10NN}, implying $ | \mu |^2 \neq - \lambda_1^2 $ and thus $ \mu, \lambda_1 $ not both zero. II: The first subcase in Lm.\,\ref{lem:Avoidance} (c) now says $ \lambda_1' + \lambda_2' \neq 0 $. In the complementary one, i.e.\,$ \lambda_1' + \lambda_2' = 0 $, we claim that \eqref{eq:AvoidanceC} holds true always, which would complete this part of the proof. Indeed, $ c_2 \overline{c_0} $ and $ c_1 \overline{c_0} $ are now real, cf.\,Rem.\,below \eqref{eq:AvoidanceC}, with
		\begin{equation}
			\begin{gathered}
				c_0 = | \mu |^2 + \lambda_1^2 \,, \\
				c_1 \overline{c_0} = ( 2 \lambda_1 \lambda_1' + \overline{\mu} \mu' + \mu \overline{\mu}' ) (|\mu|^2 + \lambda_1^2) \,, \\
				c_2 \overline{c_0} = ( 1 + \lambda_1'^2 + | \mu'|^2 ) (| \mu|^2 + \lambda_1^2) \,,
			\end{gathered}
		\end{equation}
		so that \eqref{eq:AvoidanceC} reduces to
		\begin{equation}
			4 (1 + \lambda_1'^2 + | \mu'|^2) (|\mu|^2 + \lambda_1^2) > ( 2 \lambda_1 \lambda_1' + \overline{\mu} \mu' + \mu \overline{\mu}' )^2 \,.
		\end{equation}
		It holds true since the l.h.s.\,is bounded below by
		\begin{equation}
			4 ( \lambda_1'^2 + | \mu' |^2 ) (\lambda_1^2 + | \mu |^2)
		\end{equation}
		and the r.h.s.\,is bounded above by
		\begin{equation}
			4 ( |\lambda_1| |\lambda_1'| + |\mu| |\mu'| )^2
		\end{equation}
		and the comparison is completed by the Cauchy-Schwarz inequality. III: Nothing to prove.
		
		\item[(d)] I: The conditions of Lm.\,\ref{lem:CurveTraced} (d) translate to $ | \mu|^2 \neq \lambda_1 \lambda_2 $, $ \lambda_1 + \lambda_2 \neq 0 $, cf.\,(\ref{eq:c0NN}, \ref{eq:d10NN}). II: We have to rule out $ c(P) = 0 $. Indeed, in that case, $ \tilde{c} (P) = 0 $ by (\ref{eq:cofPFactorization}, \ref{eq:cTildeofPNN}), and in turn $ \lambda_1 + \lambda_2 = 0 $, which is however ruled out by part I. III: We have $ c(P) \gtrless 0 $ if $ | \mu |^2 \gtrless \lambda_1 \lambda_2 $ by \eqref{eq:cofP} and $ \tilde{c} (P) > 0 $, as well as $ - \op{sgn} d_{10} = - \op{sgn} (\lambda_1 + \lambda_2) \cdot \op{sgn} (|\mu|^2 - \lambda_1 \lambda_2) $ by \eqref{eq:d10NN}. Eq.\,\eqref{eq:NofPNN} now follows from \eqref{eq:WindingParabola}.
	\end{itemize}
\end{proof}

\bibliography{bib}

\end{document}